\documentclass[USletter,10pt]{article}
\usepackage[margin=1in,dvips]{geometry}
\usepackage{graphicx,cancel}
\usepackage{graphics}
\usepackage[dvipsnames]{xcolor}
\usepackage{soul}
\usepackage{amsmath, amsthm, slashed}
\usepackage{amssymb}
\usepackage{rotating}
\usepackage{mathrsfs}
\usepackage{upgreek}  
\usepackage{epsfig}
\usepackage{verbatim}
\usepackage[affil-it]{authblk}
\usepackage{rotating}
\usepackage{graphicx}
\usepackage{accents, enumerate}

\usepackage[colorlinks,linkcolor=MidnightBlue,citecolor=MidnightBlue]{hyperref}
\usepackage{float}
\usepackage{tikz-cd} 
\usepackage{pgfplots}
\usetikzlibrary{decorations.pathmorphing}
\usetikzlibrary{decorations.markings}
\usetikzlibrary{patterns}
\usetikzlibrary{shapes}
\usetikzlibrary{plotmarks}

\newtheorem{theorem}{Theorem}[section]
\newtheorem{definition}[theorem]{Definition}

\newtheorem{corollary}[theorem]{Corollary}
\newtheorem{proposition}[theorem]{Proposition}
\newtheorem{lemma}[theorem]{Lemma}
\newtheorem{remark}[theorem]{Remark}

\newtheorem*{conjecture*}{Conjecture}
\newtheorem*{theorem*}{Theorem}
\newtheorem*{corollary*}{Corollary}
\newcommand{\Sf}{\accentset{\land}{\mathscr{S}}}
\newcommand{\Sff}{{\accentset{\land}{\mathscr{S}}^\prime}}
\newcommand{\Si}{\underaccent{\lor}{\mathscr{S}}}

\newcommand{\ns}{\slashed{\nabla}}
\newcommand{\divs}{\slashed{\mathrm{div}}}
\newcommand{\ds}{\slashed{\Delta}}

\newcommand{\Olin}{\Omega^{-1}\accentset{\scalebox{.6}{\mbox{\tiny (1)}}}{\Omega}}
\newcommand{\Olino}{\accentset{\scalebox{.6}{\mbox{\tiny (1)}}}{\Omega}}
\newcommand{\glinh}{\accentset{\scalebox{.6}{\mbox{\tiny (1)}}}{\hat{\slashed{g}}}}
\newcommand{\glin}{\accentset{\scalebox{.6}{\mbox{\tiny (1)}}}{\slashed{g}}}
  \newcommand{\glinto}{\accentset{\scalebox{.6}{\mbox{\tiny (1)}}}{\sqrt{\slashed{g}}}}
\newcommand{\bmlin}{\accentset{\scalebox{.6}{\mbox{\tiny (1)}}}{b}}
\newcommand{\Pmcalin}{\accentset{\scalebox{.6}{\mbox{\tiny (1)}}}{\mathcal{P}}}
\newcommand{\Qlin}{\accentset{\scalebox{.6}{\mbox{\tiny (1)}}}{\mathcal{Q}}}

\newcommand{\xblin}{\accentset{\scalebox{.6}{\mbox{\tiny (1)}}}{\underline{\hat{\chi}}}}
\newcommand{\xlin}{\accentset{\scalebox{.6}{\mbox{\tiny (1)}}}{{\hat{\chi}}}}

\newcommand{\eblin}{\accentset{\scalebox{.6}{\mbox{\tiny (1)}}}{\underline{\eta}}}
\newcommand{\elin}{\accentset{\scalebox{.6}{\mbox{\tiny (1)}}}{{\eta}}}
\newcommand{\otx}{\accentset{\scalebox{.6}{\mbox{\tiny (1)}}}{\left(\Omega \mathrm{tr} \chi\right)}}
\newcommand{\otxb}{\accentset{\scalebox{.6}{\mbox{\tiny (1)}}}{\left(\Omega \mathrm{tr} \underline{\chi}\right)}}
\newcommand{\olin}{\accentset{\scalebox{.6}{\mbox{\tiny (1)}}}{\omega}}
\newcommand{\olinb}{\accentset{\scalebox{.6}{\mbox{\tiny (1)}}}{\underline{\omega}}}

\newcommand{\ablin}{\accentset{\scalebox{.6}{\mbox{\tiny (1)}}}{\underline{\alpha}}}
\newcommand{\alin}{\accentset{\scalebox{.6}{\mbox{\tiny (1)}}}{{\alpha}}}
\newcommand{\pblin}{\accentset{\scalebox{.6}{\mbox{\tiny (1)}}}{\underline{\psi}}}
\newcommand{\plin}{\accentset{\scalebox{.6}{\mbox{\tiny (1)}}}{{\psi}}}

\newcommand{\bblin}{\accentset{\scalebox{.6}{\mbox{\tiny (1)}}}{\underline{\beta}}}
\newcommand{\blin}{\accentset{\scalebox{.6}{\mbox{\tiny (1)}}}{{\beta}}}
\newcommand{\rlin}{\accentset{\scalebox{.6}{\mbox{\tiny (1)}}}{\rho}}
\newcommand{\slin}{\accentset{\scalebox{.6}{\mbox{\tiny (1)}}}{{\sigma}}}
\newcommand{\Klin}{\accentset{\scalebox{.6}{\mbox{\tiny (1)}}}{K}}

\newcommand{\Psilin}{\accentset{\scalebox{.6}{\mbox{\tiny (1)}}}{\Psi}}
\newcommand{\Psilinb}{\accentset{\scalebox{.6}{\mbox{\tiny (1)}}}{\underline{\Psi}}}

\DeclareMathAlphabet\mathbfcal{OMS}{cmsy}{b}{n}

\title{Linear Stability of Schwarzschild-Anti-de Sitter spacetimes I: \\ The system of gravitational perturbations}

\author[1]{Olivier Graf\thanks{olivier.graf@univ-grenoble-alpes.fr}}
\author[2,3]{Gustav Holzegel\thanks{gholzegel@uni-muenster.de}}

\affil[1]{\small Univ.~Grenoble~Alpes, CNRS, IF, 38000 Grenoble, France \vskip.2pc \ }
\affil[2]{\small Universit\"at M\"unster,
Mathematisches~Institut, Einsteinstrasse~62~48149~M\"unster,~Bundesrepublik~Deutschland \vskip.2pc \ }
\affil[3]{\small Imperial College London,
Department of Mathematics,
South~Kensington~Campus,~London~SW7~2AZ,~United~Kingdom}
\begin{document}
\maketitle
\begin{abstract}
This is the main paper of a series establishing the linear stability of Schwarzschild-Anti-de Sitter (AdS) black holes to gravitational perturbations. Specifically, we prove that solutions to the linearisation of the Einstein equations $\textrm{Ric}(g) = \Lambda g$ with $\Lambda<0$ around a Schwarzschild-AdS metric arising from regular initial data and with standard Dirichlet-type boundary conditions imposed at the conformal boundary (inherited from fixing the conformal class of the non-linear metric) remain globally uniformly bounded on the black hole exterior and in fact decay inverse logarithmically in time to a linearised Kerr-AdS metric. The proof exploits a hierarchical structure of the equations of linearised gravity in  double null gauge and crucially relies on boundedness and logarithmic decay results for the Teukolsky system, which are independent results proven in Part II of the series. Contrary to the asymptotically flat case, addition of a residual pure gauge solution to the original solution is not required to prove decay of all linearised null curvature and Ricci coefficients. One may however normalise the solution at the conformal boundary to be in standard AdS-form by adding such a pure gauge solution, which is constructed dynamically from the trace of the original solution at the conformal boundary and quantitatively controlled by initial data.
\end{abstract}

{
  \hypersetup{linkcolor=black}
  \tableofcontents
}

\hypersetup{linkcolor=MidnightBlue}

\section{Introduction}\label{sec:intro}

The study of the stability of black hole solutions of the Einstein equations with cosmological constant $\Lambda$,
\begin{align} \label{EVEL}
  \mathrm{Ric}(g) = \Lambda g \, ,
\end{align}
originates in the physics literature with the pioneering work of Regge and Wheeler \cite{Reg.Whe57} for the Schwarzschild solution. In the past two decades the subject has seen a tremendous development through the introduction of modern PDE theory. As of today, satisfactory \emph{non-linear} stability results are available for Schwarzschild and very slowly rotating Kerr black holes, i.e.~$\Lambda=0$ in (\ref{EVEL}), and Kerr-de Sitter black holes, i.e.~$\Lambda>0$ in (\ref{EVEL}); see \cite{Daf.Hol.Rod.Tay21, Gio.Kla.Sze22, Hin.Vas16} and references therein. On the other hand, still very little is known about the non-linear evolution of perturbations of the Schwarzschild-Anti de Sitter, and more generally Kerr-Anti de Sitter, family of solutions, i.e.~$\Lambda<0$ in (\ref{EVEL}). 

The main difficulty in the analysis of asymptotically Anti-de Sitter (aAdS) spacetimes lies in their non-globally hyperbolic nature: The spacetimes possess a timelike conformal boundary at infinity, which is most easily seen for the maximally symmetric solution of (\ref{EVEL}) with $\Lambda<0$, Anti-de Sitter (AdS) space \cite{Haw.Ell08}. The existence of the boundary necessitates the study of a boundary initial value problem to understand the dynamics of the hyperbolic system (\ref{EVEL}). Formulating geometric boundary conditions to establish local well-posedness is highly non-trivial. See \cite{Fri95, Enc.Kam19} for some classical well-posedness theorems for (\ref{EVEL}) with $\Lambda<0$. 

Turning to the \emph{global} dynamics in the $\Lambda<0$ case, stability properties of stationary solutions are expected to depend crucially on the type of boundary condition imposed. Notably, here even the simplest case -- that of perturbations around pure AdS -- is still open. In the case that \emph{Dirichlet} type boundary conditions are imposed at the boundary, non-linear \emph{instability} of AdS has been proven for spherically symmetric toy-models \cite{Mos18, Mos20} and is expected to hold in general.\footnote{Dirichlet conditions can be thought of as a form of reflecting boundary conditions. In particular, the asymptotic mass is held constant along the conformal boundary, hence gravitational radiation cannot escape through the conformal boundary. Geometrically, it corresponds to fixing the conformal class of the metric on the boundary to be that induced by the pure AdS metric. See \cite{Biz.Ros11} for very influential numerical study of the Dirichlet problem is the context of the spherically symmetric scalar field model.} On the other hand, the (linear) results of \cite{Hol.Luk.Smu.War20} suggest non-linear stability to hold in the case of \emph{dissipative} boundary conditions, where radiation is allowed to escape through the conformal boundary. Going from pure AdS to black hole spacetimes, the additional characteristic phenomena of trapped null geodesics and superradiance couple with the effect of the boundary making the analysis of the problem even more difficult. Nevertheless, the general expectation is again that of instability for reflecting boundary conditions and stability for dissipative boundary conditions.

\subsection{The scalar wave equation, linear stability and non-linear stability}
From a PDE perspective, non-linear stability results typically rely on a robust understanding of the underlying linearised problem including quantitative estimates on the rates of decay in the geometry under consideration. Such an analysis also seems a prerequisite for non-linear \emph{instability} results in order to gain control on potential blow-up or growth mechanisms. Estimating the linearisation of the equations (\ref{EVEL}) requires choosing a gauge and, independently of the specifics of the gauge, already results in a complicated coupled system of equations. A good first intuition can often be gained from the study of the scalar wave equation $\Box_{g} \psi = 0$ on the background under consideration. This removes the problem of gauge as well as the coupled tensorial character of the problem. At the same time, the scalar equation (due to its Lagrangian structure) inherits natural coercive conservation laws from the symmetries of the background which can be exploited in the analysis.

In the case of pure AdS, one thus discovers the existence of time-periodic (i.e.~non-decaying) solutions of $\Box_{g_{AdS}} \psi =0$, which lie at the heart of the non-linear instability exploited in \cite{Mos18}. In the case of asymptotically AdS black holes, the corresponding scalar problem was studied in \cite{Hol.Smu13, Hol.Smu14}, where it is shown that solutions of $\Box_{g_{KAdS} \psi} + \mu \psi= 0$ with Dirichlet boundary conditions decay inverse logarithmically and not faster on Kerr-AdS black holes whose parameters satisfy the Hawking-Reall bound.\footnote{Beyond that bound, one has exponentially growing solutions. See \cite{Dol17}.} In view of the slow decay, the authors of \cite{Hol.Smu14} conjectured non-linear instability of these black holes. Recently, a concrete instability mechanism (related to weak turbulence and the growth of higher order Sobolev norms) has been suggested for a non-linear scalar toy-model on Schwarzschild-AdS \cite{Keh.Mos}. 

\subsection{The main result}
The goal of this series of works is to show that the logarithmic decay established for the scalar toy problem also holds for the linearised Einstein equations on Schwarzschild-AdS. More precisely we will prove the following statement:
\begin{theorem*}[Informal version]
Solutions to the linearisation of the Einstein
equations $Ric=-3k^2 g$ around a Schwarzschild-AdS metric arising from regular initial data and with
standard Dirichlet boundary conditions at the conformal boundary (inherited from fixing the conformal class of the non-linear metric) remain globally uniformly bounded on the black hole exterior and in fact decay inverse logarithmically to a linearised Kerr-AdS metric.
\end{theorem*}

For a precise statement of the theorem, see already Theorem \ref{mtheo:boundedness} below. 

\begin{remark}
The theorem should be directly compared with the result of \cite{Daf.Hol.Rod19} in the $\Lambda=0$ case. The main difference is that here only a logarithmic decay rate (as opposed to inverse polynomial in  \cite{Daf.Hol.Rod19}) can be concluded. This is characteristic of the reflective boundary conditions as explained above. In Part III of the series, we actually prove that the decay rate cannot be improved for general solutions. The other main difference compared with the asymptotically flat case is that here all quantities can be shown to decay \emph{without} adding a residual pure gauge solution. This is to be constrasted with Theorem 3 in  \cite{Daf.Hol.Rod19} which establishes boundedness and Theorem 4 in \cite{Daf.Hol.Rod19} where decay is established after having added to the solution an appropriately future normalised (dynamically determined) pure gauge solution.
\end{remark}

\begin{remark}
Note that despite the \underline{linear stability} statement of the theorem above, one may still expect \underline{non-linear instability} in view of the slow decay rate. However, the statement can still be used to establish rigidity properties concerning the Schwarzschild-AdS metric. These will be explored elsewhere. 
\end{remark}

\begin{remark}
We finally remark that an analogous theorem is expected to hold in the Kerr-AdS case for black hole parameters satisfying the Hawking-Reall bound. For small $a$, such a result should follow perturbatively from the techniques of this paper. We leave this for future study and refer to \cite{Gra.Hol23} for further discussion.
\end{remark}

\subsection{Overview of the paper and comments on the proof}
We will only very briefly comment on the global structure of the proof of the main theorem here. Afterwards, we immediately provide the formal set-up for the problem in Section  \ref{sec:overview}, which includes a derivation of the linearised Einstein equations in double null gauge. Section \ref{sec:data}  is concerned with the construction of appropriate initial data and solutions to this linearised system, i.e.~well-posedness of the linearised system.  A formal version of the main theorem is then formulated in Section \ref{sec:maintheorem} and proven in Section \ref{sec:proof}. The impatient reader wishing to take the existence of solutions of the linearised system (Theorem \ref{theo:wp}) for granted may turn immediately to the main theorem in Section \ref{sec:maintheorem}, which concerns the global properties of such solutions. 

 At the highest level, our strategy follows closely that of \cite{Daf.Hol.Rod19} in the asymptotically flat ($\Lambda=0$) case and starts by expressing the linearised Einstein equations in a double null gauge. A first key ingredient of the analysis, carried out in our companion paper \cite{Gra.Hola}, is to prove boundedness and decay estimates for the so-called Teukolsky quantities, denoted $\alin, \ablin$. These are certain linearised null-curvature components of the linearised system, which (a) do not depend on the specific gauge in which the equations (\ref{EVEL}) are linearised and (b) satisfy decoupled wave equations.\footnote{For $\Lambda=0$ these observations go back to Bardeen--Press \cite{Bardeen73} and Teukolsky \cite{Teu72} in the physics literature and are easily generalised to $\Lambda \neq 0$, see \cite{Kha83}. In our case, the two equations couple through the boundary condition imposed at the conformal boundary.} The second key ingredient is to exploit the hierarchical structure of the double null gauge to prove boundedness of all geometric quantities in a gauge normalised with respect to initial data using the bounds for the Teukolsky quantities. As already mentioned, in contrast to the asymptotically flat case, all Ricci coefficients and null curvature components can be shown to decay \emph{without} adding a residual pure gauge solution.
The reason can be understood as follows. As in \cite{Daf.Hol.Rod19}, one first proves boundedness and decay of the linearised shear $ \xlin$ from the estimates for $\alin$. This relies on the (commuted) redshift effect for $\xlin$. The quantity $\xblin$ then inherits this decay through the boundary conditions and we can hence integrate $\xblin$ in the ingoing direction from the boundary (using the estimates for $\ablin$) to establish boundedness and decay for $\xblin$.\footnote{This second step is not possible in the asymptotically flat case as there is no boundary and the $\xblin$ equation cannot be integrated directly from data in the ingoing direction either because of unfavourable $r$-weights in the integrating factor.} Decay for the other Ricci-coefficients and curvature components then follows by going hierarchically through the system analogous to \cite{Daf.Hol.Rod19}, except that here the boundary condition and its consequences need to be exploited at various stages. While one does not have to add a pure gauge solution to establish decay, one can improve the radial decay of certain geometric quantities and ensure that also the metric on the double null spheres converges to the round metric \emph{in standard form} by adding one. Such a pure gauge solution is constructed from the trace of the original solution at the conformal boundary and controlled uniformly by initial data. See Theorem \ref{mtheo:decay} below.

\subsection{Acknowledgements} 
G.H.~acknowledges support by the Alexander von Humboldt Foundation in the framework of the Alexander von Humboldt Professorship endowed by the Federal Ministry of Education and Research as well as ERC Consolidator Grant 772249. Both authors acknowledge funding through Germany’s Excellence Strategy EXC 2044 390685587, Mathematics M\"unster: Dynamics-Geometry-Structure.

\section{Preliminaries} \label{sec:overview}

In this section we provide the necessary background to set up the problem. We define the manifold on which the analysis takes place, introduce the double null gauge and explain the linearisation procedure leading us to the system of gravitational perturbations on the Schwarzschild-AdS manifold. The boundary conditions for the system are then also derived from the non-linear theory. The section ends with a discussion of pure gauge and linearised Kerr-AdS solutions of the system of gravitational perturbations.

\subsection{The manifold with boundary}
Let $\mathcal{Q} \subset \mathbb{R}^2_{U,V}$ be the $2$-dimensional submanifold with (piecewise smooth) boundary defined by 
\begin{align}
\mathcal{Q} :=  \left( \left[-1,0\right]_U \times \left[1,\infty\right)_V \right) \cap \{ -VU \leq 1\} \, .
\end{align}
We define the associated $4$-dimensional manifold
\begin{align}
\mathcal{M} := \mathcal{Q} \times S^2 \, , 
\end{align}
equipped with coordinates $(U,V,\theta^1,\theta^2)$, which we will refer to as Kruskal coordinates on $\mathcal{M}$. We denote the boundary components
\begin{align}
\mathcal{I} := \mathcal{M} \cap \{ -VU = 1\} \ \ \ , \ \ \ \mathcal{H}^+ = \mathcal{M} \cap \{U=0\} \ \ \ , \ \ \ N_{data} = \mathcal{M} \cap \{V=1\} \, , 
\end{align}
which will be referred to as null infinity, the future event horizon and the initial data hypersurface respectively. Observe that all boundary components are topologically $\mathbb{R} \times \mathbb{S}^2$.  We denote by $S^2_{U,V}$ the $2$-spheres $(U,V) \times S^2$ in $\mathcal{M}$. 

Given a fixed parameter $M>0$ we can define coordinates $(u,v,\theta^1,\theta^2)$ on $\mathcal{M} \setminus \mathcal{H}^+$ by
\begin{align}
u = - 2M \log (-U) \ \ \ , \ \ \ v = 2M \log V \, . 
\end{align}
Note $v \geq 0$ on $\mathcal{M} \setminus \mathcal{H}^+$ and $v=u$ on $\mathcal{I}$. 
Defining also $r^\star (u,v) = v-u$ and $t(u,v)=v+u$ we have another coordinate system $(t,r^\star,\theta^1,\theta^2)$ on $\mathcal{M} \setminus \mathcal{H}^+$. We observe that $r^\star$ is a boundary defining function for $\mathcal{I}$ in that it vanishes at $\mathcal{I}$ and $\partial_u r^\star=-1$, $\partial_v r^\star =1$. One may parametrise the boundary $\mathcal{I}=[0,\infty)_t \times S^2_{t,t}$ by the coordinate $t$. Finally, we can define on $\mathcal{M} \setminus \mathcal{H}^+$ a function $r(u,v)=r(r^\star=v-u)$ by the relation
\begin{align} \label{rstardef}
\frac{dr}{dr^\star} = 1-\frac{2M}{r} +k^2 r^2 \ \ \ , \ \ \ r(0) = \infty \, ,
\end{align}
where $-3k^2 =\Lambda$. Note that the function $r$ depends on $M$ and that we have the asymptotics 
\begin{align}
|(v-u) r| = |r^\star  r| = \frac{1}{k^2} + O(r^{-2}) .  \label{rrstarrel}
\end{align}

\subsection{One-parameter families of aAdS metrics in double null gauge} \label{sec:onepf}
Let us denote 
\begin{align}
\mathcal{M}_{int} = \mathcal{M} \setminus \mathcal{I}
\end{align}
and fix $M>0$, which we may think of as the mass of a Schwarzschild-AdS background metric we are about to install on $\mathcal{M}_{int}$. Given $\mathcal{M}_{int}$, equipped with local coordinates $(u, v, {\theta}^1, {\theta}^2)$ on $\mathcal{M}_{int} \setminus \mathcal{H}^+$, we consider a $1$-parameter family of metrics $\boldsymbol{g}(\epsilon)$ expressed in double null gauge\footnote{Note that the $\boldsymbol{b}$ is on $du$ here instead of on $dv$ as in \cite{Daf.Hol.Rod19}. As is well-known, this does not change the null-structure and Bianchi equations (collected in Section \ref{sec:propequations} below). The only change is in the propagation equation for the metric component $\boldsymbol{b}$, equation (\ref{bnle}), which is now in the outgoing direction and with an additional minus on the right. An \emph{outgoing} transport equation for $\boldsymbol{b}$ is desirable as it can be integrated from data, where $\boldsymbol{b}$ is normalised. Note also that it is $\boldsymbol{\Omega}^{-2} \boldsymbol{b}$ which extends regularly to the horizon $\mathcal{H}^+$ as can be seen by transforming (\ref{gindn}) to the regular Kruskal coordinates. \label{footnoteb}} 
\begin{align} \label{gindn}
\boldsymbol{g} \left(\epsilon\right) = -4 \boldsymbol\Omega^2 \left(\epsilon\right) d{u} d{v} + \slashed{\boldsymbol{g}}_{AB} \left(\epsilon\right) \left(d\theta^A - \boldsymbol{b}^A\left(\epsilon\right)   d{u}\right) \left(d\theta^B - \boldsymbol{b}^B \left(\epsilon\right) d{u}\right) \, 
\end{align}
such that
\begin{enumerate}
\item The $\boldsymbol{g} \left(\epsilon\right)$ 
satisfy on $\mathcal{M}_{int} \setminus \mathcal{H}^+$ the Einstein equations with negative cosmological constant $\Lambda = - 3k^2$,
\begin{align}
\mathrm{Ric}(\boldsymbol{g}) = -3k^2 \boldsymbol{g} \, .
\end{align}
\item We have that $\boldsymbol{g} \left(0\right)$ is the Schwarzschild-AdS metric of mass $M$ and cosmological constant $\Lambda = - 3k^2$ 
\begin{align}
\boldsymbol{g} \left(0\right) := -4 \left(1-\frac{2M}{r(u,v)} + k^2 \cdot r^2({u},{v}) \right) d{u} d{v} + r^2 ({u},{v}) {\gamma}_{AB} d{\theta}^A d{\theta}^B \, ,
\end{align}
where $r(u,v)$ is defined as in (\ref{rstardef}) and ${\gamma}$ denotes the round metric on the unit sphere.
\item The $\boldsymbol{g} \left(\epsilon\right)$ are asymptotically Anti-de Sitter in that the function
$(v-u)^2 \boldsymbol\Omega^2$ as well as the conformally rescaled metric $\boldsymbol\Omega^{-2} \boldsymbol{g}$ extend regularly to $\mathcal{I}$, i.e.~in particular they can be defined on the larger manifold $\mathcal{M} \setminus \mathcal{H}^+$. Specifically, recalling $t=u+v$, the family
\begin{align}
\boldsymbol{g}^{(3)}_{\mathcal{I}} (\epsilon)= -d{t}^2 + \frac{\slashed{\boldsymbol{g}}_{AB}(\epsilon)}{\boldsymbol\Omega^2(\epsilon)} \left( d\theta^A - \boldsymbol{b}^A(\epsilon) \frac{d{t}}{2}\right)  \left( d\theta^B - \boldsymbol{b}^B(\epsilon) \frac{d{t}}{2}\right) \, 
\end{align}
defines a smooth family of $3$-dimensional Lorentzian metrics on $\mathcal{I}$.
\item The $\boldsymbol{g}^{(3)}_{\mathcal{I}} (\epsilon)$ are all conformal to the Lorentzian cylinder $(\mathbb{R} \times S^2, -k^2\mathrm{d}t^2 + \gamma)$, the latter being the metric induced on $\mathcal{I}$ by the conformally rescaled Schwarzschild-AdS metric $\boldsymbol\Omega^{-2} \boldsymbol{g}(0)$. In particular, the $\boldsymbol{g}^{(3)}_{\mathcal{I}} (\epsilon)$ are all locally conformally flat.  

\item In regular Kruskal coordinates, the $\boldsymbol{g} \left(\epsilon\right)$ and regular derivatives thereof 
extend smoothly to the boundary $\mathcal{H}^+$ of $\mathcal{M}_{int}$ in the sense of Section 5.1.1 of \cite{Daf.Hol.Rod19}. In particular, arbitrary concatenations of frame vectors from the set $\{ \boldsymbol\Omega^{-2}(\epsilon) (\partial_u + \boldsymbol{b}^A \partial_A ), \partial_v, \partial_{\theta^1},\partial_{\theta^2} \}$ applied to $\boldsymbol{g} \left(\epsilon\right)$ extend regularly to $\mathcal{H}^+$.

\end{enumerate}

The existence of such families is of course an implicit assumption. Locally in time, i.e.~for a finite $v$-interval this can be fully justified by a local-well-posedness theorem in double null gauge. Such a theorem could either be inferred from the literature \cite{Fri95, Enc.Kam19} or be proven directly by combining the (linear) estimates obtained in this paper with an appropriate contraction mapping argument. Since our main theorem is formulated directly as a statement concerning the linearised system (for which we will prove well-posedness directly) we will not address this issue further. 


\subsection{The geometry of a double null gauge}
For the reader's convenience we briefly recall the basic geometric notions in a double null gauge. The familiar reader can jump immediately to Section \ref{sec:linproc} while the reader unfamiliar with the double null gauge can consult Section 3 of \cite{Daf.Hol.Rod19} or the original \cite{Chr09} for many more details.  

Associated with a double null gauge (\ref{gindn}) on $\mathcal{M}_{int} \setminus \mathcal{H}^+$ is a double null frame consisting of the null vectorfields
\[
e_3 = \frac{1}{\boldsymbol{\Omega}} \left( \partial_u + \boldsymbol{b}^A \partial_A \right) \ \ , \ \ e_4 = \frac{1}{\boldsymbol{\Omega}} \partial_v \, ,
\]
satisfying ${\bf g}(e_3,e_4)=-2$, ${\bf g}(e_3,e_3)={\bf g}(e_4,e_4)=0$, which is complemented with a local coordinate frame on $S^2_{u,v}$, $e_A=\partial_A$ for $A \in \{1,2\}$, satisfying ${\bf g}(e_A,e_B)=\slashed{\boldsymbol{g}}_{AB}$. Note ${\bf g}(e_3,e_A)=0$ and ${\bf g}(e_4,e_A)$ for $A \in \{1,2\}$.

\subsubsection{$S^2_{u,v}$-tensor algebra} \label{sec:ta}
Let $\boldsymbol\xi, \boldsymbol{\tilde \xi}$ be arbitrary $S^2_{u,v}$ one-forms and
$\boldsymbol\theta, \boldsymbol{\tilde \theta}$
be arbitrary
symmetric  $S^2_{u,v}$ $2$-tensors.

We denote by ${}^\star \boldsymbol\xi$ and ${}^\star \boldsymbol\theta$ the Hodge-dual on $\left( S^2_{u,v}, \boldsymbol{\slashed{g}}\right)$ of $\boldsymbol\xi$ and $\boldsymbol\theta$, respectively, and
denote by $\boldsymbol \theta^\sharp$ the tensor obtained from $\boldsymbol\theta$ by raising an index with $\slashed{\boldsymbol{g}}$. We define the contractions
\[
\left(\boldsymbol\xi,\boldsymbol{\tilde \xi}\right):=\boldsymbol{\slashed{g}}^{AB}\boldsymbol\xi_A \boldsymbol{\tilde \xi}_B \textrm{ \ \ \  and \ \ \ } \left(\boldsymbol\theta,\boldsymbol{\tilde\theta}\right):=\boldsymbol{\slashed{g}}^{AB}\boldsymbol{\slashed{g}}^{CD}\boldsymbol\theta_{AC} \boldsymbol{\tilde\theta}_{BD},
\]
and
denote by $\boldsymbol\theta^\sharp \cdot \boldsymbol\xi$ the one-form $\boldsymbol\theta_A^{\phantom{A}B} \boldsymbol\xi_B$ arising from the contraction with $\slashed{\boldsymbol{g}}$. We finally define
\begin{align}
\left(\boldsymbol\theta \times \boldsymbol{\tilde\theta}\right)_{BC}&:=\slashed{\boldsymbol{g}}^{AD}\boldsymbol\theta_{AB}\boldsymbol{\tilde\theta}_{DC} \, ,\nonumber \\
\left(\boldsymbol\xi  \widehat{\otimes} \boldsymbol{\tilde \xi} \right)_{AB}
&:= \boldsymbol\xi_A \boldsymbol{\tilde \xi}_B + \boldsymbol\xi_B \boldsymbol{\tilde \xi}_A - \slashed{\boldsymbol{g}}^{AB}\boldsymbol\xi_A \boldsymbol{\tilde \xi}_B  \, ,
\nonumber \\
\boldsymbol\theta  \wedge \boldsymbol{\tilde\theta} &:= \boldsymbol{\slashed{\epsilon}}^{AB} \slashed{\boldsymbol{g}}^{CD}\boldsymbol\theta_{AC} \boldsymbol{\tilde\theta}_{BD} \, , \nonumber
\end{align}
where $\boldsymbol{\slashed{\epsilon}}_{AB}$ denotes the components of the volume form associated with $\slashed{\boldsymbol{g}}$ on  $S^2_{u,v}$. %
\subsubsection{$S^2_{u,v}$-projected Lie and
covariant derivates} \label{sec:pcd}
We define the derivative operators $\underline{\boldsymbol D}$ and ${\boldsymbol D}$ to act on an  $S^2_{u,v}$-tensor $\boldsymbol\phi$ as 
the projection onto  $S^2_{u,v}$ of the Lie-derivative of $\boldsymbol\phi$ in the direction of $\boldsymbol \Omega{\boldsymbol{e}}_3$ and $\boldsymbol \Omega{\boldsymbol{e}}_4$ respectively. We hence have the following relations between the projected Lie-derivatives $\underline{\boldsymbol D}$ and ${\boldsymbol D}$ and the  $S^2_{u,v}$-projected spacetime covariant derivatives $ \slashed{\nabla}_{\boldsymbol 3} = \slashed{\nabla}_{{\boldsymbol{e}}_3},\slashed{\nabla}_{\boldsymbol 4}= \slashed{\nabla}_{{\boldsymbol{e}}_4}$ in the direction ${\boldsymbol{e}}_3$ and ${\boldsymbol{e}}_4$ respectively:
\begin{equation} \label{Dcovtransform}
\begin{split}
\boldsymbol {D} \boldsymbol f &= \boldsymbol{\Omega} {\slashed{\nabla}}_{\boldsymbol 4} \boldsymbol f \textrm{ \ \ \ on functions $\boldsymbol{f}$,} \\
\boldsymbol {D} \boldsymbol\xi &= \boldsymbol{\Omega} {\slashed{\nabla}}_{\boldsymbol 4}  \boldsymbol \xi + \boldsymbol{\Omega} \boldsymbol\chi^\sharp \cdot \boldsymbol\xi \textrm{ \ \ \ on one-forms $\boldsymbol \xi$,}
  \\
\boldsymbol {D} \boldsymbol\theta &=\boldsymbol{\Omega} {\slashed{\nabla}}_{\boldsymbol 4}  \boldsymbol\theta + \boldsymbol{\Omega} \boldsymbol\chi \times \boldsymbol\theta + \boldsymbol{\Omega} \boldsymbol\theta \times \boldsymbol\chi \textrm{ \ \ \ on symmetric $2$-tensors $\boldsymbol\theta$,}
\end{split}
\end{equation}
and similarly for ${\slashed{\nabla}}_{\boldsymbol 3}$ replacing $\boldsymbol\chi$ by $\underline{\boldsymbol\chi}$ and ${\boldsymbol D}$ by $\underline{\boldsymbol D}$.

\subsubsection{Angular operators on  $S^2_{u,v}$} \label{sec:angop}
Let $\boldsymbol\xi$ be an arbitrary one-form and $\boldsymbol\theta$ an arbitrary symmetric traceless $2$-tensor on  $S^2_{u,v}$.
\begin{itemize}
\item $\boldsymbol{\slashed{\nabla}}$ denotes the covariant derivative associated with the metric $\slashed{\boldsymbol g}_{AB}$ on  $S^2_{u,v}$.
\item $\boldsymbol{\slashed{\mathcal{D}}}_1$ takes $\boldsymbol \xi$ into the pair of functions $\left(\slashed{\bf {div}} \boldsymbol \xi, \slashed{\bf {curl}} \boldsymbol \xi\right)$
where $\slashed{\bf {div}} \boldsymbol \xi = \slashed{\boldsymbol{g}}^{AB} \boldsymbol{\slashed{\nabla}}_A \boldsymbol\xi_B$ and $\slashed{\bf {curl}} \boldsymbol \xi = \boldsymbol{\slashed{\epsilon}}^{AB}  \boldsymbol{\slashed{\nabla}}_A \boldsymbol\xi_B$.
\item $\slashed{\mathbfcal{D}}_1^\star$, the $L^2$-adjoint of $\boldsymbol{\slashed{\mathcal{D}}}_1$, takes any pair of scalars $\boldsymbol \rho, \boldsymbol \sigma$ into the  $S^2_{u,v}$-one-form $-\boldsymbol{\slashed{\nabla}}_A \boldsymbol \rho + \boldsymbol{\slashed{\epsilon}}_{AB} \boldsymbol{\slashed{\nabla}}^B \boldsymbol \sigma$.
\item $\slashed{\mathbfcal{D}}_2$ takes $\boldsymbol\theta$ into the  $S^2_{u,v}$-one-form $\left(\slashed{\bf {div}} \boldsymbol\theta\right)_C=\slashed{\boldsymbol{g}}^{AB} \boldsymbol{\slashed{\nabla}}_A \boldsymbol\theta_{BC}$.
\item $\slashed{\mathbfcal{D}}_2^\star$, the $L^2$ adjoint of $\slashed{\mathbfcal{D}}_2$, takes $\boldsymbol\xi$ into $2$-tensor $\left(\slashed{\mathbfcal{D}}^\star_2 \boldsymbol \xi\right)_{AB}=-\frac{1}{2} \left(\boldsymbol{\slashed{\nabla}}_B \boldsymbol\xi_A + \boldsymbol{\slashed{\nabla}}_A \boldsymbol\xi_B - \left(\slashed{\bf {div}} \boldsymbol\xi\right) \slashed{\boldsymbol{g}}_{AB}\right)$.
\end{itemize}

\subsubsection{Ricci-coefficients and curvature components}  \label{sec:ta3}
We define the non-vanishing null-decomposed Ricci coefficients as follows:
\begin{equation} \label{RicC}
\begin{split}
\boldsymbol\chi_{AB} &= \boldsymbol{g} \left(\boldsymbol{\nabla}_A {\boldsymbol{e}}_4,{\boldsymbol{e}}_B\right) \textrm{ \ \ , \ \ \ \ } \underline{\boldsymbol\chi}_{AB} =  \boldsymbol{g} \left(\boldsymbol{\nabla}_A {\boldsymbol{e}}_3,{\boldsymbol{e}}_B\right) \, ,
\\
 \boldsymbol\eta_{A} &= -\frac{1}{2}  \boldsymbol{g} \left(\boldsymbol{\nabla}_{{\boldsymbol{e}}_3} {\boldsymbol{e}}_A,{\boldsymbol{e}}_4\right) \textrm{ \ \ , \ \ } \underline{\boldsymbol\eta}_{A} = -\frac{1}{2}  \boldsymbol{g} \left(\boldsymbol{\nabla}_{{\boldsymbol{e}}_4}{\boldsymbol{e}}_A,{\boldsymbol{e}}_3\right)  \ \  , \ \ \ \boldsymbol\zeta = \frac{1}{2}  \boldsymbol{g} \left(\boldsymbol{\nabla}_A {\boldsymbol{e}}_4,{\boldsymbol{e}}_3\right) \, ,
\\
\boldsymbol\omega &= \frac{1}{2} \boldsymbol{\Omega} \boldsymbol{g} \left(\boldsymbol{\nabla}_{{\boldsymbol{e}}_4} {\boldsymbol{e}}_3,{\boldsymbol{e}}_4\right) \textrm{ \ \ \ \ , \ \ \ \ \ } \underline{\boldsymbol\omega} = \frac{1}{2}  \boldsymbol{\Omega} \boldsymbol{g} \left(\boldsymbol{\nabla}_{{\boldsymbol{e}}_3} {\boldsymbol{e}}_4,{\boldsymbol{e}}_3\right)  \, .
\end{split}
\end{equation}
%
%
The above objects are $S^2_{u,v}$ scalars, one-forms and symmetric traceless tensors respectively. In particular, they transform tensorially under a choice of frame on the sphere. 
It is natural to decompose ${\boldsymbol\chi}$ into its $\slashed{\boldsymbol{g}}$-tracefree part
$\boldsymbol{\hat\chi}$ (a symmetric traceless $S^2_{u,v}$ 2-tensor) and its trace $\boldsymbol{tr}\boldsymbol\chi$,
and similarly for $\underline{\boldsymbol\chi}$. Note also the relations
\begin{align} \label{melek}
{\boldsymbol\omega} = \frac{\boldsymbol{\partial}_v {\boldsymbol\Omega}}{\boldsymbol\Omega} \ \ , \ \ \underline{{\boldsymbol\omega}} = \frac{(\boldsymbol{\partial}_{\boldsymbol{u}} + \boldsymbol{b}^A \boldsymbol{\partial}_{\boldsymbol\theta^A}) {\boldsymbol\Omega}}{\boldsymbol\Omega} \ \ , \ \
\boldsymbol\eta_A = \boldsymbol\zeta_A + \boldsymbol{\slashed{\nabla}}_A \log \boldsymbol{\Omega}  \textrm{ \ \  ,  \ \ } \underline{\boldsymbol\eta}_A = -\boldsymbol\zeta_A + \boldsymbol{\slashed{\nabla}}_A \log \boldsymbol{\Omega}  .
\end{align}

With $\boldsymbol{W}$ denoting the Weyl curvature tensor of the metric (\ref{gindn}), the null-decomposed Weyl curvature components  are defined as follows:
\begin{equation} \label{curvC}
\begin{split}
\boldsymbol\alpha_{AB} &= \boldsymbol{W} \left( {\boldsymbol{e}}_A,  {\boldsymbol{e}}_4,  {\boldsymbol{e}}_B,  {\boldsymbol{e}}_4\right) \textrm{ \ \ \ ,  \ \  } \underline{\boldsymbol\alpha}_{AB} = \boldsymbol{W} \left( {\boldsymbol{e}}_A,  {\boldsymbol{e}}_3,  {\boldsymbol{e}}_B,  {\boldsymbol{e}}_3\right) \, , \\
\boldsymbol\beta_{A} &= \frac{1}{2} \boldsymbol{W} \left( {\boldsymbol{e}}_A,   {\boldsymbol{e}}_4,  {\boldsymbol{e}}_3,   {\boldsymbol{e}}_4\right) \textrm{ \ \ \ ,  \ \ \ } \underline{\boldsymbol\beta}_{A} = \frac{1}{2} \boldsymbol{W} \left( {\boldsymbol{e}}_A,  {\boldsymbol{e}}_3,  {\boldsymbol{e}}_3,   {\boldsymbol{e}}_4\right) \, , \ \\
\boldsymbol\rho &= \frac{1}{4} \boldsymbol{W}\left(  {\boldsymbol{e}}_4,  {\boldsymbol{e}}_3, {\boldsymbol{e}}_4,  {\boldsymbol{e}}_3\right) \textrm{ \ \ \ , \ \ \ \ \ \  } \boldsymbol\sigma = \frac{1}{4} {}^\star \boldsymbol{W}\left(  {\boldsymbol{e}}_4,  {\boldsymbol{e}}_3, {\boldsymbol{e}}_4,  {\boldsymbol{e}}_3\right) \, ,
\end{split}
\end{equation}
with ${}^\star \boldsymbol{W} $ denoting the Hodge dual on $({\mathcal{M},\boldsymbol{g}})$ of $\boldsymbol{W}$.
Again the above objects are $S^2_{u,v}$-tensors (functions, vectors, symmetric $2$-tensors) on $(\mathcal{M},\boldsymbol{g})$.

\subsubsection{The null structure and Bianchi equations} \label{sec:propequations}
In the geometric setting outline above, the Einstein equations (\ref{EVEL}) imply (via the Bianchi equations and the geometric structure equations) a complicated system of coupled hyperbolic, transport and elliptic equations for $S^2_{u,v}$-tensors that we collect below. See again \cite{Daf.Hol.Rod19} and in particular \cite{Chr09} for a detailed derivation of the equations in the $\Lambda=0$ case. We have highlighted the additional terms arising from the non-vanishing cosmological constant in (\ref{EVEL}) as boxed terms below.

We have the first variational formulae:
\begin{align} \label{firstvarf}
\boldsymbol{D} \slashed{\boldsymbol{g}} = 2 \boldsymbol{\Omega} \boldsymbol\chi = 2\boldsymbol\Omega \hat{\boldsymbol\chi} + \boldsymbol\Omega {\bf tr} \boldsymbol\chi \slashed{g} \textrm{ \ \ \ \ \ and \ \ \ \ \ } \underline{\boldsymbol{D}} \slashed{\boldsymbol{g}} = 2 \boldsymbol\Omega \underline{\boldsymbol\chi} = 2 \boldsymbol\Omega \underline{\hat{\boldsymbol\chi}} + \boldsymbol\Omega {\bf  tr} {\boldsymbol\chi} \slashed{\boldsymbol{g}} \, .
\end{align}
The transport equations for the second fundamental forms take the form
\begin{align}  \label{chieq}
\boldsymbol{\slashed{\nabla}}_3 \underline{\hat{\boldsymbol\chi}} + {\bf tr} \underline{\boldsymbol\chi} \ \hat{\underline{\boldsymbol\chi}} - \boldsymbol{\Omega}^{-1} {\underline{\boldsymbol\omega}} \ \underline{\hat{\boldsymbol\chi}} =  -\underline{\boldsymbol\alpha} \ \ \ ,  \ \ \ \boldsymbol{\slashed{\nabla}}_4 {\hat{\boldsymbol\chi}} + {\bf tr} {\boldsymbol\chi} \ \hat{\boldsymbol\chi} - \boldsymbol{\Omega}^{-1}{{\boldsymbol\omega}} \ {\hat{\boldsymbol\chi}} =  -{\boldsymbol\alpha} ,
\end{align}
\begin{align}
\boldsymbol{\slashed{\nabla}}_3 \left({\bf tr} \underline{\boldsymbol\chi}\right) + \frac{1}{2}\left({\bf tr} \underline{\boldsymbol\chi}\right)^2 - \boldsymbol{\Omega}^{-1} \underline{{\boldsymbol\omega}} {\bf tr} \underline{{\boldsymbol\chi}} = - \left(\underline{\hat{\boldsymbol\chi}} , \underline{\hat{\boldsymbol\chi}}\right) \ \ , \ \ \boldsymbol{\slashed{\nabla}}_4 \left({\bf tr} \boldsymbol\chi \right) + \frac{1}{2}\left({\bf tr} \boldsymbol\chi\right)^2 - \boldsymbol{\Omega}^{-1}{\boldsymbol\omega} {\bf tr} {\boldsymbol\chi} = -
\left( {\hat{\boldsymbol\chi}}, {\hat{\boldsymbol\chi}}\right) ,
\end{align}
\begin{align} \label{chieq2a}
\boldsymbol{\slashed{\nabla}}_3 {\hat{\boldsymbol\chi}} + \frac{1}{2} {\bf tr} \underline{\boldsymbol\chi} \ \hat{{\boldsymbol\chi}} +\boldsymbol{\Omega}^{-1} {\underline{\boldsymbol\omega}} \ {\hat{\boldsymbol\chi}} = -2 \slashed{\mathbfcal{D}}_2^\star \boldsymbol{\eta} - \frac{1}{2} {\bf tr} \boldsymbol \chi \ \underline{\hat{\boldsymbol\chi}} + \left({\boldsymbol\eta} \widehat{\otimes} {\boldsymbol\eta}\right) \, ,
\\
\label{chieq2b}
\boldsymbol{\slashed{\nabla}}_4 \underline{\hat{\boldsymbol\chi}} + \frac{1}{2} {\bf tr} {\boldsymbol\chi} \ \hat{\underline{\boldsymbol\chi}} +\boldsymbol{\Omega}^{-1}{{\boldsymbol\omega}} \ \underline{\hat{\boldsymbol\chi}} = -2 \slashed{\mathbfcal{D}}_2^\star \underline{\boldsymbol\eta} - \frac{1}{2} {\bf tr} \underline{\boldsymbol \chi} \ \hat{\boldsymbol \chi} + \left(\underline{\boldsymbol\eta} \widehat{\otimes} \underline{\boldsymbol\eta}\right) \, ,
\end{align}

\begin{align}
\boldsymbol{\slashed{\nabla}}_3 \left({\bf tr} {\boldsymbol\chi}\right) + \frac{1}{2}\left({\bf tr} \underline{\boldsymbol\chi}\right) \left({\bf tr} {\boldsymbol\chi}\right)+ \boldsymbol{\Omega}^{-1}\underline{{\boldsymbol\omega}} {\bf tr} {{\boldsymbol\chi}} = - \left(\underline{\hat{\boldsymbol\chi}}, {\hat{\boldsymbol\chi}}\right) + 2 \left( \boldsymbol\eta , \boldsymbol\eta\right) + 2\boldsymbol\rho +2\slashed{\bf{div}} \boldsymbol\eta \boxed{-4k^2} \, ,
\\
 \label{endchi}
\boldsymbol{\slashed{\nabla}}_4 \left({\bf tr} \underline{\boldsymbol\chi}\right) + \frac{1}{2}\left({\bf tr} {\boldsymbol\chi}\right) \left({\bf tr} \underline{\boldsymbol\chi}\right)+\boldsymbol{\Omega}^{-1} {{\boldsymbol\omega}} {\bf tr} \underline{{\boldsymbol\chi}} = - \left( \underline{\hat{\boldsymbol\chi}}, {\hat{\boldsymbol\chi}} \right) + 2 \left( \underline{\boldsymbol\eta} ,\underline{\boldsymbol\eta}\right) + 2\boldsymbol\rho + 2\slashed{\bf{div}} \underline{\boldsymbol\eta} \boxed{- 4k^2} \, .
\end{align}
The transport equations for the torsions and the (derivative of the) lapse become:
\begin{align}
\boldsymbol{\slashed{\nabla}}_3 \underline{\boldsymbol\eta} = \underline{\boldsymbol\chi}^\sharp \cdot \left(\boldsymbol\eta - \underline{\boldsymbol\eta}\right) + \underline{\boldsymbol\beta} \ \ \ , \ \ \
\boldsymbol{\slashed{\nabla}}_4 {\boldsymbol\eta} = -{\boldsymbol\chi}^\sharp \cdot \left(\boldsymbol\eta - \underline{\boldsymbol\eta}\right) - {\boldsymbol\beta} \, ,
\end{align}
\begin{align}
\boldsymbol{D} \left( {\underline{\boldsymbol\omega}}\right) = \boldsymbol\Omega^2 \left[ 2 \left(\boldsymbol\eta, \underline{\boldsymbol\eta}\right) - |\boldsymbol\eta|^2 - \boldsymbol\rho  \boxed{-k^2}  \right] \textrm{ \ , \ } \underline{\boldsymbol{D}} \left( {\boldsymbol\omega} \right)= \boldsymbol\Omega^2 \left[ 2 \left(\boldsymbol\eta, \underline{\boldsymbol\eta}\right) - |\underline{\boldsymbol\eta}|^2 - \boldsymbol\rho \boxed{-k^2} \right] \, ,\nonumber
\end{align}
\begin{align} \label{bnle}
\boldsymbol{\partial}_{v}
\boldsymbol{b}^A = -2\boldsymbol{\Omega}^2 \left(\boldsymbol\eta^A-\underline{\boldsymbol\eta}^A\right) \, .
\end{align}
Finally, we have the elliptic relations on spheres:
\begin{align}
\slashed{\bf{curl}} \boldsymbol\eta = - \frac{1}{2} \boldsymbol\chi \wedge \underline{\boldsymbol\chi} + \boldsymbol\sigma \textrm{ \ \ \ \ \ and \ \ \ \ \ }\slashed{\bf {curl}} \underline{\boldsymbol\eta} = +\frac{1}{2} \boldsymbol\chi \wedge \underline{\boldsymbol\chi} - \boldsymbol\sigma \, ,
\end{align}
\begin{align} \label{cod1}
\slashed{\bf{ div}} \hat{\boldsymbol\chi} &=  -\frac{1}{2} \hat{\boldsymbol\chi}^\sharp \cdot \left( \boldsymbol\eta - \underline{\boldsymbol\eta}\right) + \frac{1}{4} {\bf tr} \boldsymbol\chi \left( \boldsymbol\eta - \underline{\boldsymbol\eta}\right)  + \frac{1}{2} \boldsymbol{\slashed{\nabla}} {\bf tr} {\boldsymbol \chi} - \boldsymbol\beta \nonumber \\
&=  -\frac{1}{2} \hat{\boldsymbol\chi}^\sharp \cdot \left( \boldsymbol\eta - \underline{\boldsymbol\eta}\right) - \frac{1}{2} {\bf tr} \boldsymbol\chi  \underline{\boldsymbol\eta}  + \frac{1}{2\boldsymbol\Omega} \boldsymbol{\slashed{\nabla}} \left( {\bf \Omega} {\bf tr} {\boldsymbol \chi}\right) - \boldsymbol\beta \, ,
\end{align}
\begin{align} \label{cod2}
\slashed{\bf {div}} \underline{\hat{\boldsymbol\chi}} &=  \frac{1}{2} \underline{\hat{\boldsymbol\chi}}^\sharp \cdot \left( \boldsymbol\eta - \underline{\boldsymbol\eta}\right) - \frac{1}{4} {\bf tr} \underline{\boldsymbol\chi} \left( \boldsymbol\eta - \underline{\boldsymbol\eta}\right)  + \frac{1}{2} \boldsymbol{\slashed{\nabla}} {\bf tr} \underline{\boldsymbol \chi} + \underline{\boldsymbol\beta} \nonumber \\
&=  \frac{1}{2} \underline{\hat{\boldsymbol\chi}}^\sharp \cdot \left( \boldsymbol\eta - \underline{\boldsymbol\eta}\right) - \frac{1}{2} {\bf tr} \underline{\boldsymbol\chi} {\boldsymbol\eta}  + \frac{1}{2\boldsymbol\Omega} \boldsymbol{\slashed{\nabla}} \left({\bf \Omega} {\bf tr} \underline{\boldsymbol \chi}\right) + \underline{\boldsymbol\beta} \, ,
\end{align}
\begin{align} \label{gauss}
\mathbf{K} = -\frac{1}{4} {\bf tr} \boldsymbol\chi {\bf tr} \underline{\boldsymbol\chi} + \frac{1}{2}\left( \hat{\boldsymbol\chi} , \hat{\underline{\boldsymbol\chi}} \right) - \boldsymbol\rho \boxed{-k^2}\, .
\end{align}
Equations (\ref{cod1})--(\ref{cod2}) are known as the Codazzi equations, (\ref{gauss}) is the Gauss equation on $S^2_{u,v}$. 
We finally collect the Bianchi equations for the null Weyl curvature components which are formally unchanged in the presence of a cosmological constant. 
\begin{align}
\boldsymbol{\slashed{\nabla}}_3  \boldsymbol\alpha + \frac{1}{2} {\bf  tr} \underline{\boldsymbol\chi} \boldsymbol\alpha + 2 \boldsymbol{\Omega}^{-1} \underline{{\boldsymbol\omega}} \boldsymbol\alpha &= -2 {\bf \slashed{\mathbfcal{D}}_2^\star }\boldsymbol\beta - 3 \hat{\boldsymbol\chi} \boldsymbol\rho - 3{}^\star \hat{\boldsymbol\chi} \boldsymbol\sigma  + \left(4\boldsymbol\eta + \boldsymbol\zeta\right) \hat{\otimes} \boldsymbol\beta \, , \nonumber \\
\boldsymbol{\slashed{\nabla}}_4 \boldsymbol\beta + 2 {\bf  tr}\boldsymbol\chi \boldsymbol\beta - \boldsymbol{\Omega}^{-1} {\boldsymbol\omega} \boldsymbol\beta &=  \slashed{\bf{div}}\boldsymbol\alpha + \left(\underline{\boldsymbol\eta}^\sharp + 2 \boldsymbol\zeta^\sharp\right) \cdot \boldsymbol\alpha 
 \, ,  \nonumber \\
\boldsymbol{\slashed{\nabla}}_3 \boldsymbol\beta + {\bf  tr} \underline{\boldsymbol\chi} + \boldsymbol{\Omega}^{-1} \underline{{\boldsymbol\omega}} \boldsymbol\beta &= \slashed{\mathbfcal{D}}_1^\star \left(-\boldsymbol\rho, \boldsymbol\sigma\right) + 3 \boldsymbol\eta \boldsymbol\rho + 3{}^\star \boldsymbol\eta \boldsymbol\sigma + 2\hat{\boldsymbol\chi}^\sharp \cdot \underline{\boldsymbol\beta}
 \, ,  \nonumber \\
\boldsymbol{\slashed{\nabla}}_4 \boldsymbol\rho  + \frac{3}{2}{\bf  tr} \boldsymbol\chi\boldsymbol\rho &= \slashed{\bf{div}} \boldsymbol\beta + \left(2\underline{\boldsymbol\eta} + \boldsymbol\zeta, \boldsymbol\beta \right) - \frac{1}{2} \left(\underline{\hat{\boldsymbol\chi}}, \boldsymbol\alpha\right) 
 \, ,  \nonumber \\
\boldsymbol{\slashed{\nabla}}_4 \boldsymbol\sigma + \frac{3}{2} {\bf  tr} \boldsymbol\chi \boldsymbol\sigma &= -\slashed{\bf{curl}}\boldsymbol\beta - \left(2\underline{\boldsymbol\eta} + \boldsymbol\zeta \right) \wedge \boldsymbol\beta  + \frac{1}{2} \underline{\hat{\boldsymbol\chi}} \wedge \boldsymbol\alpha 
 \, ,  \nonumber
\\
\boldsymbol{\slashed{\nabla}}_3 \boldsymbol\rho + \frac{3}{2} {\bf  tr} \underline{\boldsymbol\chi}\boldsymbol\rho &= -\slashed{\bf{div}}\underline{\boldsymbol\beta} - \left(2\boldsymbol\eta - \boldsymbol\zeta, \underline{\boldsymbol\beta}\right) - \frac{1}{2} \left(\hat{\boldsymbol\chi}, \underline{\boldsymbol\alpha}\right) 
 \, , \nonumber \\
\boldsymbol{\slashed{\nabla}}_3 \boldsymbol\sigma + \frac{3}{2} {\bf  tr} \underline{\boldsymbol\chi} \boldsymbol\sigma &= -\slashed{\bf{curl}} \underline{\boldsymbol\beta} 
- \left(2\boldsymbol\eta - \boldsymbol\zeta\right) \wedge \underline{\boldsymbol\beta} - \frac{1}{2} \hat{\boldsymbol\chi} \wedge \underline{\boldsymbol\alpha}
 \, ,  \nonumber \\
\boldsymbol{\slashed{\nabla}}_4 \underline{\boldsymbol\beta} + {\bf  tr} \boldsymbol\chi  \underline{\boldsymbol\beta} + \boldsymbol{\Omega}^{-1} {\boldsymbol\omega} \underline{\boldsymbol\beta} &= \slashed{\mathbfcal{D}}_1^\star \left(\boldsymbol\rho, \boldsymbol\sigma\right) - 3 \underline{\boldsymbol\eta} \boldsymbol\rho + 3{}^\star \underline{\boldsymbol\eta} \boldsymbol\sigma  + 2 \underline{\hat{\boldsymbol\chi}}^\sharp \cdot \boldsymbol\beta 
 \, , \nonumber \\
\boldsymbol{\slashed{\nabla}}_3 \underline{\boldsymbol\beta} + 2 {\bf  tr} \underline{\boldsymbol\chi}  \underline{\boldsymbol\beta} - \boldsymbol{\Omega}^{-1} {\underline{\boldsymbol\omega}} \underline{\boldsymbol\beta} &= - \slashed{\bf{div}} \underline{\boldsymbol\alpha} - \left(\boldsymbol\eta^\sharp - 2 \boldsymbol\zeta^\sharp\right) \cdot \underline{\boldsymbol\alpha}
 \, ,  \nonumber \\
\boldsymbol{\slashed{\nabla}}_4\underline{\boldsymbol\alpha} + \frac{1}{2}{\bf  tr} \boldsymbol\chi \underline{\boldsymbol\alpha} + 2 \boldsymbol{\Omega}^{-1} {\boldsymbol\omega} \underline{\boldsymbol\alpha} &=  2 \slashed{\mathbfcal{D}}_2^\star \underline{\boldsymbol\beta} - 3 \underline{\hat{\boldsymbol\chi}} \boldsymbol\rho + 3{}^\star \underline{\hat{\boldsymbol\chi}} \boldsymbol\sigma - \left(4 \underline{\boldsymbol\eta} - \boldsymbol\zeta\right) \hat{\otimes} \underline{\boldsymbol\beta}
\, .
\nonumber
\end{align}


\subsection{Boundary regularity and boundary conditions}\label{sec:boundarynonlinear} 
We collect the asymptotic behaviour towards the conformal boundary for the geometric quantities on $\mathcal{M}_{int}$ from the assumption that (by $3.$ of Section \ref{sec:onepf}) the metric $(u-v)^2 {\boldsymbol{g}}$ extends regularly to the larger manifold $\mathcal{M}$. The proof of Proposition~\ref{prop:ad} is postponed to Appendix~\ref{app:proofpropad}.
\begin{proposition} \label{prop:ad}
  Assume that $\boldsymbol{g}$ is asymptotically Anti-de Sitter ($3.$ of Section \ref{sec:onepf}). The following geometric quantities associated with $\boldsymbol{g}$ defined in Section \ref{sec:onepf} extend regularly to the conformal boundary $\mathcal{I}$ in the sense that the components with respect to a $\boldsymbol{g}$-orthonormal frame extend smoothly:
  \begin{align} \label{boundsextend}
    \begin{aligned}
      r^2\left(k^2(u-v)^2\boldsymbol{\Omega}^2 -1\right), r^{-1} \boldsymbol{b},\\
      \boldsymbol{\Omega} {\bf tr} {\boldsymbol{\chi}} - \frac{2}{u-v} , \boldsymbol{\Omega} {\bf tr} \underline{\boldsymbol{\chi}} + \frac{2}{u-v} , r \hat{\boldsymbol{\chi}}, r  \hat{\underline{\boldsymbol{\chi}}}, r^2 \boldsymbol{\eta}, r^2 \underline{\boldsymbol{\eta}}, r\Big(\boldsymbol{\omega} - \frac{1}{u-v}\Big), r\Big(\boldsymbol{\underline{\omega}} + \frac{1}{u-v}\Big),\\
      r \Big(\boldsymbol{\Omega} {\bf tr} {\boldsymbol{\chi}} - \boldsymbol{\Omega} {\bf tr} \underline{\boldsymbol{\chi}} - \frac{4}{u-v}\Big),  r^2 \left( \hat{\boldsymbol{\chi}} -   \hat{\underline{\boldsymbol{\chi}}}\right), r^3\left(\boldsymbol{\eta}+\boldsymbol{\underline{\eta}}\right), r^2\left(\boldsymbol{\omega} + \boldsymbol{\underline{\omega}}\right),  \\
      r^3 \boldsymbol{\alpha},  r^3\underline{\boldsymbol{\alpha}}, r^3 \boldsymbol{\beta}, r^3 \underline{\boldsymbol{\beta}}, r^3 \boldsymbol\rho, r^3 \boldsymbol\sigma \, .
      \end{aligned}
\end{align}
\end{proposition}

From the assumption that the metric is conformal to the Anti-de Sitter metric on the boundary at infinity ($4.$ of Section~\ref{sec:onepf}), one has the following boundary conditions for the null curvature components. The proof is postponed to Appendix~\ref{app:proofpropad}.
\begin{proposition}\label{prop:NullWeylBC}
  Assume that $\boldsymbol{g}$ is asymptotically Anti-de Sitter ($3.$ of Section \ref{sec:onepf}) and its conformally induced metric on $\mathcal{I}$ is conformal to the conformally induced metric of Anti-de Sitter ($4.$ of Section \ref{sec:onepf}). Then, the following boundary conditions hold
  \begin{align} 
    \lim_{v \rightarrow u} (r^\star)^{-3} \boldsymbol\alpha &= \lim_{v \rightarrow u}  (r^\star)^{-3} \underline{\boldsymbol\alpha} \, ,  \label{bc1} \\
    \lim_{v \rightarrow u} (r^\star)^{-3} \boldsymbol\beta &= -\lim_{v \rightarrow u}  (r^\star)^{-3} \underline{\boldsymbol\beta} \, , \label{bc2} \\
    \lim_{v \rightarrow u} (r^\star)^{-3} \boldsymbol\sigma &= 0, \,  \label{bc3}
  \end{align}
  as well as 
  \begin{align}
    \lim_{v \rightarrow u}  [\boldsymbol{\Omega} \boldsymbol{\slashed{\nabla}}_4 - \boldsymbol\Omega \boldsymbol{\slashed{\nabla}}_3]( (r^\star)^{-3} \boldsymbol\alpha )&= -\lim_{v \rightarrow u}    [\boldsymbol{\Omega} \boldsymbol{\slashed{\nabla}}_4 - \boldsymbol\Omega \boldsymbol{\slashed{\nabla}}_3]( (r^\star)^{-3} \underline{\boldsymbol\alpha}) \, . \label{bc4}
  \end{align} 
\end{proposition}

\subsection{The Schwarzschild-AdS background}

In this section we discuss the Schwarzschild-AdS manifold $(\mathcal{M}_{int},g=\boldsymbol{g}(0))$. In complete analogy with \cite{Daf.Hol.Rod19} we use unbolded notation to indicate $\epsilon=0$ quantities, for instance we write $\Omega$, $\slashed{g}$ and $b$ for the metric components, $\hat{\chi}$ for the outgoing shear etc.

Moreover, all the constructions and definitions of Sections \ref{sec:ta}--\ref{sec:ta3} may be repeated in unbolded notation. As this is done in detail in Section 4.3.1 of \cite{Daf.Hol.Rod19} we only give brief summary. 

\subsubsection{Ricci coefficients, curvature components} 
The only non-vanishing Ricci-coefficients in the $\epsilon=0$ case are:
\begin{align}
\Omega \chi_{AB} =  -\Omega\underline{\chi}_{AB} = \frac{\Omega^2}{r} r^2 \gamma_{AB}  \ \ \ , \ \ \ \omega = -\underline{\omega} = \frac{M}{r^2}+ k^2 \, r \, , 
\end{align}
where $\Omega^2 = 1-\frac{2M}{r} + k^2 r^2$ and $r$ is defined implicitly as in (\ref{rstardef}). 
In particular, $\Omega\chi$ and $\Omega\underline{\chi}$ are $\slashed{g}$-traceless with $\Omega tr \chi = -\Omega tr \underline{\chi}= \frac{2\Omega^2}{r}$. The only non-vanishing null-curvature component is $\rho = -\frac{2M}{r^3}$. 

\subsubsection{Differential operators and commutation formulae} 
We have the simplified coordinate formulae for the the projected Lie-derivatives for a general $S^2_{u,v}$-tensor $\xi$ of rank $N$:
\[
(D \xi)_{A_1...A_N}=\partial_v (\xi_{A_1...A_N}) \ \  \textrm{and} \ \   (\underline{D} \xi)_{A_1...A_N}=\partial_u (\xi_{A_1...A_N})
\]
For the projected covariant derivatives one finds 
\[
\Omega (\slashed{\nabla}_3 \xi)_{A_1...A_N} =  \partial_u (\xi_{A_1...A_N}) - \frac{N}{2} \Omega tr \underline{\chi}  (\xi_{A_1...A_N})  \ \  \textrm{and} \ \ \Omega (\slashed{\nabla}_4 \xi)_{A_1...A_N} =  \partial_v (\xi_{A_1...A_N}) - \frac{N}{2} \Omega tr {\chi}  (\xi_{A_1...A_N}).
\]
We recall the (unbolded) $S^2_{u,v}$-angular operators $\slashed{\nabla}$,$\slashed{\mathcal{D}}_1$, $\slashed{\mathcal{D}}_2$, $\slashed{\mathcal{D}}_1^\star$, $\slashed{\mathcal{D}}_2^\star$, now all defined with respect to the metric $\slashed{g} = r^2 \gamma$ on $S^2_{u,v}$. 


We define the vectorfield 
\[
T = \frac{1}{2} \partial_u + \frac{1}{2} \partial_v \, ,
\]
which is the static Killing field of Schwarzschild-AdS. The shall employ the notation $2\slashed{\nabla}_T\xi= \Omega \slashed{\nabla}_3 \xi +  \Omega \slashed{\nabla}_4\xi$. Note that $(\slashed{\nabla}_T \xi)_{A_1 \ldots A_N} = \partial_t (\xi_{A_1 \ldots A_n})$ since $\Omega tr \chi + \Omega tr \underline{\chi}=0$ for the background Schwarzschild-AdS metric. 

We finally collect the commutation formulae holding on $S^2_{u,v}$ tensors $\xi$:
\begin{align} \label{angularcommute}
[\Omega \slashed{\nabla}_3 , r\slashed{\nabla}_A] \xi = 0 \ \ \ , \ \ \ [\Omega \slashed{\nabla}_4 , r\slashed{\nabla}_A] \xi = 0 \ \ \ , \ \ \ [\Omega \slashed{\nabla}_3 , \Omega \slashed{\nabla}_4] \xi = 0 \, ,
\end{align}
which will be used frequently. As in \cite{Daf.Hol.Rod19} we will define the angular operators $\mathcal{A}^{[i]}$ (which commute trivially with $\slashed{\nabla}_3$, $\slashed{\nabla}_4$) acting on symmetric traceless tensors as follows:
\begin{align}
\mathcal{A}^{[0]}=1 \ \ , \textrm{and then inductively}  \ \ \mathcal{A}^{[2i+1]} = r \slashed{\mathcal{D}}_2 \mathcal{A}^{[2i]} \ \ , \ \ \mathcal{A}^{[2i]} = r^2 \slashed{\mathcal{D}}^\star_2  \slashed{\mathcal{D}}_2 \mathcal{A}^{[2i-2]} \, .
\end{align}
Elementary elliptic theory on the round sphere (see Section \ref{sec:elliptic}) establishes that these operators have trivial kernel. Consistent with the above, we will here also allow $\mathcal{A}^{[i]}$ to act on one-forms as $\mathcal{A}^{[i]} \xi = \mathcal{A}^{[i-1]} r \slashed{\mathcal{D}}^\star_2 \xi$.

\subsubsection{Norms on the spheres $S^2_{u,v}$}
Let $\theta$ and $\phi$ denote the standard spherical coordinates on $S^2_{u,v}$. We define the pointwise norm on $S^2_{u,v}$ tensors $\xi$ of rank $N$ by 
\[
|\xi|^2 = \slashed{g}^{A_1 ... A_N}  \slashed{g}^{B_1, ..., B_N} \xi_{A_1...A_N} \xi_{B_1...B_N} \, .
\]
A weighted $L^2(S^2_{u,v})$ norm on such tensors is then defined by
\begin{align} \label{spherenorm}
\|\xi\|^2_{u,v} := \int_{S^2_{u,v}} |\xi|^2 \sin \theta d\theta d\phi \, .
\end{align}
Note the absence of a factor of $r^2(u,v)$ in the integral, which if present would make (\ref{spherenorm}) the \emph{induced} norm. 

\subsubsection{The $\ell=0$ and $\ell=1$ modes}  \label{sec:l01d}
We recall the spherical harmonics $Y^\ell_m$ on the round sphere, where $\ell \in \mathbb{N}^0$ and $m \in \{-\ell, ... \ell\}$ which form a basis of $L^2(S^2)$. The $\ell=0$ and $\ell=1$ spherical harmonic will play a distinguished role in our problem and are given explicitly by
\[
Y^0_0 = \frac{1}{\sqrt{4\pi}} \ \ \ , \ \ \ Y^{1}_{0} = \sqrt{\frac{3}{4\pi}} \cos \theta   \ \ \ , \ \ \ Y^{1}_{ 1} = \sqrt{\frac{3}{4\pi}} \sin \theta \cos \phi \ \ \ , \ \ \ Y^{1}_{-1} = \sqrt{\frac{3}{4\pi}} \sin \theta \sin \phi  \, .
\]
A function is supported for $\ell \geq 1$ is a function whose spherical means vanishes. A function supported for $\ell \geq 2$ is a function supported for $\ell \geq 1$ whose projection to the $Y^{1}_m$ also vanishes.

We can also make sense of $S^2_{u,v}$ one-forms and symmetric traceless tensors being supported on specific $\ell$ modes as in Section 4.2.2 of \cite{Daf.Hol.Rod19}. To summarise this, recall that any $S^2_{u,v}$ one-form $\xi$ can be written uniquely as
\begin{align} \label{dpo}
\xi = r \slashed{\mathcal{D}}_1^\star (f,g)
\end{align}
for functions $f$ and $g$ of vanishing mean. (Note that only constants are in the kernel of $\slashed{\mathcal{D}}_1^\star$). We say that $\xi$ is supported for $\ell \geq 2$ if both $f$ and $g$ in the above representation are supported on $\ell \geq 2$. Furthermore, we define the projection of $\xi$ to $\ell=1$ by the expression (\ref{dpo}) where $f$ and $g$ are projected to $\ell=1$. 

Similarly, an $S^2_{u,v}$ symmetric traceless tensor $\xi$ can be represented uniquely by functions $f$ and $g$ supported on $\ell \geq 2$ as 
\[
\xi = r^2 \slashed{\mathcal{D}}_2^\star  \slashed{\mathcal{D}}_1^\star (f,g) \, .
\]
Note in particular that the kernel of the operator $r\slashed{\mathcal{D}}_2^\star$ consists precisely of functions supported for $\ell \geq 1$ only (as shown explicitly in \cite{Daf.Hol.Rod19}). It is in the above sense that we can say that one-forms are supported for $\ell \geq 1$ and symmetric traceless tensors for $\ell \geq 2$. 

\subsubsection{Basic elliptic estimates} \label{sec:elliptic}
We finally collect a few elliptic estimates that are immediate consequences of Section 4.4.3 in \cite{Daf.Hol.Rod19}:

\begin{proposition} 
Let $(f,g)$ be a pair of functions supported on $\ell \geq 2$, then for any $j \geq 0$
\begin{align}
\sum_{i=0}^{j+2} \| [r\slashed{\nabla}]^i (f,g)\|_{u,v}^2 \lesssim \| \mathcal{A}^{[j]} r^2 \slashed{\mathcal{D}}_2^\star  \slashed{\mathcal{D}}_1^\star (f,g) \|^2_{u,v} \, .
\end{align}
Let $\eta$ be an $S^2_{u,v}$ one-form supported on $\ell \geq 2$, then for any $j \geq 0$
\begin{align}
\sum_{i=0}^{j+1} \| [r\slashed{\nabla}]^i \eta \|_{u,v}^2 \lesssim \| \mathcal{A}^{[j]} r \slashed{\mathcal{D}}_2^\star  \eta \|^2_{u,v} \, .
\end{align}
Let $\xi$ be an $S^2_{u,v}$ symmetric traceless tensor, then for any $j \geq 0$
\begin{align}
\sum_{i=0}^{j} \| [r\slashed{\nabla}]^j \xi \|_{u,v}^2 \lesssim \| \mathcal{A}^{[j]} \xi \|^2_{u,v} \, .
\end{align}
\end{proposition}

\subsection{The linearisation procedure} \label{sec:linproc}
Recall that in our set up of Section \ref{sec:onepf}, the members of the $1$-parameter family of metrics (\ref{gindn}) all live on the same underlying manifold $\mathcal{M}$. Moreover, the hypersurfaces $u=\textrm{const.}$ and $v=\textrm{const.}$ are null for \emph{any} metric in the family. In other words, the notion of $S^2_{u,v}$-tensor is independent of $\epsilon$ and we can, in particular, add and subtract $S^2_{u,v}$-tensors associated with different $\boldsymbol{g}(\epsilon)$. If $\boldsymbol\xi$ denotes an $S^2_{u,v}$-tensor (a null-decomposed Ricci-coefficient or curvature component associated with the metric $\boldsymbol{g}(\epsilon)$) and $\xi$ denotes the corresponding tensor for $\epsilon=0$, we define its linearisation by
\begin{align} \label{linod}
\overset{(1)}{\xi} := \frac{d}{d\epsilon} \boldsymbol{\xi}|_{\epsilon=0} = \lim_{\epsilon \rightarrow 0} \frac{\boldsymbol{\xi} -\xi}{\epsilon} \, .
\end{align}
The linearised Einstein equations in double null gauge are then obtained formally in the following way: One writes down the null-decomposed Bianchi and null structure equations first for general $\epsilon$ (i.e.~for the metric $\boldsymbol{g}(\epsilon)$) and secondly for $\epsilon=0$ (i.e.~for the Schwarzschild-AdS metric) and then subtracts the respective equations, divides by $\epsilon$ and takes the limit $\epsilon \rightarrow 0$ inserting the definition (\ref{linod}). This then yields the system of gravitational perturbations in double null gauge as collected in Section \ref{sec:system}. It should be noted that for most of the equations deriving the linearisation is trivial because many of the Schwarzschild-AdS background quantities vanish, which trivialises a significant number of null structure and Bianchi equations for $\epsilon=0$. 

\subsection{The system of gravitational perturbations} \label{sec:system}
In summary, the system of gravitational perturbations in double null gauge is encoded by the linearised metric quantities
\begin{align} \label{metricq}
\glinto , \glinh_{AB}, \bmlin_A,\Olin,
\end{align}
where $\glinto=\frac{1}{2} 
\sqrt{\slashed{g}} \cdot \mathrm{tr}_{\slashed{g}} \glin$ and $ \glinh_{AB}= \glin_{AB}- \frac{1}{2} \slashed{g}_{AB} \mathrm{tr}_{\slashed{g}} \glin$ are defined from the linearised metric $\glin$, which is in turn defined by (\ref{linod}),  the linearised connection coefficients
\begin{align} \label{connectionq}
\otx,\otxb,\olin,\olinb,\elin_A,\eblin_A, \xlin_{AB},\xblin_{AB} ,
\end{align}
and the linearised curvature components
\begin{align} \label{curvatureq}
\rlin,\slin,\blin_A,\bblin_A, \alin_{AB},\ablin_{AB}.
\end{align}
Depending on the number of indices, the above quantities are $S^2_{u,v}$ scalars, one-forms and symmetric traceless tensors respectively. 
As in~\cite{Daf.Hol.Rod19}, we will speak of a solution $\mathscr{S}$ to the system of gravitational perturbations to mean a a collection of quantities 
\begin{align} \label{scollect}
\mathscr{S}=\left(\, \glinh \, , \, \glinto \, , \, \Olino \, , \,  \bmlin\, , \,  \otx \, , \,  \otxb\, , \,  \xlin\, , \, \xblin\, , \,  \eblin \, , \,  \elin \, , \, \olin \, , \,  \olinb \, , \,  \alin \, , \,  \blin \, , \,  \rlin \, , \,  \slin \, , \,  \bblin \, , \,  \ablin \, , \, \Klin \right)
\end{align}
satisfying the system (\ref{stos})--(\ref{Bianchi10}) below, which we call the system of linearised gravity on the Schwarzschild background. Finally, it follows just as in Section 5.1.3 of \cite{Daf.Hol.Rod19} (from the $1$-parameter family of metrics being smooth in the extended sense) that the following linearised quantities extend smoothly to the horizon:
\begin{align} \label{smoothextended}
\big(\, \glinh , \, \glinto , \, \Olin , \,  \bmlin \Omega^{-2} , \,  \otx \, , \,  \Omega^{-2} \otxb\, , \,  \Omega\xlin\, , \, \Omega^{-1} \xblin\, , \,  \eblin \, , \,  \elin \, , \, \olin \, , \,  \Omega^{-2} \olinb \, , \,  \Omega^2 \alin \, , \,  \Omega \blin \, , \,  \rlin  , \,  \slin \, , \, \Omega^{-1} \bblin \, , \,  \Omega^{-2}\ablin \, , \, \Klin \big)  . 
\end{align}
The $\Omega^{-2}$ weight for $\bmlin$ does not appear in \cite{Daf.Hol.Rod19} as in that paper, the shift satisfies an equation in the ingoing direction. See footnote \ref{footnoteb}.

\subsubsection{Equations for the linearised metric components}
\label{=lmc}
%
The following equations hold for the linearised metric components, $\glinto \, , \, \glinh \, , \, \bmlin \, , \, \Olin$:
\begin{align} \label{stos} 
\partial_u \Big(\frac{\glinto}{\sqrt{\slashed{g}}}\Big)  = \otxb -  \divs\, \bmlin
\qquad , \qquad 
\partial_v \Big(\frac{\glinto}{\sqrt{\slashed{g}}}\Big) = \otx   \, ,
\end{align}
\begin{align} \label{stos2}
\Omega \slashed{\nabla}_3  \, \glinh   =2\Omega\, \xblin + 2\slashed{\mathcal{D}}_2^\star \bmlin 
\ \ \ ,  \ \ \
\Omega \slashed{\nabla}_4   \, \glinh &=2\Omega\, \xlin   ,
\end{align}
\begin{align} \label{bequat} 
\partial_v \bmlin^A = -2 \Omega^2\left(\elin^A - \eblin^A\right) \, ,
\end{align}
\begin{align} \label{oml3}
\partial_v \left( \Olin \right) = \olin \textrm{ \ , \ }  \partial_u \left(\Olin\right)=\olinb \textrm{ \  , \ }   2 \slashed{\nabla}_A \left(\Olin\right) = \elin_A + \eblin_A.
\end{align}

\subsubsection{Equations for the linearised Ricci coefficients}
\label{=lRc}
For $\otx \, , \, \otxb$ we have the equations
\begin{align} \label{dtcb}
\partial_v \otxb  = \Omega^2 \left( 2 \divs\, \eblin + 2\rlin + 4 \left(\rho-2k^2 \right) \, \Olin \right) - \frac{1}{2}  \Omega tr \chi \Big( \otxb - \otx  \Big) ,
\end{align}
\begin{align} \label{dbtc}
\partial_u \otx  = \Omega^2 \left( 2 \divs\, {\elin} + 2 \rlin + 4 \left(\rho-2k^2\right) \, \Olin \right) - \frac{1}{2}  \Omega tr \chi \Big( \otxb - \otx  \Big) ,
\end{align}
\begin{align} \label{uray}
\partial_v \otx = - \left(\Omega tr \chi\right)\otx + 2 \omega \otx  + 2  \left(\Omega tr \chi \right) \olin ,
\end{align}
\begin{align} \label{vray}
\partial_u \otxb = - \left(\Omega tr \underline{\chi}\right) \otxb  + 2 \underline{\omega} \otxb + 2  \left(\Omega tr \underline{\chi} \right) \olinb ,
\end{align}
while for $\xlin \, , \, \xblin$ we have
\begin{equation} \label{tchi} 
\begin{split}
\slashed{\nabla}_3  \Big(\Omega^{-1} \xblin  \Big)  +  \Omega^{-1} \left(tr \underline{\chi}\right) \xblin = -\Omega^{-1} \ablin \, , \\
\slashed{\nabla}_4  \Big(\Omega^{-1} \xlin \Big)  +  \Omega^{-1} \left(tr{\chi}\right) \xlin= -\Omega^{-1} \alin   \, ,
\end{split}
\end{equation}
\begin{align} \label{chih3}
\slashed{\nabla}_3  \left(\Omega \xlin \right)  + \frac{1}{2} \left(\Omega tr \underline{\chi}\right) \xlin + \frac{1}{2} \left( \Omega tr \chi\right) \xblin  &= -2 \Omega \slashed{\mathcal{D}}_2^\star \elin \, , \\
\slashed{\nabla}_4  \left(\Omega \xblin  \right) + \frac{1}{2} \left(\Omega tr \chi \right) \xblin  + \frac{1}{2} \left( \Omega tr \underline{\chi}\right) \xlin &= -2 \Omega \slashed{\mathcal{D}}_2^\star \eblin \, . \label{chih3b}
\end{align}
We also have the (purely elliptic) linearised Codazzi equations on the spheres $S^2_{u,v}$, which read
\begin{equation}
\begin{split}
\slashed{div} \xblin = -\frac{1}{2} \left(tr \underline{\chi}\right)  \elin + \bblin + \frac{1}{2\Omega} \slashed{\nabla}_A \otxb , \\
\slashed{div} \xlin = -\frac{1}{2} \left( tr {\chi}\right) \eblin  -\blin + \frac{1}{2\Omega} \slashed{\nabla}_A \otx \label{ellipchi} \, .
\end{split}
\end{equation}
For $\elin$ and $\eblin$ we have the transport equations
\begin{align} \label{propeta}
\slashed{\nabla}_3 \eblin =  \frac{1}{2} \left(tr \underline{\chi}\right) \left( \elin - \eblin\right)  + \bblin
\textrm{ \ \ \ \ , \ \ \ \ }
\slashed{\nabla}_4 \elin =  -  \frac{1}{2} \left( tr {\chi}\right) \left( \elin - \eblin\right) - \blin ,
\end{align}
together with the elliptic equations on the spheres $S^2_{u,v}$
\begin{align} \label{curleta}
\slashed{curl} \elin = \slin \ \ \ , \ \ \ \slashed{curl} \eblin = -\slin \, .
\end{align}
We finally have the transport equations for $\olin$ and $\olinb$ 
\begin{align} \label{oml1}
\partial_v \olinb = -\Omega^2 \left(\rlin + 2 \left(\rho+k^2\right) \Olin \right) \, ,
\end{align}
\begin{align} \label{oml2}
\partial_u \olin = -\Omega^2 \left(\rlin + 2 \left(\rho+k^2\right)\Olin \right) \, ,
\end{align}
and the linearised Gauss equation on the spheres $S^2_{u,v}$, which reads
\begin{equation} \label{lingauss}
\Klin = -\rlin - \frac{1}{4} \frac{tr {\chi}}{\Omega}\left( \otxb - \otx  \right) +\frac{1}{2}\Olin \left(tr \chi tr \underline{\chi} \right) \, .
\end{equation}
We also note that $\Klin$, the linearised Gauss curvature of the double null spheres satisfies (see (221) of~\cite{Daf.Hol.Rod19})
\begin{align} \label{gaussfootnote}
2\Klin=-\frac{1}{2}\ds \mathrm{tr}_{\slashed{g}} \glin +\divs \divs\glinh - \frac{1}{r^2}\mathrm{tr}_{\slashed{g}} \glin \ \ \ \  \textrm{where $\glinto=\frac{1}{2} 
\sqrt{\slashed{g}} \cdot \mathrm{tr}_{\slashed{g}} \glin$.}
\end{align}

\subsubsection{Equations for linearised curvature components}
\label{=lcc}
We finally collect the equations satisfied by the linearised curvature components, which arise from the linearisation of the Bianchi equations:
\begin{align}
\slashed{\nabla}_3 \alin + \frac{1}{2} tr \underline{\chi}\alin + 2 \Omega^{-1} \underline{{\omega}} \alin &= -2 \slashed{\mathcal{D}}_2^\star \blin - 3 \rho\, \xlin \, ,  \label{Bianchi1} \\
\slashed{\nabla}_4 \blin + 2 (tr \chi) \blin -  \Omega^{-1} {\omega} \blin &= \slashed{div}\, \alin \, , \label{Bianchi2} \\
\slashed{\nabla}_3 \blin + (tr \underline{\chi}) \blin + \Omega^{-1} \underline{{\omega}} \blin &= \slashed{\mathcal{D}}_1^\star \left(-\rlin \, , \, \slin \, \right) + 3\rho \, \elin \, ,   \label{Bianchi3}
\\
\slashed{\nabla}_4 \rlin + \frac{3}{2} (tr \chi) \rlin &= \slashed{div}\, \blin - \frac{3}{2} \frac{\rho}{\Omega}  \otx \, , \label{Bianchi4}
\\
\slashed{\nabla}_3 \rlin + \frac{3}{2} (tr \underline{\chi}) \rlin &= -\slashed{div}\, \bblin - \frac{3}{2} \frac{\rho}{\Omega} \otxb \, , \label{Bianchi5}
\\
\slashed{\nabla}_4 \slin + \frac{3}{2} (tr \chi) \slin&= -\slashed{curl}\, \blin \, , \label{Bianchi6} \\
\slashed{\nabla}_3 \slin + \frac{3}{2} (tr \underline{\chi}) \slin &= -\slashed{curl}\, \bblin \, , \label{Bianchi7} \\
\slashed{\nabla}_4 \bblin + (tr \chi)  \bblin + \Omega^{-1} {\omega} \bblin &= \slashed{\mathcal{D}}_1^\star \left(\rlin \, ,  \, \slin \, \right) - 3 \rho\, \eblin  \, ,  \label{Bianchi8} \\
\slashed{\nabla}_3 \bblin + 2 (tr \underline{\chi})  \bblin - \Omega^{-1} {\underline{\omega}} \bblin &= - \slashed{div}\, \ablin \, , \label{Bianchi9} \\
\slashed{\nabla}_4 \ablin + \frac{1}{2} (tr \chi) \ablin + 2 \Omega^{-1} {\omega} \ablin &=  2 \slashed{\mathcal{D}}_2^\star \bblin - 3 \rho\, \xblin \, .  \label{Bianchi10}
\end{align}

\subsubsection{Projections to the $\ell=0$ and $\ell=1$ modes}

Suppose we are given a smooth solution $\mathscr{S}$ of the above system of gravitational perturbations. Then we may project all quantities of $\mathscr{S}$ (see (\ref{scollect})) to $\ell=0$ and $\ell=1$ respectively (as defined in Section \ref{sec:l01d}), thereby obtaining a collection of quantities denoted by $\mathscr{S}_{\ell=0}$ and $\mathscr{S}_{\ell=1}$ respectively. One now readily checks that $\mathscr{S}_{\ell=0}$ and $\mathscr{S}_{\ell=1}$ solve the system of gravitational perturbations individually.\footnote{More abstractly, this is a consequence of the spherical symmetry of the background (in particular projection operators commuting with $\slashed{\nabla}_3$ and $\slashed{\nabla}_4$) and the linearity of the equations.} We can therefore decompose
\[
\mathscr{S} = \mathscr{S}_{\ell=0} + \mathscr{S}_{\ell=1} +  \mathscr{S}_{\ell\geq 2} \, ,
\]
with the last term defined by the equation. This decomposition will later allow us to deal with the $\ell=0,1$ part of the solution independently (as far as initial data and boundary conditions are concerned), which will turn out to be convenient, as the $\ell=0,1$ part of the solution can be computed (more or less) explicitly.
\subsection{Boundary conditions for the system of gravitational perturbations}
Recall that the boundary at infinity, $\mathcal{I}$, is not part of our interior manifold $\mathcal{M}_{int}$. On the other hand, to formulate boundary conditions (on certain weighted quantities of $\mathscr{S}$ in (\ref{scollect})) we will need to consider $S^2_{u,v}$ tensors $\xi$ on the Schwarzschild-AdS manifold $(\mathcal{M}_{int},g)$, which extend smoothly to $\mathcal{I}$, i.e.~to the larger manifold $\mathcal{M}$. To keep notation clean, we will often simply write $\xi (u,u,\theta,\phi)$ or $\|\xi\|_{u_0,u_0}$ to denote the appropriate limit of such tensors on $\mathcal{I}$. Recall in this context from (\ref{spherenorm}) that the $\|\cdot\|_{u,v}$ norm is independent of the radius of the sphere $S^2_{u,v}$.

The boundary conditions for the non-linear spacetime null-curvature components (see Proposition~\ref{prop:NullWeylBC}) can easily be linearised,\footnote{Equations~\eqref{bc1}-\eqref{bc4} are all trivial to linearise since the quantities $\boldsymbol\alpha, \boldsymbol\beta, \underline{\boldsymbol\alpha}, \underline{\boldsymbol\beta}, \boldsymbol\sigma$ all \emph{vanish} for the background Schwarzschild-AdS metric.}  leading to the following definition.
\begin{definition} \label{def:bc}
We will say that a smooth solution $\mathscr{S}$ of the system of gravitational perturbations satisfies conformal boundary conditions provided we have for any $u\geq u_0=0$ the limits
\begin{align}
\lim_{v \rightarrow u} r^3 \alin &= \lim_{v \rightarrow u} r^{3} \ablin  \, , \label{bcl1} \\
\lim_{v \rightarrow u} [\Omega \slashed{\nabla}_4 - \Omega\slashed{\nabla}_3]( r^3 \alin )&= -\lim_{v \rightarrow u}   [\Omega \slashed{\nabla}_4 - \Omega\slashed{\nabla}_3] ( r^3 \ablin) \label{bcl1b} \, , \\
\lim_{v \rightarrow u} r^3 \blin &= -\lim_{v \rightarrow u}  r^3 \bblin \, , \label{betaboundary} \\
\lim_{v \rightarrow u} r^3 \slin &= 0 \label{bcl4} \, .
\end{align}
Here the tensorial limits are to be understood componentwise in an orthonormal frame on the spheres $S^2_{u,v}$. 
\end{definition}

\begin{remark}
Note that if the solution $\mathscr{S}$ satisfies conformal boundary conditions then also
\begin{align} \label{bclxi}
\lim_{v \rightarrow u} r \xlin = \lim_{v \rightarrow u} r \xblin  \, .
\end{align}
using the linearised Bianchi and null-structure equations. The bound (\ref{bclxi}) could of course also be deduced directly  from the fact that $\lim_{v \rightarrow u} r \hat{\boldsymbol{\chi}}= \lim_{v \rightarrow u} r \underline{\hat{\boldsymbol{\chi}}}$ holds by Proposition \ref{prop:ad} and trivially linearises to (\ref{bclxi}). 
\end{remark}

We close the section with one more definition, which translates the asymptotic behaviour of the non-linear geometric quantities collected
in Proposition \ref{prop:ad} to the linearised setting. 
\begin{definition} \label{def:aAdSlin}
We will say that a smooth solution $\mathscr{S}$ of the system of gravitational perturbations is \underline{asymptotically AdS in the linearised sense} if the following quantities as well as arbitrary many derivatives from the set $\{ r^2 \Omega^{-1}\slashed{\nabla}_3, \Omega \slashed{\nabla}_4, [r \slashed{\nabla}]\}$ extend to the conformal boundary $\mathcal{I}$:
\begin{align}
 \glinh \, , \, \frac{\glinto}{\sqrt{\slashed{g}}} \, , \, r^2\frac{\Olino}{\Omega} \, , \,  r \Omega^{-2} \bmlin\, , \,  \otx \, , \,  r^2 \Omega^{-2} \otxb\, , \,  \Omega \xlin\, , \, r^2\Omega^{-1} \xblin\, , \,  r^2 \eblin \, , \,  r^2 \elin \, , \, r\olin \, , \,  \Omega^{-2}r^{3} \olinb \, , \, r^2 \Klin   \label{lin1} 
\end{align}
and
\begin{align} \label{lin2}
r \Omega^2 \alin \, , \,  r^2 \Omega \blin \, , \,  r^3\rlin \, , \,  r^3 \slin \, , \, r^4 \Omega^{-1} \bblin \, , \,  r^5\Omega^{-2}\ablin .
\end{align}

\end{definition}


\begin{remark} \label{rem:betterbounds}
The bounds on (\ref{lin1})--(\ref{lin2}) should be thought as having been derived by linearising the non-linear statement in (\ref{boundsextend}). In fact, from the bounds on (\ref{lin1})--(\ref{lin2}) we can (and will) also deduce bounds for the difference quantities $r^2 \xlin - r^2 \xblin$, $r^2 \olin + r^2 \olinb$ and $r \otx -  r^2\otxb$ consistent with (\ref{boundsextend}) later in the paper, see (\ref{sum1}) and Remarks \ref{rem:soc} and \ref{rem:toc}. For simplicity, we have not included them in the above definition. 
\end{remark}

\subsection{Special solutions}
\subsubsection{Pure gauge solutions}
There are special solutions to the system of gravitational perturbations (\ref{stos})--(\ref{Bianchi10})  corresponding to infinitesimal coordinate transformation that generate a change of double null gauge (i.e.~a choice of nearby sphere and corresponding foliations of the associated ingoing and outgoing cone). In complete analogy to \cite{Daf.Hol.Rod19} we call these \emph{pure gauge solutions}. In our setting, the additional requirement that the pure gauge solutions should preserve the boundary conditions reduces the admissible pure gauge solutions and they can in fact be parametrised by a single scalar function.

\begin{lemma} \label{lem:exactsol} 
Given an arbitrary smooth function $f : \mathbb{R}_0^+ \times S^2 \rightarrow \mathbb{R}$, the corresponding functions $f_u = f(u, \theta, \phi)$ and $f_v=f(v,\theta,\phi)$, interpreted as functions on $\mathcal{M} \setminus \mathcal{H}^+$  independent of one of the coordinates, generate the following (pure gauge) solution of the system of gravitational perturbations on $\mathcal{M} \setminus \mathcal{H}^+$:
\begin{align}
\frac{\Olino}{\Omega} &= \frac{1}{2\Omega^2} \partial_v \left(\Omega^2f_v \right)+\frac{1}{2\Omega^2}\partial_u\big(\Omega^2f_u\big)  , & \glinh&= +\frac{4}{r} r^2 \slashed{\mathcal{D}}_2^\star \slashed{\nabla} {f_u}  \nonumber \\
 \frac{\glinto}{\sqrt{\slashed{g}}} &= \frac{2\Omega^2 ({f_v-f_u})}{r} -\frac{2}{r} r^2 \slashed{\Delta}{f_u}  , 
 &\bmlin&= +2r^2 \slashed{\nabla} \left[ {\partial_u} \left(\frac{{f_u}}{r}\right)\right]+2\Omega^2\ns {f_v} \nonumber \\
\elin &= \frac{\Omega^2}{r^2} [r \slashed{\nabla}] f_v +\frac{r}{\Omega^2}\ns\Big[\partial_u\Big(\frac{\Omega^2}{r}f_u\Big)\Big]  , & \eblin &= \frac{r}{\Omega^2}  \slashed{\nabla} \left[\partial_v \left(\frac{\Omega^2}{r}f_v\right) \right] -\frac{\Omega^2}{r}\ns f_u  \nonumber \, ,
\nonumber \\
\xlin &= -2\frac{\Omega}{r^2} r^2 \slashed{\mathcal{D}}_2^\star (\slashed{\nabla} f_u)   , & \otx &= 2 \partial_v \left(\frac{f_v \Omega^2}{r}\right)+\frac{2\Omega^2}{r^2}\Big[\Big(1-\frac{4M}{r} - k^2 r^2  \Big)f_u+\Delta_{S^2}f_u\Big],     \nonumber \\
 \xblin &= -2\frac{\Omega}{r^2} r^2 \slashed{\mathcal{D}}_2^\star (\slashed{\nabla} f_v) ,& \otxb &=  \frac{2\Omega^2}{r^2} \left[\Delta_{\mathbb{S}^2} f_v +  \Big(1-\frac{4M}{r} -k^2 r^2 \Big)f_v \right] -\partial_u\Big(\frac{2\Omega^2}{r}f_u\Big)  , \nonumber \\
 \blin &=-\frac{6M\Omega}{r^4}  [r \slashed{\nabla} ]f_u , & \bblin&= \frac{6M\Omega}{r^4}  r \slashed{\nabla} f_v , \nonumber \\
\rlin &= \frac{6M \Omega^2}{r^4} (f_v-f_u)   , & \Klin &= -\frac{\Omega^2}{r^3}\left(\Delta_{\mathbb{S}^2} (f_v-f_u) + 2(f_v-f_u)\right) \nonumber
\end{align}
and
\[
 \alin = \ablin = 0  \ \ \ , \ \ \ \slin = 0 \nonumber \, .
\]
The solution satisfies the conformal boundary conditions of Definition \ref{def:bc}. 
We will call $f$ a gauge function and denote the corresponding pure gauge solution by $\mathscr{G}_f$. Finally, if $\Omega^2(u,v) f(u,\theta,\phi)$ extends smoothly to $\mathcal{H}^+$ then so does the associated pure gauge solution.\footnote{In particular, the quantities (\ref{smoothextended}) extend smoothly to $\mathcal{H}^+$.}
\end{lemma}

\begin{proof}
This is verified exactly as in \cite{Daf.Hol.Rod19} by direct computation. Since $f_v - f_u$ vanishes on $\mathcal{I}$, the linearised boundary conditions (\ref{bcl1})--(\ref{bcl4}) are indeed satisfied.
\end{proof}

There is a further pure gauge solution which only changes the linearised metric quantities but leaves linearised Ricci-coefficients and curvature components invariant:

\begin{lemma} \label{lem:puregaugemetric}
For any smooth functions $q_1(u,\theta,\phi)$ and $q_2(u,\theta,\phi)$ the following is a pure gauge solution of the system of gravitational perturbations 
\begin{align}
\Olin = 0 \ \ \ , \ \ \ \glinh= 2 r^2 \slashed{\mathcal{D}}_2^\star \slashed{\mathcal{D}}_1^\star(q_1,q_2) \ \ , \ \ \ \   
 \frac{\glinto}{\sqrt{\slashed{g}}} = r^2 \slashed{\Delta} q_1    \ \ , \ \ \ \  \bmlin= r^2 \slashed{\mathcal{D}}_1^\star (\partial_u q_1, \partial_u q_2) \, , 
\end{align}
while all linearised Ricci and null curvature components vanish identically. We denote the solution by $\mathscr{G}_q$.
\end{lemma}

\subsubsection{The family of linearised Kerr-AdS solutions}
It is well-known that the Schwarzschild-AdS family sits as a $1$-parameter family in the larger $2$-parameter family of Kerr-AdS metrics. At the linear level there exists (due to the spherical symmetry of the background) a $3$-dimensional (choosing an axis and a magnitude) family of explicit solutions that move the Schwarzschild-AdS metric to a nearby Kerr-AdS metric. Moreover, there is also the $1$-parameter family of changing the mass. We summarise both in Lemma \ref{lem:kerr} below. Let us already remark that for the $\ell=0$ modes, the pure gauge solution takes a significantly more complicated form compared to the asymptotically flat case. The underlying reason is that the variable $r$ (defined in terms of the (fixed) Eddington-Finkelstein coordinates $(u,v)$) by (\ref{rstardef}) depends implicitly on the mass, with a dependence that is more involved than in the asymptotically flat case. See Appendix \ref{sec:l0} for computational details regarding the $\ell=0$ mode. 

\begin{lemma} \label{lem:kerr}
For any $\mathfrak{a} \in \mathbb{R}$ and $m \in \{-1,0,1\}$ the following linearised metric quantities generate a smooth solution of the system of gravitational perturbations on $\mathcal{M}$:
\begin{align} \label{kui}
\Olin = 0 \ \ \ , \ \ \ \glinh = 0 \ \ \ , \ \ \ \glinto = 0 \ \ \ ,  \ \ \ \bmlin^A = \left(b^{KAdS,m}\right)^A = -2 \left(\frac{2M}{r} -k^2 r^2\right) \mathfrak{a} \slashed{\epsilon}^{AB} \partial_B Y^{\ell=1}_m \, . 
\end{align}
The solution has the following non-vanishing Ricci-coefficients and curvature components:
\begin{align}
\elin^A = -\eblin^A = \left(\eta^{KAdS,m}\right)^A := \frac{3M\mathfrak{a}}{r^2} \slashed{\epsilon}^{AB} \partial_B Y^{\ell=1}_m  \ \ \ , \ \ \ \blin = -\bblin = \frac{\Omega}{r} \eta^{KAdS, m} \ \ \ , \ \ \ \slin = \frac{6}{r^4} \mathfrak{a} M Y^{\ell=1}_m \, .
\end{align}
Moreover for any $\mathfrak{m} \in \mathbb{R}$, the following linearised metric quantities generate a (spherically symmetric) smooth solution of the system of gravitational perturbations on $\mathcal{M}$ (where we have set $l^2 = k^{-2} = -\frac{1}{3 \Lambda}$)
\begin{align} \label{glins}
 \glinh &= 0 \, , &\qquad \frac{\glinto}{\sqrt{\slashed{g}}} &=\mathfrak{m}  \left(
 -\frac{2M}{r \left(1+\frac{3r^2}{l^2}\right)} +\frac{M\Omega^2}{r} \int_r^\infty \frac{2}{\Omega^2} \left(\frac{l^2(l^2-3r_+^2)}{(l^2+3r_+^2)^2} \frac{1}{r_+} -  \frac{l^2(l^2-3\tilde{r}^2)}{(l^2+3\tilde{r}^2)^2}\frac{1}{ \tilde{r}}\right) d\tilde{r} \right) \, ,
 \\
\bmlin &= 0 \, ,  &\qquad \Olin &= \frac{1}{4} \left( \frac{2\mathfrak{m} M}{r \Omega^2} - 4\mathfrak{m} M \frac{l^2(l^2-3r_+^2)}{r_+(l^2+3r_+^2)^2}\right) +\frac{1}{4} \left(\frac{2M}{r^2} + \frac{2r}{l^2} \right) \frac{\glinto}{\sqrt{\slashed{g}}}\frac{ r}{\Omega^2}  \, . \label{omlins}
\end{align}
In particular,\footnote{For convenience we collect the values of all Ricci-coefficients and curvature components in Appendix \ref{sec:l0}. 
} this solution satisfies
\begin{align} \label{gaugedep0}
r^3 \rlin - 3M \frac{\glinto}{\sqrt{\slashed{g}}} = \mathfrak{m} \cdot M \, .
\end{align}
We call the first type of solution a linearised Kerr-AdS solution with fixed mass and the second solution a linearised Schwarzschild-AdS solution. Together these solutions form a $4$-parameter family of linearised Kerr-AdS solutions. Given parameters $(\mathfrak{m}, \mathfrak{a}_{-1}, \mathfrak{a}_0,\mathfrak{a}_1)$ we denote by $\mathscr{K}_{m,\vec{\mathfrak{a}}}$ the sum of the four corresponding linearised Kerr-AdS solutions. 
\end{lemma}

\begin{remark}
Note in particular that the quantity on the left hand side of (\ref{gaugedep0}) is gauge invariant: Any pure gauge solution supported on $\ell=0$ leaves the quantity invariant. This follows directly from Lemma \ref{lem:exactsol}.
\end{remark}

\begin{remark}
The solution (\ref{kui}) actually remains a solution without the $k^2 r^2$ term in the definition of $\bmlin$ as this term corresponds to a pure gauge solution from Lemma \ref{lem:puregaugemetric} (with $q_1=0$ and $q_2$ supported on $\ell=1$).
\end{remark}

\subsubsection{Regularity at the horizon and at infinity}
We close this section noting that our special solutions are regular at the horizon and asymptotically AdS in the linearised sense:
\begin{proposition}
The following are smooth solutions of the system of gravitational perturbations, which are moreover asymptotically AdS in the linearised sense (Definition \ref{def:aAdSlin}), satisfy the boundary conditions (\ref{bcl1})--(\ref{bcl4}) and are such that the quantities (\ref{lin1}), (\ref{lin2}) extend smoothly to $\mathcal{H}^+$. 
\begin{itemize}
\item The four-parameter family of linearised Kerr-AdS solutions of Lemma \ref{lem:kerr}.
\item The pure gauge solutions of Lemma \ref{lem:exactsol}, provided the function $f(u)$ generating them is smooth on $\mathbb{R}^+$ and $f(u) \frac{\Omega^2(u,v)}{r^2(u,v)}$ is smooth in the extended sense at $\mathcal{H}^+$.
\item The pure gauge solutions of Lemma~\ref{lem:puregaugemetric}, provided the functions $q_1(u),q_2(u)$ generating them are smooth on $\mathbb{R}^+$ and $q_1(u), q_2(u)$ are smooth in the extended sense at $\mathcal{H}^+$.
\end{itemize}
\end{proposition}
\begin{proof}
Both linearised Kerr-AdS and pure gauge solutions are by construction smooth solutions to the system of gravitational perturbations. The validity of the boundary conditions is straightforward to check for the linearised Kerr-AdS solutions (note in particular $\slin \sim r^{-4}$). For the pure gauge solutions we only need to check the boundary condition for $r^3 \blin + r^3\bblin \rightarrow 0$ which follows from $[r\slashed{\nabla}]( f_u(u)-f_v(v))$ vanishing at the conformal boundary (where $v=u$). It remains to check the extension of the quantities (\ref{lin1}) and (\ref{lin2}) towards the horizon and infinity. Near the horizon the regularity claims are easily verified. Near infinity, the claim on $\Olin$ for linearised Kerr-AdS solution follows by carefully expanding the integrand in (\ref{glins}) and observing a cancellation of the leading order term in (\ref{omlins}). For the pure gauge solution, the hardest to check is that $r^2 \Olin$ extends to the boundary. For this we Taylor expand the expression of Lemma \ref{lem:exactsol} 
\begin{align}
\Olin &= \frac{1}{2} (\partial_v (f_v) + \partial_u (f_u)) + \left(k^2 r + \frac{M}{r^2}\right)(f_v- f_u) \nonumber \\
&= f^\prime(v) + \frac{ f^{\prime \prime}(v)}{2} (v-u) + \left(k^2 r + \frac{M}{r^2}\right)\left(-f^\prime(v) (u-v) + \frac{f^{\prime \prime}(v)}{2} (v-u)^2 \right) + \mathcal{O}(r^{-2}) =  \mathcal{O}(r^{-2}) , \nonumber
\end{align}
where we have recalled $f_v=f(v)$ and $f_u=f(u)$ and used (\ref{rrstarrel}) in conjunction with $u-v \geq 0$. The conditions on higher derivatives implicit in Definition \ref{def:aAdSlin} are straightforward to check. The last item is immediate and and this finishes the proof of the proposition.
\end{proof}

\subsection{The Teukolsky equations} \label{sec:teuk}

Remarkably, the extremal linearised curvature components $\alin$ and $\ablin$, which by Lemma \ref{lem:exactsol} and~\ref{lem:puregaugemetric} vanish identically for pure gauge solutions, satisfy decoupled equations. These are the well-known Teukolsky equations. We now derive these equations and define the associated gauge invariant hierarchies $(\alin, \plin, \Psilin)$ and $(\ablin,\pblin,\Psilinb)$. See our companion paper \cite{Gra.Hola} and Section \ref{sec:teukolsky} below for analytic results on the Teukolsky system.

\subsubsection{Derivation of the equations}
\begin{lemma}
Given a smooth solution $\mathscr{S}$ of the system of gravitational perturbations 
the quantities $\alin$ and $\ablin$ satisfy the decoupled Teukolsky equations:
\begin{align}
(\Omega \slashed{\nabla}_4)(\Omega \slashed{\nabla}_3)  (r \Omega^2 \alin) - 2  \left(-\frac{2}{r} + \frac{6M}{r^2}\right)\Omega \slashed{\nabla}_3  (r \Omega^2 \alin) - \frac{\Omega^2}{r^2} \left(-2 r^2 \slashed{\mathcal{D}}_2^\star  \slashed{div} \Omega^2 r \alin - \frac{6M}{r}  r \Omega^2 \alin \right)  &= 0 \, ,   \label{Teuk1} \\
(\Omega \slashed{\nabla}_3)(\Omega \slashed{\nabla}_4) (r\Omega^2 \ablin) + 2 \left(-\frac{2}{r} + \frac{6M}{r^2}\right)\Omega \slashed{\nabla}_4  (r \Omega^2 \ablin) - \frac{\Omega^2}{r^2} \left(-2 r^2 \slashed{\mathcal{D}}_2^\star  \slashed{div} \Omega^2 r \alin - \frac{6M}{r}  r \Omega^2 \ablin \right)  &= 0  \, . \label{Teuk2}
\end{align}
\end{lemma}

\begin{proof}
We give the argument for $\alin$, the one for $\ablin$ being entirely analogous. Write (\ref{Bianchi1}) as
\begin{align}
(\Omega \slashed{\nabla}_3) (r \Omega^2 \alin) = \frac{\Omega^4}{r^4} \left(-2 r \slashed{\mathcal{D}}_2^\star (\Omega^{-1} \blin r^4) +6M \Omega^{-1} r^2 \xlin \right) \, .
\end{align}
Apply $\Omega \slashed{\nabla}_4$, commute derivatives using (\ref{angularcommute}) and insert (\ref{Bianchi1}) and (\ref{tchi}) to deduce \begin{align}
(\Omega \slashed{\nabla}_4)(\Omega \slashed{\nabla}_3)  (r \Omega^2 \alin) = 2 \left(-\frac{2}{r} + \frac{6M}{r^2}\right)\Omega \slashed{\nabla}_3 (r \Omega^2 \alin) + \frac{\Omega^4}{r^4} \left(-2 r^5 \slashed{\mathcal{D}}_2^\star \slashed{div} \alin - 6M r^2 \alin \right) \, .
\end{align}
\end{proof}

\subsubsection{The gauge invariant hierarchy} \label{sec:gih}
Given symmetric traceless tensors $\alin$, $\ablin$ we define (consistently with \cite{Daf.Hol.Rod19}) the quantities $(\plin, \pblin)$ by:
\begin{align}
-2r^3 \Omega \plin := \frac{r^2}{\Omega^2} \Omega\slashed{\nabla}_3 (r \Omega^2 \alin) &= - 2 r \slashed{\mathcal{D}}_2^\star \left(\Omega r^2 \blin\right) - 3 \rho r^3 \Omega \xlin  \, , \label{rewriteBianchi} \\
2r^3 \Omega \pblin :=\frac{r^2}{\Omega^2} \Omega\slashed{\nabla}_4 (r \Omega^2 \ablin) &= 2 r \slashed{\mathcal{D}}_2^\star \left(\Omega r^2 \bblin\right) - 3 \rho r^3 \Omega \xblin .\label{rewriteBianchibar}
\end{align}
In the above, the second equality follows from a rewriting of the Bianchi equations (\ref{Bianchi1}) and (\ref{Bianchi10}) respectively. Note that we can also rewrite (\ref{rewriteBianchibar}) in a form where all terms extend smoothly to $\mathcal{H}^+$. 
\begin{align}
\Omega\slashed{\nabla}_4 (r^5 \Omega^{-2} \ablin) - 4 r^5 \Omega^{-2} \ablin \left(\frac{1}{r} - \frac{3M}{r^2}\right)&=  2 r^5 \slashed{\mathcal{D}}_2^\star \Omega^{-1} \bblin - 3 \rho r^5 \Omega^{-1} \xblin \, . \label{rewriteBianchi2} 
\end{align}
Again consistent with \cite{Daf.Hol.Rod19} we define also the higher order quantities $(\Psilin, \Psilinb)$ by
\begin{align}
\Psilin := \frac{r^2}{\Omega^2}  \Omega \slashed{\nabla}_3 (r^3 \Omega \plin)= r^2 \Omega^{-1} \slashed{\nabla}_3 \left(-\frac{1}{2}r^2 \Omega^{-1} \slashed{\nabla}_3 (r \Omega^2 \alin)\right) &= r^5 \slashed{\mathcal{D}}_2^\star  \slashed{\mathcal{D}}_1^\star \left(-\rlin, \slin\right) + \frac{3}{4} r^5 \rho tr \chi (\xlin - \xblin)   \label{Prel}\, , \\
-\Psilinb := + \frac{r^2}{\Omega^2}  \Omega \slashed{\nabla}_4 (r^3 \Omega \pblin) = r^2 \Omega^{-1} \slashed{\nabla}_4 \left(\frac{1}{2}r^2 \Omega^{-1} \slashed{\nabla}_4 (r \Omega^2 \ablin)\right) &= r^5 \slashed{\mathcal{D}}_2^\star  \slashed{\mathcal{D}}_1^\star \left(\rlin, \slin\right) - \frac{3}{4} r^5 \rho tr \underline{\chi} (\xlin - \xblin) .
 \label{Prel2} 
\end{align}
Here the last equality follows by  plugging in (\ref{rewriteBianchi}) and (\ref{rewriteBianchibar}) and diligently inserting the relevant null-structure and Bianchi equations produced by the extra derivative. Defining $\mathcal{L} := -r^2\slashed{\Delta}+4$, one has from~\eqref{Teuk1},~\eqref{Teuk2} and the definitions~\eqref{Prel},~\eqref{Prel2}, that $\Psilin,\Psilinb$ satisfy the \emph{Regge-Wheeler} equation (see~\cite{Gra.Hola})
\begin{align}
  \label{eq:RW}
  (\Omega \slashed{\nabla}_4)(\Omega \slashed{\nabla}_3)\Phi + \frac{\Omega^2}{r^2}\left(\mathcal{L} - \frac{6M}{r}\right)\Phi & = 0,
\end{align}
where $\Phi=\Psilin,\Psilinb$. We further define\footnote{Recall from Proposition 4.4.4 of \cite{Daf.Hol.Rod19} that $\mathcal{L}$ has eigenvalues $\ell(\ell+1) \geq 6$ acting on symmetric traceless tensors.}
\begin{align}
\Psilin^D := \Psilin- \Psilinb \ \ \ , \ \ \ \Psilin^R := \Psilin + \Psilinb + 12M \mathcal{L}^{-1} (\mathcal{L}-2)^{-1}\partial_t \left( \Psilin- \Psilinb\right).
\end{align}
It is easy to see that, if $\alin$ and $\ablin$ are regular at the horizon and at infinity (see~\eqref{smoothextended} and Definition~\ref{def:aAdSlin}), $\Psilin$ and $\Psilinb$ and hence also $\Psilin^D$ and $\Psilin^R$ extend regularly to both $\mathcal{H}^+$ and $\mathcal{I}$. In fact, the quantities $\Psilin^D$ and $\Psilin^R$ correspond (up to an unimportant numerical constant) to the analogous quantities $\Psi^D$ and $\Psi^R$ defined in \cite{Gra.Hola}, where they are spin weighted functions. We refer the reader to \cite{Gra.Hola} for further details. 

\section{Construction of initial data and local well-posedness} \label{sec:data}
In this section we define the class of solutions of the system of gravitational perturbations that will be the relevant class for the main theorem of the paper. This is Definition \ref{def:ssogp}. The remainder of the section is concerned with constructing such solutions from an appropriate notion of seed initial data on $\underline{C}_{v_0}$ by solving an initial boundary value problem. The reader wishing to take for granted the existence of this class of solutions upon first reading may move directly to Section \ref{sec:maintheorem} after having read Definition \ref{def:ssogp}. 

\subsection{The class of solutions}

In the following, we will consider a class of solutions $\mathscr{S}$ of the system of gravitational perturbations. 
\begin{definition} \label{def:ssogp}
We will say that $\mathscr{S}$ (as in (\ref{scollect})) is a smooth solution of the system of gravitational perturbations satisfying the boundary conditions if 
\begin{itemize}
\item $\mathscr{S}$ satisfies the equations (\ref{stos})--(\ref{Bianchi10}) on $\mathcal{M}_{int} \setminus \mathcal{H}^+$, with all dynamical quantities of $\mathscr{S}$ being smooth functions on $\mathcal{M}_{int} \setminus \mathcal{H}^+$ and 
\item the boundary conditions (\ref{bcl1})--(\ref{bcl4}) hold on $\mathcal{I}$.
\end{itemize}
\end{definition}

As mentioned, our goal is to construct such solutions uniquely from an appropriate notion of smooth seed data on $\underline{C}_{v_0}$. 

\begin{remark} \label{rem:typesol}
  The solutions we will construct (and hence the data) will have the additional regularity property of being aAdS in the linearised sense, see Definition \ref{def:aAdSlin}. In fact, we will prove uniform bounds on all quantities appearing in (\ref{lin1}) and (\ref{lin2}) in Section \ref{sec:proof} of the paper. However, we have not included the condition of being aAdS in the linearised sense in Definition \ref{def:ssogp} to make showing existence of solutions easier. 
\end{remark}

\subsection{Smooth seed initial data} 
We will now define the notion of smooth seed initial data along the cone $\underline{C}_{v_0=0}$. Below, an $S^2_{u,v}$-tensor $\xi$ is called smooth along $\underline{C}_{v_0}$ if for all $i, j \in \mathbb{N} \cup \{0\}$, the tensor $\left[\frac{r^2}{\Omega^2} \Omega \slashed{\nabla}_3\right]^{i} [r\slashed{\nabla}]^j \xi$ extends smoothly to infinity and to the horizon along $\underline{C}_{v_0}$.\footnote{In particular, the components in an orthonormal frame extend as smooth functions to $\mathcal{I}$.} 

\begin{definition} \label{def:seeddata}
A smooth seed initial data set on $\underline{C}_{v_0=0}$ for the system of gravitational perturbations consists of
\begin{enumerate}
\item  a tuple  $(\ablin_0; \pblin_0,\Psilinb_0)$, called {\bf the gauge independent part},
\item a tuple $(\underline{B}_0,  R_0; \underline{H}_0; \Olin_0,\bmlin_0, G^{\ell=1}_0,\hat{G}_0)$, called {\bf the gauge dependent part},
\item a $4$-dimensional vector $(\mathfrak{m},\mathfrak{a}_{-1},\mathfrak{a}_0,\mathfrak{a}_1) \in \mathbb{R}^4$ called {\bf the Kerr-AdS part}.
\end{enumerate}
The gauge independent part $(\ablin_0; \pblin_0,\Psilinb_0)$ consists of
\begin{itemize}
\item a smooth symmetric traceless tensor $\ablin_0$ prescribed along $\underline{C}_{v_0}$ with $r^3 \ablin_0$ extending smoothly to $S^2_{u_0=v_0,v_0}$,
\item smooth symmetric traceless $S^2_{u_0=v_0,v_0}$ tensors $\pblin_0$ and $\Psilinb_0$.
\end{itemize}
The gauge dependent part $(\underline{B}_0,  R_0; \underline{H}_0; \Olin_0,\bmlin_0, G^{\ell=1}_0,\hat{G}_0)$ consists of
\begin{itemize}
\item $S^2_{u_0=v_0,v_0}$ scalars $\underline{B}_0$,  $R_0$, $\underline{H}_0$, $G^{\ell=1}_0$ with $G^{\ell=1}_0$ supported for $\ell =1$,
\item a smooth symmetric traceless tensor $\hat{G}_0$ on $S^2_{u_0=v_0,v_0}$,
\item a smooth lapse function $\Olin_0$ prescribed along $\underline{C}_{v_0}$ with $r^2\cdot \Olin_0$ extending smoothly to $S^2_{u_0=v_0,v_0}$,
\item a smooth shift function $\bmlin_0$ prescribed along $\underline{C}_{v_0}$,with $\frac{1}{r} \bmlin_0$ extending smoothly to $S^2_{u_0=v_0,v_0}$.

\end{itemize}
\end{definition}

Some remarks are in order regarding the interpretation of the quantities prescribed. Suppose we can indeed construct, from a seed initial data set as above, a solution $\mathscr{S}$ as in (\ref{scollect}) of the system of gravitational perturbations satisfying the boundary conditions (as we will eventually in Theorem \ref{theo:wp} below). Then we would like that solution to be related to the seed data as follows.

\begin{definition} \label{def:assumedata}
Given a smooth seed initial data set 
\[
(\ablin_0; \pblin_0,\Psilinb_0) \ \ ,  \ \ (\underline{B}_0,  R_0; \underline{H}_0; \Olin_0,\bmlin_0, G^{\ell=1}_0,\hat{G}_0) \ \ , \ \  (\mathfrak{m},\mathfrak{a}_{-1},\mathfrak{a}_0,\mathfrak{a}_1),
\]
we say that a solution $\mathscr{S}$ of the system of gravitational perturbations satisfying the boundary conditions in the sense of Definition~\ref{def:ssogp} realises the given seed initial data if we have $ \ablin (u,v_0, \theta) = \ablin_0(u, \theta)$, $\Olin (u,v_0,\theta) = \Olin_0 (u,\theta)$ and $\bmlin(u,v_0,\theta)=\bmlin_0(u,\theta)$ along $\underline{C}_{v_0}$ and 
\begin{align}
\pblin_0 &= \lim_{u \rightarrow u_0} \left( \Omega\slashed{\nabla}_4 (r^3 \ablin)\right) (u,v_0,\theta)=  \lim_{u \rightarrow u_0} k\left( 2 r^4 \slashed{\mathcal{D}}_2^\star  \bblin + 6M r \xblin \right) (u,v_0,\theta) \label{psi0rel} \, ,  \\
\Psilinb_0 &=   \lim_{u \rightarrow u_0} \left( \Omega\slashed{\nabla}_4 (\Omega \slashed{\nabla}_4 (r^3 \ablin))\right) (u,v_0,\theta)
= -2 k^2r^2 \slashed{\mathcal{D}}_2^\star  \slashed{\nabla} R_0 + 6M k^3 \lim_{u \rightarrow u_0} r^2 (\xlin - \xblin) (u,v_0,\theta)\label{Psi0rel} \, , \\
\underline{B}_0 &=\lim_{u \rightarrow u_0} [r\slashed{div}] \bblin r^3(u,v_0,\theta) \, , \\
R_0 &= \lim_{u \rightarrow u_0} r^3 \rlin (u,v_0,\theta) \, ,  \\
\underline{H}_0 &=\lim_{u \rightarrow u_0} r^2 \Omega^{-2} \otxb  (u,v_0,\theta) \, ,  \\
 G^{\ell=1}_0 &= \lim_{u \rightarrow u_0} \Big(\frac{\glinto}{\sqrt{\slashed{g}}} \Big)_{\ell=1} (u,v_0,\theta) \, , \label{eq:H0bc}\\
 \hat{G}_0 &= \lim_{u \rightarrow u_0} \glinh (u,v_0,\theta) \, .
\end{align}
Finally, 
\[
 \mathfrak{a}_{m} = \frac{1}{6M} \lim_{u \rightarrow u_0} (r^4 \slashed{curl} \blin)_{\ell=1,m} (u,v_0,\theta) \ \ \ \textrm{and} \ \ \  \mathfrak{m}  =  \frac{1}{M} \lim_{u \rightarrow u_0} \bigg( r^3 \rlin - 3M \frac{\glinto}{\sqrt{\slashed{g}}} \bigg)_{\ell=0} (u,v_0,\theta) \, .
\]
\end{definition}



It is clear that the quantities $(\ablin_0; \pblin_0,\Psilinb_0)$ are gauge independent. The terminology ``gauge dependent" for $(\Olin_0, \bmlin_0; \underline{B}_0, \underline{H}_0, R_0, G^{\ell=1}_0, \hat{G}_0)$ and ``Kerr-AdS” for $(\mathfrak{m},\mathfrak{a}_{-1},\mathfrak{a}_0,\mathfrak{a}_1)$ becomes clear from the following proposition. To state it, note that clearly any solution $\mathscr{S}$ of the system of gravitational perturbations satisfying the boundary conditions induces a seed initial data set by restricting the solution to $\underline{C}_{v_0}$ and taking the above limits. 

\begin{proposition} \label{prop:pgz}
Consider a smooth seed initial data set 
\[
(\ablin_0; \pblin_0,\Psilinb_0) \ \ , \ \ (\underline{B}_0,  R_0; \underline{H}_0; \Olin_0,\bmlin_0, G^{\ell=1}_0,\hat{G}_0) \ \ , \ \ ({\mathfrak{m}},{\mathfrak{a}}_{-1},{\mathfrak{a}}_0,{\mathfrak{a}}_1)
\]
and assume there exists a solution $\mathscr{S}$ of the system of gravitational perturbations satisfying the boundary conditions and realising the given seed data in the sense of Definition \ref{def:assumedata}. Then there exists a function $f : [u_0,\infty) \times \mathbb{S}^2 \rightarrow \mathbb{R}$, a function $q : [u_0,\infty) \times \mathbb{S}^2 \rightarrow \mathbb{R}$ 
such that  the solution 
\[
\mathscr{S} + \mathscr{G}_f + \mathscr{G}_q - \mathscr{K}_{\mathfrak{m}, \vec{\mathfrak{a}}}
\]
induces a seed initial data set whose gauge dependent part and Kerr-AdS part vanishes identically. Here $\mathscr{G}_f$ is the pure gauge solution induced by $f$ as in Lemma \ref{lem:exactsol}, $\mathscr{G}_q$ the pure gauge solution induced by $q$ as in Lemma \ref{lem:puregaugemetric} and $\mathscr{K}_{\mathfrak{m}, \vec{\mathfrak{a}}}$ the linearised Kerr-AdS solution of Lemma \ref{lem:kerr}.
\end{proposition}

\begin{proof}
{\bf Step 1. Subtracting the linearised Kerr-AdS solution $\mathscr{K}_{\mathfrak{m},\vec{\mathfrak{a}}}$.} We note that 
\[
\bigg( r^3 \rlin - 3M \frac{\glinto}{\sqrt{\slashed{g}}}\bigg)_{\ell=0} \Bigg|_{\mathscr{S} - \mathcal{K}_{\mathfrak{m},\vec{\mathfrak{a}}}}  (u_0,v_0)  = \bigg( r^3 \rlin - 3M \frac{\glinto}{\sqrt{\slashed{g}}}\bigg)_{\ell=0} \Bigg|_{\mathscr{S}} (u_0,v_0) - \mathfrak{m} \cdot M = 0 \, .
\]
We choose the $\vec{\mathfrak{a}}=(\mathfrak{a}_{1},\mathfrak{a}_0,\mathfrak{a}_{-1})$ in the linearised Kerr-AdS solution such that for $m \in \{-1,0,1\}$ we have 
\[
\left( r \slashed{curl} (\blin r^3)\right)_{\ell=1, m} \Big|_{\mathscr{S} - \mathcal{K}_{\mathfrak{m},\vec{\mathfrak{a}}}}(u_0,v_0) = \left( r \slashed{curl} (\blin r^3)\right)_{\ell=1, m} \Big|_{\mathscr{S}}(u_0,v_0) -6M \mathfrak{a}_m = 0 \, .
\]
Note that both these quantities are gauge invariant and hence not affected by adding $\mathscr{G}_f$ and $\mathscr{G}_q$ in the following steps. Clearly, by construction $\tilde{\mathscr{S}}:=\mathscr{S}-\mathscr{K}_{\mathfrak{m},\vec{\mathfrak{a}}}$ induces seed data with vanishing Kerr-AdS part. \\

{\bf Step 2. Defining $\mathscr{G}_f$.} Given the solution $\tilde{\mathscr{S}}=\mathscr{S}-\mathscr{K}_{\mathfrak{m},\vec{\mathfrak{a}}}$, we define the function $f$ inducing the desired pure gauge solution $\mathscr{G}_f$ (by setting $f_u=f(u, \theta)$ and $f_v=f(v,\theta)$ in Lemma \ref{lem:exactsol}) as follows:
\begin{align}
f (u,\theta) &= \frac{ \lambda_1 r(u,v_0) + \lambda_0\Omega^2 (u,v_0) + \lambda_2}{\Omega^2(u,v_0)} - \frac{1}{\Omega^2(u,v_0)}  \int_{u_0}^u 2\Omega^2 \cdot \Olin(\hat{u},v_0,\theta) d\hat{u} 
\end{align}
for scalars $\lambda_0, \lambda_1, \lambda_2$ defined in turn by\footnote{The $\lambda_0, \lambda_1, \lambda_2$ are unique up to the spherical means of $\lambda_0$. Note that the $\ell=0$ mode of $\lambda_0$ generates a trivial pure gauge solution (changing $u$ and $v$ by a constant) corresponding to the static isometry of the background.}
\begin{itemize}
\item $\lambda_0$ such that $(\underline{B}_0)_{\tilde{\mathscr{S}}} (\theta) +6Mk r^2 \slashed{\Delta} \lambda_0=0$,
\item $\lambda_1$ such that $(R_0)_{\tilde{\mathscr{S}}} (\theta) + 6M \lambda_1 = 0$,
\item $\lambda_2$ such that $(\underline{H}_0)_{\tilde{\mathscr{S}}} (\theta) +  2 \Delta_{\mathbb{S}^2} \lambda_0 - 2 \lambda_2=0$.
\end{itemize}
We claim that with $\mathscr{G}_f$ thus defined, the solution $\tilde{\mathscr{S}} + \mathscr{G}_f$ already has $(\Olin_0)_{\tilde{\mathscr{S}} + \mathscr{G}_f}=0$ and $(\underline{B}_0)_{\tilde{\mathscr{S}} + \mathscr{G}_f}=0$, $(R_0)_{\tilde{\mathscr{S}} + \mathscr{G}_f}=0$ and $(\underline{H}_0)_{\tilde{\mathscr{S}} + \mathscr{G}_f}=0$. To verify this, we first note that 
\[
\lim_{u \rightarrow u_0} f_u (u, \theta) =\lim_{v \rightarrow v_0} f_v (v,\theta) = \lambda_0 \ \ \ \textrm{and} \ \ \ \ \lim_{u \rightarrow u_0}(\partial_u f_u) (u,\theta)=\lim_{v \rightarrow v_0}(\partial_v f_v) (v,\theta)=\lambda_1 \, .
\]
Using the expressions in Lemma \ref{lem:exactsol} we now easily check that 
\begin{align}
&(\Olin)_{\tilde{\mathscr{S}}+\mathscr{G}_f}(u,v_0,\theta) \nonumber \\
&=(\Olin) (u,v_0,\theta)_{\tilde{\mathscr{S}}}+\frac{1}{2} (\partial_v f_v) (v_0,\theta) 
+ \frac{1}{2} (\partial_r\Omega^2) (u,v_0) f_v(v_0,\theta)+\frac{1}{2\Omega^2}\partial_u\big(\Omega^2f_u\big)(u,v_0,\theta) \nonumber \\ 
&=(\Olin)_{\tilde{\mathscr{S}}}  (u,v_0,\theta) +\frac{1}{2} \lambda_1 + \frac{1}{2} \lambda_0 (\partial_r \Omega^2)(u,v_0) + \frac{1}{2} \left(-\lambda_1-\lambda_0 \partial_r (\Omega^2)  - 2\Olin_{\tilde{\mathscr{S}}} \right)(u,v_0,\theta) = 0
\end{align} 
on the ingoing cone $\underline{C}_{v_0}$ independently of the $\lambda_i$. From Lemma \ref{lem:exactsol} we then verify
\begin{align}
\lim_{u \rightarrow u_0} (r\slashed{div} r^4 \Omega^{-1} \bblin)_{\tilde{\mathscr{S}}+\mathscr{G}_f} (u,v_0,\theta)  &= \lim_{u \rightarrow u_0} (r\slashed{div} r^4 \Omega^{-1} \bblin)_{\tilde{\mathscr{S}}} (u,v_0,\theta) + 6M r^2 \slashed{\Delta} \lambda_0 = 0 \, .
\end{align}
To check the condition on $\rlin$, we first compute the limit
\[
\lim_{u \rightarrow v} \left((r(u,v)(f_v (v,\theta) - f_u(u,\theta))\right) = \lim_{v \rightarrow u} \left(\frac{f_v - f_u}{v-u}\right) = f^\prime(u) \, .
\]
It therefore follows from Lemma \ref{lem:exactsol} that we have 
\begin{align}
(r^3 \rlin)_{\tilde{\mathscr{S}}+\mathscr{G}_f} (v_0,v_0,\theta) &= (r^3 \rlin)_{\tilde{\mathscr{S}}} (v_0,v_0,\theta) + 6M \lim_{u \rightarrow v_0} \left((r(u,v_0)(f_v (v_0,\theta) - f_u(u,\theta))\right)  \nonumber \\
&=  (r^3 \rlin)_{\tilde{\mathscr{S}}} (v_0,v_0,\theta) + 6M \lim_{u \rightarrow v_0} \left(\frac{f(v_0,\theta) - f(u,\theta)}{v_0-u}\right)\nonumber \\ &=  (r^3 \rlin)_{\tilde{\mathscr{S}}} (v_0,v_0,\theta) + 6M f^\prime(v_0) \, ,
\end{align}
and we see that $(R_0)_{\tilde{\mathscr{S}}+\mathscr{G}_f}=0$ if $\lambda_1$ is chosen as above. For the last limit we note using Lemma \ref{lem:exactsol} and the fact that we already obtained $\olinb_{\tilde{\mathscr{S}}+\mathscr{G}_f}=0$, the equalities
\begin{align}
\frac{r^2}{\Omega^2} \otxb_{\tilde{\mathscr{S}}+\mathscr{G}_f} (u,v_0,\theta)   
\nonumber \\
= \frac{r^2}{\Omega^2} \otxb_{\tilde{\mathscr{S}}} (u,v_0,\theta) + 2\left[\Delta_{\mathbb{S}^2} \lambda_0 + \lambda_0\Big(1-\frac{4M}{r} - k^2 r^2 \Big)(u,v_0) \right] - \frac{r^2}{\Omega^2}  \partial_u\Big(\frac{2\Omega^2}{r}f_u\Big) (u,v_0,\theta) \, . \nonumber
\end{align}
We want to take the limit $u \rightarrow u_0=v_0$. We compute up to terms vanishing in the limit $u \rightarrow u_0$ (indicated by $\equiv$)
\begin{align}
 - \frac{r^2}{\Omega^2}  \partial_u\Big(\frac{2\Omega^2}{r}f_u\Big)(u,v_0,\theta) &\equiv -2 \Omega^2 f_u (u,v_0,\theta) -2r \left(-\lambda_1 - (\Omega^2)_r \lambda_0 - 2 \Olin_{\tilde{\mathscr{S}}}   \right) (u,v_0,\theta) \nonumber \\
 &\equiv -2 \left( + \lambda_0\Omega^2 (u,v_0) -r (\Omega^2)_r \lambda_0+ \lambda_2\right) \nonumber \\
 &\equiv -2 \left( - \lambda_0 k^2r^2 (u,v_0) + \lambda_0  + \lambda_2\right) \, .
\end{align}
We conclude
\begin{align}
(\underline{H}_0)_{\tilde{\mathscr{S}} + \mathscr{G}_f} = (\underline{H}_0)_{\tilde{\mathscr{S}}}  +  2 \Delta_{\mathbb{S}^2} \lambda_0 - 2 \lambda_2 = 0 \, . 
\end{align}
{\bf Step 3. Defining $\mathscr{G}_q$.} Let us define (unique) functions $\phi_1,\phi_1$ with vanishing spherical means by $\bmlin_{\tilde{\mathscr{S}}+\mathscr{G}_f}=r \slashed{\mathcal{D}}_1^\star (\phi_1,\phi_2)$. Set $q_1,q_2$ to solve 
\[
\partial_u q_1 + \frac{\phi_1}{r} = 0  \ \ \ , \ \ \ \partial_u q_2 + \frac{\phi_2}{r} = 0  \, ,
\]
with initial conditions at $u_0$ determined by $q_1(u_0,\theta)$ and $q_2(u_0,\theta)$ having vanishing spherical means and solving $\glinh(u_0,\theta)= 2 r^2 \slashed{\mathcal{D}}_2^\star \slashed{\mathcal{D}}_1^\star(q_1(u_0,\theta),q_2(u_0,\theta))$ and $r^2 \slashed{\Delta} q_1 = G^{\ell=1}_0 (u_0,\theta)$. Note that this determines $q_1, q_2$ uniquely up to $\ell=1$ modes of $q_2$ on the sphere $S^2_{u_0,v_0}$ (corresponding to the rotational invariance), which generate trivial gauge transformations.  It is now immediate from Lemma \ref{lem:puregaugemetric} that $ (G^{\ell=1}_0)_{\tilde{\mathscr{S}}+\mathscr{G}_f+\mathscr{G}_q} (\theta)=0$, $ (\hat{G})_{\tilde{\mathscr{S}}+\mathscr{G}_f+\mathscr{G}_q} (\theta)=0$ and $\bmlin_{\tilde{\mathscr{S}}+\mathscr{G}_f+\mathscr{G}_q} (u,v_0,\theta) = 0$, while all other metric, Ricci and curvature components remain unaffected. 
%
\end{proof}

\subsection{Construction of all geometric quantities on the initial cone}
\begin{proposition} \label{prop:allqdata}
Given a smooth seed initial data set 
\[
(\ablin_0; \pblin_0,\Psilinb_0) \ \ , \ \ (\underline{B}_0,  R_0; \underline{H}_0; \Olin_0,\bmlin_0, G^{\ell=1}_0,\hat{G}_0) \ \ , \ \ (\mathfrak{m},\mathfrak{a}_{-1},\mathfrak{a}_0,\mathfrak{a}_1)
\]
one can construct uniquely all geometric quantities (\ref{scollect}) of the system of gravitational perturbations (including all their tangential and transversal derivatives) on $\underline{C}_{v_0}$ such that
\begin{enumerate}[(1)]
\item  The equations  (\ref{stos})--(\ref{Bianchi10}) and all tangential and transversal derivatives thereof, hold on the cone $\underline{C}_{v_0}$.
\item  The boundary conditions (\ref{bcl1})--(\ref{bclxi}) hold on the sphere $S^2_{u_0,v_0}$ of the conformal boundary. 
\item The condition $[\slashed{\nabla}_T]^p [r\slashed{\nabla}]^q \olin = 0$ holds on the sphere $S^2_{u_0,v_0}$ for any $p,q \in \mathbb{N}^0$.
\item  On $\underline{C}_{v_0}$, the relations  of Proposition \ref{def:assumedata} between the seed data and the geometric quantities hold.
\end{enumerate}
Moreover, if the seed initial data set vanishes identically then so do all geometric quantities on $\underline{C}_{v_0}$. 
\end{proposition}

\begin{remark}
To explain $(3)$ above, we observe that some condition on $\olin$ at $S^2_{u_0,v_0}$ is necessary to determine all quantities on the cone as the quantity $\olin$ only admits a null structure equation in the $3$-direction and hence needs to be supplemented with data on $S^2_{u_0,v_0}$. Our choice in $(3)$ is weaker but consistent with our desire to construct data and solutions that are being aAdS in the linearised sense (Definition \ref{def:aAdSlin}).
\end{remark}


\begin{proof}
The logic of the proof is to construct all the geometric quantities of the solution using freely equations and relations that have to hold on $\underline{C}_{v_0}$ by Items $(1)-(4)$ above. In a second step, once all geometric quantities have been constructed without contradiction, we verify that all equations and relations have been used in the construction. 

First, we define $\olinb = \Omega \slashed{\nabla}_3 (\Olin_0) (u,v_0,\theta)$ consistent with (\ref{oml3}). Note that $r\olinb(u_0,v_0,\theta)$ extends smoothly to $u=u_0$. 

Since all of (\ref{stos})--(\ref{Bianchi10}) have to hold for the geometric quantities we want to construct, it follows that also the Teukolsky equations (\ref{Teuk1})--(\ref{Teuk2}) restricted to $\underline{C}_{v_0}$ must hold. This determines (by transport along $\underline{C}_{v_0}$) the regular transversal derivatives $\Omega \slashed{\nabla}_4 (r^5 \Omega^{-2} \ablin)$ and $\Omega \slashed{\nabla}_4 (\Omega \slashed{\nabla}_4 (r^5 \Omega^{-2} \ablin))$ along all of $\underline{C}_{v_0}$ in terms of the prescribed seed data $\ablin_0$ and $\pblin_0$, $\Psilinb_0$ using (\ref{psi0rel}), (\ref{Psi0rel}).  

Next, we impose that on the boundary sphere $S^2_{u_0,v_0}$,
\begin{align}\label{eq:defbblinpblin}
  -(r^2 \slashed{\Delta} + 2K) r\slashed{curl} (r^3 \bblin) = 2r^3  \slashed{curl} \slashed{div} \slashed{\mathcal{D}}_2^\star (r^3 \bblin)=r^2   \slashed{curl} \slashed{div} \pblin_0,
\end{align}
which, since the $\ell=1$ mode of $r^4 \slashed{curl} \bblin$ is determined by the seed data on $S^2_{u_0,v_0}$, determines $ r^4 \slashed{curl} \bblin (u_0,v_0,\theta)$ at infinity from seed data. Since the seed data require $\lim_{u \rightarrow u_0} r^4 \slashed{div} \bblin (u,v_0,\theta) = \underline{B}_0$, the quantity $ r^3 \bblin (u_0,v_0,\theta)$ is determined uniquely by elliptic theory. We now define the limit $\underline{X} := \lim_{u \rightarrow u_0}  r \underline{\hat{\chi}}$ by
\begin{align}\label{eq:defunderlineXboundary}
  6Mk \underline{X} & = \pblin_0 - \lim_{u\to u_0} 2kr^4\slashed{\mathcal{D}}^\star_2 \bblin, 
\end{align}
consistently with (\ref{psi0rel}).
We can now obtain $\bblin$ along all of $\underline{C}_{v_0}$ by integrating (\ref{Bianchi9}) written as
\begin{align*}
  \Omega \slashed{\nabla}_3 (r^4 \Omega^{-1} \bblin) = \Omega^2 r^4 \slashed{div} \ablin \Omega^{-2} \, . 
\end{align*}
Note that the right hand side is integrable in $u$ on $\underline{C}_{v_0}$ and thus determines $r^4 \Omega^{-1} \bblin$. Similarly, we can integrate (\ref{tchi}) from infinity with boundary condition~\eqref{eq:defunderlineXboundary} to obtain $r^2 \Omega^{-1}\xblin$ along $\underline{C}_{v_0}$. This comes with the (smooth) expansion
\begin{align}\label{eq:expxblinX}
r^2 \Omega^{-1} \xblin = \underline{X} + \mathcal{O}(r^{-2})\, .
\end{align}
Note that we can now determine $r^3 \slin$ from (\ref{Bianchi7}) and the boundary condition $r^3 \slin (u_0,v_0,\theta)=0$. 

From (\ref{vray}) we have
\begin{align*}
\partial_u ((\Omega^{-2} \otxb r^2) = -4r  \olinb \, .
\end{align*}
Since $r \olinb$ extends smoothly to $S^2_{u_0,v_0}$ the right hand side is integrable. Using the boundary condition~\eqref{eq:H0bc} this defines $\otxb \Omega^{-2} r^2$ on all of $\underline{C}_{v_0}$. Using (\ref{Bianchi5}), we can also determine $\rlin$ from $R_0$ at infinity (set by seed data) and the bounds on $\otxb$ and $\bblin$. 
%
%
We now write the Codazzi equation (\ref{ellipchi}) as 
\begin{align} \label{elid}
 \left( r^2 \elin  \right) + \frac{1}{2} r \slashed{\nabla} \Big(\otxb r^2 \Omega^{-2} \Big)  = r^3 \slashed{div} \Omega^{-1} \xblin - r^3 \Omega^{-1} \bblin \, .
\end{align}
All four terms have a finite limit on null infinity ($r^3 \Omega^{-1} \bblin$ actually vanishes) which therefore determines $r^2\elin$ (which extends regularly to $u=u_0$). We also set, consistent with (\ref{oml3}),
\begin{align}\label{eq:defeblinid}
\eblin = - \elin + 2 \slashed{\nabla} (\Olin_0) \, . 
\end{align}
One easily shows that the $3$-equation for $\eblin$, (\ref{propeta}), holds as a consequence of (\ref{elid}) and the $3$-equations holding for $\xblin$, $\otxb$, $\Olin$ and $\blin$. Moreover, using~\eqref{propeta} and the definition of $\slin$ by \eqref{Bianchi7}, one can show that $\Omega\slashed{\nabla}_3\left(r^3\slashed{curl}\eblin + r^3\slin\right) = 0$. Now, both quantities in the parenthesis vanish at infinity, which shows that the elliptic equations~\eqref{curleta} hold on $\underline{C}_{v_0}$. Indeed, $r^3\slin \to 0$ by definition of $\slin$ and, by~\eqref{eq:defbblinpblin} and~\eqref{eq:defunderlineXboundary}, one has $\lim_{u \rightarrow u_0} r^3 \slashed{curl} \slashed{div} \xblin = 0$, which, by the fact that $r^2\bblin=0$ as $u \rightarrow u_0$ and~\eqref{eq:defeblinid}, shows that $ r^3\slashed{curl}\eblin = r^3\slashed{curl}\elin = 0$ as $u \rightarrow u_0$. 

We next determine $r^2\Omega \blin$ from (\ref{Bianchi3}) and the previously defined quantities, using the boundary condition at infinity $r^3 \blin=-r^3 \bblin$. To determine $\xlin$ we write (\ref{chih3}) as 
\begin{align}
\Omega \slashed{\nabla}_3  \left(r \Omega \xlin -r \underline{X}\right)   &= -2 \Omega^2 r \slashed{\mathcal{D}}_2^\star \elin -  \Omega^2 \left( \Omega \xblin -\underline{X} \right) \, .
\end{align}
Using~\eqref{eq:expxblinX}, the right hand side is integrable and we define $\Omega\xlin$ by imposing the condition $r \Omega \xlin -r \underline{X} \to 0$ when $u\to u_0$. 

We can now determine the last curvature component $\alin$ from (\ref{Bianchi1}) and the fact that $r^3 \alin=r^3 \ablin$ at infinity. Note that with this all Bianchi equations in the $3$-directions hold by construction. 

The definition of $\olin$ follows from (\ref{oml2}) and the condition that $\olin=0$ on $S^2_{u_0,v_0}$.

Consistently with (\ref{Psi0rel}) and the linearised formula for the Gauss curvature~\eqref{gaussfootnote}, we define the asymptotic linearised Gauss curvature $\Klin$ to be 
\begin{align} \label{axe}
  \lim_{u \rightarrow u_0} r^3 \slashed{\nabla} \Klin := \frac{1}{3M} \left(r\slashed{div} \Psilinb_0+2r^3 \slashed{div} \slashed{\mathcal{D}}_2^\star  \slashed{\nabla} R_0 \right) = \lim_{u \rightarrow u_0} \left(-r^3  \slashed{div}  \xblin  + r^3  \slashed{div} \xlin \right) \, ,
\end{align}
where the second identity comes from the previous definitions of $\xblin$ and $\xlin$. This determines $(\Klin)_{\ell \geq 2}$. 
Now recall that we have  $(\frac{\glinto}{\sqrt{\slashed{g}}})_{\ell=0} (u_0,v_0)=\frac{1}{3M}(R_0)_{\ell=0}-\frac{1}{3}\mathfrak{m}$ and $(\frac{\glinto}{\sqrt{\slashed{g}}})_{\ell=1} (u_0,v_0)=0$ by construction, which in particular determines $(\Klin)_{\ell=0,1}$ from formula (\ref{gaussfootnote}). Conversely, formula (\ref{gaussfootnote}) determines $(\frac{\glinto}{\sqrt{\slashed{g}}})_{\ell\geq 2} (u_0,v_0)$ when read as an elliptic equation for this quantity, since $\glinh (u_0,v_0,\theta)=\hat{G}_0$ is prescribed at infinity. Finally, (\ref{stos}) and (\ref{stos2}) determine determine $(\frac{\glinto}{\sqrt{\slashed{g}}})$ and  $\glinh$ on all of $\underline{C}_{v_0}$.

We next use the linearised Gauss equation (\ref{lingauss}) to determine $\otx$ on $\underline{C}_{v_0}$. Note that this shows in particular that $\otx-\otxb=0$ at $(u_0,v_0)$. Differentiating the linearised Gauss equation in $\slashed{\nabla}_3$ we verify that (\ref{dbtc}) has to hold on $\underline{C}_{v_0}$. To show validity of the other Codazzi equations of~\eqref{ellipchi}, which we rewrite as
\begin{align} \label{codm}
r \slashed{div} \Omega \xlin = -\Omega^2 \eblin  -r \Omega \blin + \frac{1}{2} r \slashed{\nabla} \otx \, , 
\end{align}
on $\underline{C}_{v_0}$, we first use the boundary asymptotics on $\blin,\bblin$ and $\otx,\otxb$, the definition of $\eblin$ and the asymptotics of $\Olin_0$, and the validity of the underlined Codazzi equation~\eqref{elid}, to verify that (\ref{codm}) holds on $S^2_{u_0,v_0}$, and then differentiate in $\slashed{\nabla}_3$: Since the resulting expression vanishes after inserting the evolution equations already established, we deduce that (\ref{codm}) indeed holds on $\underline{C}_{v_0}$. 

We have thus determined all quantities from seed data with quantitative estimates and obtained validity of all evolution equations in the $3$-direction as well as elliptic relations on spheres contained in (\ref{stos})--(\ref{Bianchi10}). Obviously, tangential derivatives to $\underline{C}_{v_0}$ applied to these equations also hold. 

Moreover, all relations of Definition \ref{def:assumedata} and all the boundary conditions of Definition~\ref{def:bc} hold as they have been explicitly used to define quantities. To determine finally the transversal derivatives, we use that any geometric quantity (except $\olin$ and $\alin$) satisfies a Bianchi or null structure equation which determines its $\Omega \slashed{\nabla}_4$ derivative in terms of angular (=tangential) derivatives. This consistently constructs all transversal derivatives and ensures validity of the relevant equations by definition.  For the exceptional $\alin$ we can determine transversal derivatives from the Teukolsky equation~\eqref{Teuk1} and the boundary conditions relating $\alin$ and the seed data $\ablin$ at $S^2_{u_0,v_0}$.  For $\olin$ we have can obtain the $\Omega \slashed{\nabla}_4$-derivative from the fact that we can determine the $\slashed{\nabla}_T$-derivative from the boundary condition in Item $2$ and the commuted equation for $\olin$, equation (\ref{oml2}). 

The last claim about the trivial data follows easily from redoing the above proof with trivial data. 
 \end{proof}

 \begin{corollary} \label{cor:puregaugepluskerr}
Consider a smooth seed initial data set 
\[
(\ablin_0; \pblin_0,\Psilinb_0) \ \ , \ \ (\underline{B}_0,  R_0; \underline{H}_0; \Olin_0,\bmlin_0, G^{\ell=1}_0,\hat{G}_0) \ \ , \ \ ({\mathfrak{m}},{\mathfrak{a}}_{-1},{\mathfrak{a}}_0,{\mathfrak{a}}_1)
\]
and assume there exists a solution $\mathscr{S}$ of the system of gravitational perturbations satisfying the boundary conditions, realising the given seed data in the sense of Definition \ref{def:assumedata} and satisfying condition (3) of Proposition \ref{prop:allqdata}.  If the gauge independent part of the solution $\mathscr{S}$ vanishes (which in particular happens if the solution is supported for $\ell \leq 1$), the solution $\mathscr{S}$ is the sum of a pure gauge and a linearised Kerr-AdS solution.
 \end{corollary}
 
 \begin{proof}
 We add a pure gauge and linearised Kerr-AdS solution $\mathscr{G_f}+\mathscr{G}_q - \mathscr{K}_{\mathfrak{m},\vec{\mathfrak{a}}}$ as in Proposition \ref{prop:pgz} to the solution $\mathscr{S}$ such that also the gauge dependent and the Kerr-AdS part of the seed data vanish. By Proposition \ref{prop:allqdata} (which applies since the data induced by both $\mathscr{S}$ and $\mathscr{S}^\prime = \mathscr{S} + \mathscr{G_f}+\mathscr{G}_q - \mathscr{K}_{\mathfrak{m},\vec{\mathfrak{a}}}$ on $\underline{C}_{v_0}$ satisfy $(1)-(4)$) we conclude that all geometric quantities of $\mathscr{S}^\prime$ vanish on $\underline{C}_{v_0}$. In particular, Lemma \ref{lem:zeromodevanishing} applies and the solution $\mathscr{S}^\prime$ must have vanishing $\ell=0$ mode. Moreover, since $\ablin$ and $\alin$ vanish identically on $\underline{C}_{v_0}$, it follows from the decoupled Teukolsky equations (\ref{Teuk1})--(\ref{Teuk2}) and uniqueness of its solutions (see for instance Theorem 1.4 of \cite{Gra.Hol23} that $\ablin$ and $\alin$ vanish identically globally in $\mathcal{M}$. Since $\xlin$ and $\blin$ vanish on $\underline{C}_{v_0}$, they vanish globally by their evolution equation in the $4$-direction. Inserting the boundary conditions we conclude vanishing of $\xblin$ and $\bblin$ on $\mathcal{I}$ and by their evolution equation in the $3$-direction, globally. Similarly, $\slin$ vanishes by the equation (\ref{Bianchi6}) and the vanishing on $\underline{C}_{v_0}$. 
 From (\ref{Bianchi3}) and (\ref{Bianchi8}) and the vanishing of $\slin,\blin,\bblin$ we now conclude $\elin - \eblin = 0$. The equation (\ref{propeta}) and the vanishing of $\elin$  on $\underline{C}_{v_0}$ produces global vanishing of $\elin$ and hence of $\eblin$ individually. Revisiting  (\ref{Bianchi3}) it follows that $\rlin$ vanishes (recall we have established in the beginning that the $\ell=0$ mode vanishes). Codazzi shows the global vanishing of $\otx$ and $\otxb$ and (\ref{oml1}), (\ref{oml2}) that of $\olin$ and $\olinb$ respectively. The proof is complete. 
 \end{proof}

\subsection{Local Well-posedness}

We can finally state the well-posedness theorem for solutions to the system of gravitational perturbations:

\begin{theorem} \label{theo:wp}
Given a smooth seed initial data set
\[
(\ablin_0; \pblin_0,\Psilinb_0) \ \ , \ \ (\Olin_0, \bmlin_0; \underline{B}_0, \underline{H}_0, R_0, G^{\ell=1}_0, \hat{G}_0) \ \ , \ \ (\mathfrak{m},\mathfrak{a}_{-1},\mathfrak{a}_0,\mathfrak{a}_1) \, , 
\]
 there exists a unique solution $\mathscr{S}$ of the system of gravitational perturbations on $\mathcal{M}$ satisfying the boundary conditions (Definition \ref{def:ssogp}), realising the given seed data in the sense of Definition \ref{def:assumedata} and satisfying that $[\slashed{\nabla}_T]^p [r\slashed{\nabla}]^q \olin = 0$ holds on the sphere $S^2_{u_0,v_0}$ for any $p,q \in \mathbb{N}^0$.
\end{theorem}

\begin{remark}
With more work one can show that the solution constructed in the above theorem is aAdS in the linearised sense according to Definition \ref{def:aAdSlin}. In particular, we will establish uniform bounds on all quantities in (\ref{lin1}) and (\ref{lin2}) as part of our main theorem. See also Remarks \ref{rem:betterbounds} and \ref{rem:typesol} above.
\end{remark}

\begin{proof}\footnote{The authors would like to thank Leonhard Kehrberger for discussions and suggesting the argument with $\sigma^\prime$ in Step 2 below.}
The uniqueness part follows from the proof of Corollary \ref{cor:puregaugepluskerr} where it is shown that a zero seed initial data set can only produce the zero solution. We also know from the proof of Proposition \ref{prop:pgz} that the given seed data restricted to $\ell \leq 1$ agrees with the seed data induced by an appropriate pure gauge solution supported on $\ell=1$ plus the seed data of a linearised Kerr-AdS solution which establishes existence for $\ell \leq 1$. 

In summary, we only need to construct the solution for $\ell \geq 2$ from seed data supported for $\ell \geq 2$.

{\bf Step 1: \emph{Constructing} geometric quantities in $\mathcal{M}$, Part I.}
We first construct from the seed data set all quantities on $\underline{C}_{v_0}$ as in the proof of Proposition \ref{prop:allqdata}. This in particular determines smooth $\ablin$, $\alin$ on $\underline{C}_{v_0}$ such that in particular the master energy defined in (\ref{masterdata}) is finite for all $n$. We can hence apply the well-posedness theorem for the Teukolsky system (\ref{Teuk1})--(\ref{Teuk2}) (cf.~Theorem 1.4 in \cite{Gra.Hol23}) and obtain smooth $\ablin$ and $\alin$ globally on $\mathcal{M}$. We next determine $\xlin$ and $\blin$ globally from their transport equation in the $4$-direction, i.e.~integrating  (\ref{tchi}) and (\ref{Bianchi2}) from data. We then determine $\xblin$ and $\bblin$ globally by integrating the transport equations (\ref{tchi}) and (\ref{Bianchi9}) from the boundary (using the boundary conditions as initial conditions for $\xblin$ and $\bblin$, i.e.~$r^3 \bblin= -r^3 \blin$ and $r\xlin=r\xblin$ on $\mathcal{I}$). With this, (\ref{Bianchi2}) holds by definition and it is easy to see that (\ref{Bianchi1}) also holds because it holds on data and $\Omega \slashed{\nabla}_4$ of this equation vanishes by the fact that the Teukolsky equation~\eqref{Teuk1} holds for $\alin$ and the equations (\ref{tchi}), (\ref{Bianchi2}) hold in the $4$-direction by construction. Similarly (\ref{Bianchi10}) holds on the boundary (by the fact that (\ref{Bianchi1}) holds there and the boundary conditions imposed) and its $\Omega \slashed{\nabla}_3$-derivative is zero by the validity of the Teukolsky equation and (\ref{Bianchi9}) and (\ref{tchi}). Observe also that all quantities constructed are smooth up to the boundary. This will continue to be true for the quantities constructed below. 

We next define the quantity $\sigma$ by integrating (\ref{Bianchi6}) from initial data. Note that with this definition, $r^3\slin$ has a finite limit on the boundary (which we do not know vanishing of yet). We then set
$\slashed{curl} \elin := \slin$ and $\slashed{curl} \eblin := -\slin$ so that (\ref{curleta}) is satisfied. On the other hand, we determine $\slashed{div} \elin$ and $\slashed{div} \eblin$ globally from (\ref{chih3}), (\ref{chih3b}) by taking the $\slashed{div} \slashed{div}$ of these equations and observing that $\slashed{div} \slashed{div}\slashed{\mathcal{D}}_2^\star \xi = \left( -\frac{1}{2} \slashed{\Delta} - K \right) \slashed{div} \xi$ with $\left( \slashed{\Delta} + 2K \right) $ having trivial kernel on the space of functions with $\ell \geq 2$. This determines $\elin$ and $\eblin$ uniquely by standard elliptic theory and one can also show that $r (\elin + \eblin) = 0$ holds on the conformal boundary $\mathcal{I}$.\footnote{To see the latter, follow the proof of Proposition \ref{prop:improvedchidifference} below to establish (\ref{diff1}), from which the claim follows after subtracting (\ref{chih3}) and (\ref{chih3b}) multiplied by $r$.}

{\bf Step 2. Verifying the $\slashed{curl}$-equations.} With the quantities defined we can already verify some of the equations. We claim that the $\slashed{curl}$ applied to the Codazzi equation (\ref{ellipchi}) for $\hat{\chi}$ holds. To see this, note that 
\[
\Omega \slashed{\nabla}_4 \left( r^2 \slashed{curl} \slashed{div} \Omega^{-1}r^2\xlin + \frac{1}{r} (r \slashed{curl}) \Omega^{-1}r^4 \blin  - r^3 \slin \right) = 0 \, .
\]
by the propagation equations we have defined in the $4$-direction. Since the quantity in round brackets also vanishes on $\underline{C}_{v_0}$ we conclude that it is zero globally and hence $\slashed{curl}$ applied to the Codazzi equation (\ref{ellipchi}) for $\hat{\chi}$ holds. 
We similarly conclude that the $\slashed{curl}$ applied to (\ref{Bianchi3}) holds: Indeed, defining 
\[
B:=\Omega \slashed{\nabla}_3 \Big(  \frac{\Omega^2}{r^2} r\slashed{curl} \frac{r^4 \blin}{\Omega} \Big) +r^3 \Omega^2 \slashed{\Delta}\slin - 3\rho r^3 \Omega^2 \slin 
\]
we derive (using the transport equations in the $4$-direction for $\blin$ and $\slin$ as well as (\ref{Bianchi1}) and the just established $\slashed{curl}$ applied to the Codazzi equation (\ref{ellipchi}) for $\hat{\chi}$) an equation of the form $\Omega \slashed{\nabla}_4 B = -\left(\frac{2}{r} - \frac{6M}{r^2}\right)B$. Since $B=0$ on $\underline{C}_{v_0}$ we conclude that $B=0$ globally.

We next define an auxiliary quantity $\slin^\prime$ such that $\lim_{u\to v}r^3\slin^\prime = \lim_{u\to v}r^3\slin$ and 
\begin{align} \label{auxprop}
\Omega \slashed{\nabla}_3 (r^3 \slin^\prime) = -\slashed{curl} r^3 \bblin \Omega \, .
\end{align}
Note that $\slin^\prime$ agrees with $\slin$ on $\underline{C}_{v_0}$ and on $\mathcal{I}$. We verify that 
\[
\Omega \slashed{\nabla}_3 \Big( r^2 \slashed{curl} \slashed{div} \Omega^{-1}\hat{\underline{\chi}}r^2 - \frac{1}{r} (r \slashed{curl}) \Omega^{-1}r^4 \underline{\beta}  - r^3 \slin^\prime \Big) = 0 \, , 
\]
and since the quantity vanishes on $\mathcal{I}$ (by the previously imposed/derived boundary relations and the Codazzi equation for $\xlin$), the expression in round brackets vanishes globally. Defining 
\begin{align} \label{Budef}
\underline{B}:=\Omega \slashed{\nabla}_4 \left(  \frac{\Omega^2}{r^2} r\slashed{curl} \frac{r^4 \underline{\blin}}{\Omega} \right) +r^3 \Omega^2 \slashed{\Delta}\slin^\prime - 3\rho r^3 \Omega^2 \slin^\prime \, ,
\end{align}
we first check that $\Omega \slashed{\nabla}_3 \left(\frac{r^2}{\Omega^2} \underline{B}\right)=0$ holds (a computation similar to the one for $B$ above) and then verify that $\underline{B}$ vanishes at infinity (which follows from $T(r^4 \slashed{curl} \blin + r^4 \slashed{curl} \bblin)=0$ which is in turn a consequence of how we defined $\bblin$). The two observations allow us to conclude $\underline{B}=0$ globally. We can finally conclude that $\slin^\prime$ satisfies the Regge-Wheeler equation~\eqref{eq:RW} by applying $\Omega \slashed{\nabla}_4$ to (\ref{auxprop}) and inserting (\ref{Budef}). In summary, $\slin$ and $\slin^\prime$ satisfy the same wave equation with the same data and boundary condition at $\mathcal{I}$ and hence agree globally, \emph{i.e.} $\slin=\slin^\prime$ on $\mathcal{M}$. In particular, Bianchi equation~\eqref{Bianchi7} holds by~\eqref{auxprop} and, adding the two Bianchi equations~\eqref{Bianchi6}, \eqref{Bianchi7}, one has that $\lim_{u\to v}T(r^3\slin) = 0$. Thus, using that on the initial sphere at infinity $S_{u_0,v_0}$ $r^3\slin$ vanishes, one has that $r^3\slin$ vanishes globally on $\mathcal{I}$. At this point we have ensured the validity of the Bianchi equations (\ref{Bianchi1}), (\ref{Bianchi2}), (\ref{Bianchi6}), (\ref{Bianchi7}), (\ref{Bianchi9}), (\ref{Bianchi10}) as well as the $\slashed{curl}$ of (\ref{Bianchi3}) and -- by the vanishing of~\eqref{Budef} and $\slin^\prime=\slin$ -- (\ref{Bianchi8}) respectively. By construction also the equations (\ref{curleta}) and the $\slashed{curl}$ of (both of) (\ref{ellipchi}) hold.  One easily checks that the $\slashed{curl}$ of (\ref{propeta}) holds. Moreover, we check that $\slashed{curl} \slashed{div}$ applied to (\ref{chih3}) and (\ref{chih3b}) respectively hold (inserting the $\slashed{curl}$ of (\ref{ellipchi}) and using the propagation equations already established). This means that (\ref{chih3}) and (\ref{chih3b}) hold unconditionally since the $\slashed{div} \slashed{div}$ part of these holds by construction. 

{\bf Step 3: \emph{Constructing} geometric quantities in $\mathcal{M}$, Part II.}
We now set $2\slashed{\nabla}(\Olin) := \elin + \eblin$, which is well defined as $\slashed{curl} (\elin + \eblin) = 0$ by our definition of $\elin, \eblin$ above. Moreover, we set $\olin = \Omega \slashed{\nabla}_4 (\Olin)$ and $\olinb = \Omega \slashed{\nabla}_3 (\Olin)$. Next we define $\otx$ by (\ref{ellipchi}) (which is well defined as we have already verified that the $\slashed{curl}$ of this equation holds) and $\otxb$ directly integrating the evolution equation (\ref{dtcb}) from data. We define the quantity $\rho$ by integrating (\ref{Bianchi4}). Finally, the metric quantities $\glinto$, $\glinh$, $\bmlin$ are defined by integrating their equation in the $4$-direction from data, i.e.~(\ref{stos}), (\ref{stos2}) and (\ref{bequat}).

{\bf Step 4. Verifying the remaining equations.} 
Differentiating (\ref{chih3}) with respect to $\Omega \slashed{\nabla}_4$ and (\ref{chih3b}) with respect to $\Omega \slashed{\nabla}_3$ shows that (\ref{propeta}) must also both hold globally (recall we are on $\ell \geq 2$). From our definition of $\olin$ and $\olinb$ also the equations
\begin{align}
\Omega \slashed{\nabla}_4 \left(r^2 \eblin\right) = 2 r^2 \slashed{\nabla} \olin + r^2\Omega \blin \ \ \ , \ \ \ \Omega \slashed{\nabla}_3 \left(r^2 \elin\right) = 2 r^2 \slashed{\nabla} \olinb - r^2\Omega \bblin
\end{align}
hold. Using~\eqref{propeta} we can also conclude that the $\slashed{div}$ of (\ref{Bianchi3}) (and hence (\ref{Bianchi3}) unconditionally since we already verified the $\slashed{curl}$-equation) must hold globally because it holds on data and is propagated in the $4$-direction.
 
 Differentiating the Codazzi equation for $\xlin$ with $\Omega \slashed{\nabla}_3$ and $\Omega \slashed{\nabla}_4$ shows that also (\ref{dbtc}) and (\ref{uray}) hold. Differentiating (\ref{uray}) with respect to $\Omega \slashed{\nabla}_3$ shows that (\ref{oml2}) holds and (\ref{oml1}) follows from the fact that $\Omega_4 \slashed{\nabla} \olinb = \Omega \slashed{\nabla}_3 \olin$ by the way we defined $\olin$, $\olinb$. 
One now verifies that 
\begin{align}
D:= \Omega \slashed{\nabla}_3 (r^3 \rho +r^3 \slashed{div} \elin) + 2 r^2 \Omega \slashed{\Delta} \olinb - 3M \otxb = 0 
 \end{align}
 holds on the cone $\underline{C}_{v_0}$ and $\Omega \slashed{\nabla}_4 D = 0$, which implies $D=0$ in $\mathcal{M}$, which in turn implies that (\ref{Bianchi5}) holds. We next verify that (\ref{Bianchi8}) holds by noting that it holds on $\mathcal{I}$ by the boundary conditions and is propagated in the $3$-direction. Indeed, we have
 \[
 \Omega \slashed{\nabla}_4 (r^2 \Omega \bblin) = - \Omega \slashed{\nabla}_3 (r^2 \Omega \blin) 
 \]
on $\mathcal{I}$ after replacing $\Omega \slashed{\nabla}_4 = 2 \slashed{\nabla}_T - \Omega \slashed{\nabla}_3$ and inserting the Bianchi equations (\ref{Bianchi2}), (\ref{Bianchi9}) as well as the boundary conditions for $\alin, \ablin$ and $\blin$, $\bblin$. Since also $r \elin = -r \eblin$ and $r^3 \slin=0$ hold on $\mathcal{I}$ we see that (\ref{Bianchi8}) is equivalent to (\ref{Bianchi3}) on $\mathcal{I}$. Differentiating now $E:= \Omega \slashed{\nabla}_4 (r^2 \Omega \bblin) - r^2 \Omega^2 \slashed{\mathcal{D}}_1^\star(\rlin,\slin) - 3 \rho r^2 \Omega^2 \eblin$ in $\Omega \slashed{\nabla}_3$ one obtains $\Omega \slashed{\nabla}_3 \left(\frac{r^2}{\Omega^2} E\right) = 0$ after inserting the equations that have already shown to hold. As $E=0$ on $\mathcal{I}$ we conclude $E=0$ globally. 
 
 Now the $\xblin$-Codazzi equation (\ref{ellipchi}) and the equation (\ref{vray}) can be verified by noting that they hold on $\underline{C}_{v_0}$ and are propagated in the $4$-direction.  
 
We have now verified that our constructed solution satisfies all equations of the system of gravitational perturbations except the Gauss equation (\ref{lingauss}) and the $3$-equations for the metric components. But all these equations hold on $\underline{C}_{v_0}$ and applying $\Omega \slashed{\nabla}_4$ and inserting the already established equations one verifies they propagate to hold in all of $\mathcal{M}$. This finishes the proof of the proposition. 
\end{proof}

\subsection{The initial data normalisation}

Consider a given seed data with associated solution $\mathscr{S}$ from Theorem \ref{theo:wp}.
The main objective of this section is to construct a pure gauge solution $\mathscr{G}$ from the given seed data which 
when added to $\mathscr{S}$ achieves a certain normalisation of the solution $\mathscr{S} + \mathscr{G}$ at the horizon. This normalisation will be crucial in the main theorem. We begin by defining the normalisation followed by a proposition establishing that it can be achieved.

\begin{definition} \label{def:datanormalised}
Consider a smooth seed initial data set 
\[
(\ablin_0; \pblin_0,\Psilinb_0) \ \ ,  \ \ (\underline{B}_0,  R_0; \underline{H}_0; \Olin_0,\bmlin_0, G^{\ell=1}_0,\hat{G}_0) \ \ , \ \  (\mathfrak{m},\mathfrak{a}_{-1},\mathfrak{a}_0,\mathfrak{a}_1),
\]
as in Definition \ref{def:seeddata} and let $\mathscr{S}$ be the unique solution of the system of gravitational perturbations arising from Theorem \ref{theo:wp}. We say that $\mathscr{S}$ is initial data normalised if the following holds for $\mathscr{S}$ on the ingoing initial cone $\underline{C}_{v_0=0}$:
\begin{itemize}
\item $\bmlin=0$ and $\Olin = 0$ on $v=v_0=0$,
\item $\hat{G}_0$ and $G_0^{\ell=1}=0$ and $\underline{H}_0=0$,
\item $\slashed{div} \elin +\rlin = 0$ and $\otx=0$ on the horizon sphere $S^2_{\infty,v_0}$.
\end{itemize}
Moreover, we call the solution initial data normalised with vanishing $\ell=0,1$ modes if in addition the $\ell=0,1$ modes of all geometric quantities of $\mathscr{S}$ vanish.
\end{definition}

The point is that we can always achieve the initial data normalisation:
\begin{proposition} \label{prop:gotoinitialgauge}
Given a solution $\mathscr{S}$ arising from a smooth seed initial data set as in Theorem \ref{theo:wp} there exists a pure gauge solution $\mathscr{G}_f+\mathscr{G}_q$ (computable  in terms of the seed data)  and a linearised Kerr-AdS solution $\mathscr{K}_{\mathfrak{m},\vec{\mathfrak{a}}}$ such that $\mathscr{S} + \mathscr{G}_f + \mathscr{G}_q - \mathscr{K}_{\mathfrak{m},\vec{\mathfrak{a}}}$ is an initial data normalised solution with vanishing $\ell=0,1$ modes. 
\end{proposition}

\begin{proof}
The proof is a small variation of the proof of Proposition \ref{prop:pgz}. From Proposition \ref{prop:pgz} there exist $\mathscr{G}_{\tilde{f}}+\mathscr{G}_q$ and a linearised Kerr-AdS solution $\mathscr{K}_{\mathfrak{m},\vec{\mathfrak{a}}}$ such that $\mathscr{S}^\prime:=\mathscr{S} +\mathscr{G}_{\tilde{f}}+\mathscr{G}_{\tilde{q}}- \mathscr{K}_{\mathfrak{m},\vec{\mathfrak{a}}}$ has trivial seed data and in particular vanishing $\ell=0,1$ modes. We now add to this another $\mathscr{G}_{\hat{f}}$ (supported for $\ell \geq 2$) 
generated by 
\[
\hat{f} (u,\theta) = \frac{ \hat{\lambda}_1 r(u,v_0) + \hat{\lambda}_0\Omega^2 (u,v_0) + \hat{\lambda}_2}{\Omega^2(u,v_0)}
\]
(inducing $\hat{f}_u (u,\theta)=\hat{f}(u,\theta)$ and $\hat{f}_v (v,\theta)=\hat{f}(v,\theta)$ in Lemma \ref{lem:exactsol}), where the $\hat{\lambda}_i$ are functions on the unit sphere satisfying the relation $ \Delta_{\mathbb{S}^2} \hat{\lambda}_0 -  \hat{\lambda}_2=0$. It is clear from the proof of Proposition \ref{prop:pgz} that $\mathscr{S}^\prime+ \mathscr{G}_{\hat{f}}$  still satisfies $\Olin=0$ on $v=v_0$ and $\underline{H}=0$ on $S^2_{u_0,v_0}$ (hence by (\ref{vray}) on all of $\underline{C}_{v_0}$). In view of 
\[
\Omega^2 \hat{f}_u (\infty,\theta) = \hat{\lambda}_1 r_+ + \hat{\lambda}_2 \ \ , \ \ \frac{1}{\Omega^2 r} \partial_u \left(\frac{\Omega^2 \hat{f}_u}{r} \right)(\infty,\theta) = \hat{\lambda}_0 \frac{1}{r_+} \left[ -k^2 -\frac{4M}{r^3} + \frac{1}{r^2}  \right] 
\]
and $f_v(v_0,v_0)= \hat{\lambda}_0$, we now compute from Lemma \ref{lem:exactsol} the horizon sphere relations
\begin{align}
\otx \big|_{\mathscr{S}^\prime + \mathscr{G}_{\hat{f}}} (\infty,v_0,\theta) = \otx \big|_{\mathscr{S}^\prime} (\infty,v_0,\theta) + \frac{2}{r_+^2} \left[ \left(1-\frac{4M}{r_+} -  k^2 r_+^2 \right) (\hat{\lambda}_1 r_+ + \hat{\lambda}_2) + {\Delta}_{S^2} (\hat{\lambda}_1 r_+ + \hat{\lambda}_2) \right] \, , \nonumber 
\end{align}
\begin{align}
(\slashed{div} \elin + \rlin)  \big|_{\mathscr{S}^\prime + \mathscr{G}_{\hat{f}}} (\infty,v_0,\theta) = (\slashed{div} \elin + \rlin)  \big|_{\mathscr{S}^\prime} (\infty,v_0,\theta) -\frac{6M}{r_+^4} (\hat{\lambda}_1 r_+ + \hat{\lambda}_2) +(\Delta_{S^2} \hat{\lambda}_0)\frac{1}{r_+} \left[ -k^2 -\frac{4M}{r_+^3} + \frac{1}{r_+^2} \right] \, . \nonumber 
\end{align}
Recalling that we can eliminate $\hat{\lambda}_0$ from the relation $ \Delta_{S^2} \hat{\lambda}_0 -  \hat{\lambda}_2=0$ it is an algebraic exercise to determine $\hat{\lambda}_1$ and $\hat{\lambda}_2$ such that both $(\slashed{div} \elin + \rlin)  \big|_{\mathscr{S}^\prime + \mathscr{G}_{\hat{f}}} (\infty,v_0,\theta) =0$ and $\otx \big|_{\mathscr{S}^\prime + \mathscr{G}_{\hat{f}}} (\infty,v_0,\theta) = 0$ hold. While $\mathscr{G}_{\hat{f}}$ might have altered ${\hat{G}}_0$ we can simply repeat Step 3 of Proposition \ref{prop:pgz} and add a $\mathscr{G}_{\hat{q}}$ which ensures that $\mathscr{S}^\prime + \mathscr{G}_{\hat{f}} + \mathscr{G}_{\hat{q}}$ satisfies all of the desired properties. Setting $\mathscr{G}_f := \mathscr{G}_{\tilde{f}}+ \mathscr{G}_{\hat{f}}$ and $\mathscr{G}_q := \mathscr{G}_{\tilde{q}}+ \mathscr{G}_{\hat{q}}$ we are done.
\end{proof}

\section{The main results}  \label{sec:maintheorem}

We can finally give a precise formulation of our main results. In Section \ref{sec:teukolsky} we first recall the results from our companion paper \cite{Gra.Hola}  where boundedness and decay bounds on the Teukolsky system (\ref{Teuk1})--(\ref{Teuk2}) have been obtained, independently of the system of gravitational perturbations. This is Theorem \ref{theo:teukolsky} below. These bounds will play a key role in proving our main theorem, which is stated in Section \ref{sec:mtheo} as Theorem \ref{mtheo:boundedness}.

\subsection{Estimates for the gauge invariant quantities: The Teukolsky equations} \label{sec:teukolsky}

Estimates for solutions to the Teukolsky system (\ref{Teuk1})--(\ref{Teuk2}) satisfying the boundary conditions (\ref{bcl1})--(\ref{bcl1b}) have been obtained in our companion paper \cite{Gra.Hola}. We first formulate these results in a form most suitable for the present paper. We recall that in \cite{Gra.Hola}, the Teukolsky equations were expressed in an equivalent form as equations for spin weighted functions instead of symmetric traceless tensors. We briefly recall that equivalence and refer the reader for instance to Section 6 of \cite{Hol.Sha16} for more details.

\subsubsection{Spin-weighted functions vs.~symmetric traceless tensors}

Given the tensors $\alin$ and $\ablin$ and a local orthonormal frame $e_1,e_2$ on the sphere we define the complex scalars
\begin{align} \label{swr}
\alpha^{[-2]} = 2 r^4 \Omega^2  \left( \ablin (e_1,e_1) - i \ablin (e_1,e_2) \right)\ \ \ , \ \ \ {\alpha}^{[+2]} = 2 \Omega^{-2}  \left(\alin(e_1,e_1) {+} i \alin(e_1,e_2)\right) \, ,
\end{align}
which transform like spin-weighted functions of weight $\pm 2$ under a change of orthonormal frame on $S^2$. For the specific frame $e_1 =\frac{1}{r} \partial_\theta$ and $e_2 = \frac{1}{r \sin \theta} \partial_\varphi$ one obtains the Teukolsky equation for $\alpha^{[\pm 2]}$ as stated in \cite{Gra.Hola} by expressing the equations (\ref{Teuk1}) and (\ref{Teuk2}) in frame components. We also note in the notation of  \cite{Gra.Hola} the relations
\begin{align} \label{swr}
\widetilde{\alpha}^{[-2]} = 2 r \Omega^2  \left( \ablin (e_1,e_1) - i \ablin (e_1,e_2) \right)\ \ \ , \ \ \ \widetilde{\alpha}^{[+2]} = 2  r\Omega^2 \left(\alin(e_1,e_1) {+} i \alin(e_1,e_2)\right) \, .
\end{align}
Clearly the estimates on $\widetilde{\alpha}^{[\pm 2]}$ obtained in  \cite{Gra.Hola} directly translate into estimates for the (norms of the) tensors $r\Omega^2 \alin$ and $r\Omega^2 \ablin$. 

\subsubsection{Norms and energies for the gauge invariant quantities} \label{sec:normsa}

To state the estimates of \cite{Gra.Hola} in a form most useful for the present paper we first introduce certain energies on null cones. The underlying reason is that estimating quantities in the system of gravitational perturbations in a double null gauge will typically require control on fluxes on null hypersurfaces. 


To keep notation concise regarding commutations we use the following shorthand notation for derivatives:
\[
\sum_{j=0}^n |\mathfrak{D}^j \xi|^2 = \sum_{j=0}^n \sum_{|i| \leq j} \Big|T^{i_1} \left[\frac{r^2}{\Omega^2} \Omega \slashed{\nabla}_3\right]^{i_2} \left[\Omega \slashed{\nabla}_4\right]^{i_3} \left[r \slashed{\nabla}\right]^{i_4}  \xi\Big|^2 \, ,
\]
where the second sum is over all tuples $i=(i_1,i_2,i_3,i_4)$ with $i_1+i_2+i_3+i_4 \leq j$.

We first define the non-degenerate (near the horizon) outgoing and ingoing commuted energy fluxes (note that the superscript $n+1$ denotes the number of derivatives involved) of a general $S^2_{u,v}$-tensor $\xi$:
\begin{align}
\overline{\mathbb{E}}^{n+1}_u [ \xi ] \left(v_1,v_2\right) &= \sum_{i=0}^n \int_{v_1}^{v_2} \int_{S^2}\left[ |\Omega \slashed{\nabla}_4 \mathfrak{D}^i \xi|^2 + |r \slashed{\nabla} \mathfrak{D}^i \xi|^2 \right] \left(u, \bar{v}\right)  d\bar{v} \sin \theta d\theta d\phi 
\nonumber \\
\overline{\mathbb{E}}^{n+1}_v [\xi] \left(u_1,u_2\right) &= \sum_{i=0}^n \int_{u_1}^{u_2} \int_{S^2}\left[ |r^2 \Omega^{-1} \slashed{\nabla}_3 \mathfrak{D}^i \xi|^2 + |r \slashed{\nabla} \mathfrak{D}^i \xi|^2 \right] \left(u, \bar{v}\right) \Omega^2 d\bar{u} \sin \theta d\theta d\phi 
\end{align}
as well as the outgoing degenerate energy:
\begin{align}
{\mathbb{E}}^{n+1}_u [ \xi] \left(v_1,v_2\right) &= \sum_{i=0}^n \int_{v_1}^{v_2} \int_{S^2}\left[ |\Omega \slashed{\nabla}_4 \mathfrak{D}^i \xi |^2 + \frac{\Omega^2}{r^2} |r \slashed{\nabla} \mathfrak{D}^i \xi|^2 \right] \left(u, \bar{v}\right)  d\bar{v} \sin \theta d\theta d\phi 
\end{align}
The above energies will appear for the (regular both at the horizon $\mathcal{H}^+$ and the conformal boundary $\mathcal{I}$) quantities
\begin{align}
\xi \in \{ r\Omega^2 \alin, \Omega^{-2} r^5 \ablin , \Psilin^R \} \, . 
\end{align}
We shall also need an auxiliary energy on spheres at the conformal boundary $\mathcal{I}$, which arises in the renormalised energy estimates of \cite{Gra.Hola} and is defined only for $n \geq 2$:
\begin{align}
\overline{\mathbb{E}}^{n-1}_\mathcal{I} [\Psilin^R] (v) = \sum_{i=0}^{n-2} \int_{S^2}\left[ | \partial_t  \mathfrak{D}^i \mathcal{L}^{-\frac{1}{2}} (\mathcal{L}-2)^{-\frac{1}{2}} \Psilin^R|^2 + | \mathfrak{D}^i (\mathcal{L}-2)^{-\frac{1}{2}} \Psilin^R|^2 \right] (v,v)  \sin \theta d\theta d\phi \, .
\end{align}
We finally define (for $n \geq 2$) the following initial data master energy on cone $\underline{C}_{v_0=0}$:
\begin{align} \label{masterdata}
 \overline{\mathbb{E}}^{n+1}_{data} [\alin, \ablin] := \overline{\mathbb{E}}^{n+1}_0 [ \Omega^2r \alin] \left(0,\infty\right) + \overline{\mathbb{E}}^{n+1}_0  [ \Omega^{-2}r^5 \ablin]  \left(0,\infty\right) +  \overline{\mathbb{E}}^{n-1}_0  [\Psilin^R]  \left(0,\infty\right) + \overline{\mathbb{E}}^{n-1}_{\mathcal{I}} [\Psilin^R] \left(0\right) \, ,
\end{align}
which contains the energy fluxes of $\alin$ and $\ablin$, the flux of $\Psilin^R$ and a contribution on the sphere at infinity. It is this modified energy which has been shown to propagate for the Teukolsky system in \cite{Gra.Hola}, see Theorem \ref{theo:teukolsky} below. 

\begin{remark}
One could add the terms $\overline{\mathbb{E}}^{n-1}_0  [\Psilin^D]  \left(0,\infty\right) + \overline{\mathbb{E}}^{n-1}_{\mathcal{I}} [\Psilin^D] \left(0\right)$ to the energy (\ref{masterdata}). However, these terms can be shown to be controlled by the first two terms and have hence been omitted. 
\end{remark}

\subsubsection{Estimates for the Teukolsky quantities} 
From the main theorem of our companion paper \cite{Gra.Hola}, we now easily infer the following theorem by translating the estimates on spacelike slices $\Sigma_{t^\star}$ in \cite{Gra.Hola} to estimates on null cones.
\begin{theorem} \label{theo:teukolsky}
We have the following estimates for any $n\geq 3$:
\begin{itemize}
\item Boundedness estimate: For fixed $(u,v) \in \mathcal{M}$ we have
\begin{align} \label{bos}
\overline{\mathbb{E}}^n_v [ \Omega^2r \alin] \left(v,u\right) + \overline{\mathbb{E}}^n_v[ \Omega^{-2}r^5 \ablin]\left(v,u\right)+ {\mathbb{E}}^n_u [ \Omega^2 r\alin] \left(0,v\right)  + {\mathbb{E}}^n_u[ \Omega^{-2}r^5 \ablin] \left(0,v\right) & \nonumber \\
 +\overline{\mathbb{E}}^{n-2}_v [ \Psilin^R ] \left(v,u\right) + {\mathbb{E}}^{n-2}_u[ \Psilin^R ]\left(0,v\right) + \overline{\mathbb{E}}^{n-2}_\mathcal{I} [\Psilin^R] \left(v\right)& \,  \lesssim  \overline{\mathbb{E}}^{n}_{data} [\alin, \ablin] \, .
 \end{align}
\item Decay estimates: Fix an $r=r_0>r_+$. For fixed $v \geq 2$ we denote by $u(r_0,v)$ the $u$-value of the intersection of $\{r=r_0\}$ and the fixed $v$-hypersurface.  We then have for $n \geq 4$
\begin{align} \label{dos}
&\overline{\mathbb{E}}^{n-1}_v [ \Omega^2r \alin] \left(u(r_0,v),\infty\right) + \overline{\mathbb{E}}^{n-1}_{u(r_0,v)} [ \Omega^2 r\alin] \left(v,u(r_0,v)\right) \nonumber \\
+&\overline{\mathbb{E}}^{n-1}_v  [ \Omega^{-2}r^5 \ablin] \left(u(r_0,v),\infty\right) + \overline{\mathbb{E}}^{n-1}_{u(r_0,v)}  [ \Omega^{-2}r^5 \ablin] \left(v,u(r_0,v)\right) \nonumber \\
+& \overline{\mathbb{E}}^{n-3}_v [ \Psilin^R] \left(u(r_0,v),\infty\right) + \overline{\mathbb{E}}^{n-3}_{u(r_0,v)} [ \Psilin^R] \left(v,u(r_0,v)\right) +  \overline{\mathbb{E}}^{n-2}_{\mathcal{I}} [\Psilin^R] \left(u(r_0,v)\right)
  \lesssim \frac{\overline{\mathbb{E}}^{n}_{data} [\alin, \ablin]}{(\log v)^2}   \, , 
\end{align}
with the $\lesssim$ now depending on the $r_0$ (the implicit constant blows up as $r_0 \rightarrow r_+$). 
\item Estimates for the non-degenerate outgoing fluxes near the horizon: We have for any $0 \leq u \leq \infty$, $v_2 \geq v_1+1 \geq v_0$ and $n\geq 3$ the estimates 
\begin{align}
\overline{\mathbb{E}}^n_u [ \Omega^2 r\alin] \left(v_1,v_2\right) + \overline{\mathbb{E}}^{n-2}_u [ \Psilin^R] \left(v_1,v_2\right) \lesssim (v_2-v_1) \overline{\mathbb{E}}^{n}_{data} [\alin, \ablin]  \, ,  \label{gro1}
\end{align}
\begin{align}
\overline{\mathbb{E}}^{n}_u [ \Omega^2 r\alin] \left(v_1,v_2\right)+\overline{\mathbb{E}}^{n-2}_u [ \Psilin^R] \left(v_1,v_2\right) \lesssim \frac{(v_2-v_1)}{\log (v_1)^2} \overline{\mathbb{E}}^{n+1}_{data} [\alin, \ablin] \, .  \label{gro2}
\end{align}
\end{itemize}
\end{theorem}

\begin{remark}
Integrated decay estimates follow as a corollary by integrating the ingoing fluxes in $v$. The integrated decay has again growth like $ \frac{(v_2-v_1)}{\log (v_1)^2}$.
\end{remark}

\begin{remark}
Note that it is the \underline{degenerate outgoing} but \underline{non-degenerate ingoing} energy appearing in the boundedness statement. For the decay estimate, the degenerate and the non-degenerate flux are equivalent because the outgoing flux is always restricted to $r \geq r_0$. In general, the non-degenerate \underline{outgoing} flux has growth as stated in (\ref{gro1})--(\ref{gro2}). To see why the non-degenerate \underline{ingoing} flux behaves better, we recall the (timelike) redshift vectorfield $N$ from \cite{Daf.Rod08} which generates an energy identity whose  bulk term has a good sign in a region $r\leq r_1$ for some $r_1>r_+$. We apply the energy identity in a region bounded by a $\Sigma_{t^\star}$-slice, the horizon and an ingoing cone emanating from the intersection of the $\Sigma_{t^\star}$ slice with the $r=r_1$ hypersurface. This produces control on the desired ingoing flux noting that for $r \geq r_1$ the degenerate energy is equivalent to the non-degenerate one. Applying $N$ globally also yields (\ref{gro1})--(\ref{gro2}) immediately using (\ref{bos}) and (\ref{dos}) respectively (integrated in time) to control the error in $\{r \geq r_1$\} in the corresponding vectorfield identity. 
\end{remark}

\begin{corollary} \label{cor:teuonspheres}
We have for $n\geq 3$ the following estimates on spheres:
\begin{itemize}
\item Boundedness estimates 
\begin{align}
\sum_{i=0}^{n-1} \sup_{u,v} \|\mathfrak{D}^i r \Omega^2 \alin \|^2_{(u,v)}+ \sum_{i=0}^{n-1} \sup_{u,v} \| \mathfrak{D}^i \Omega^{-2} r^5 \ablin \|^2_{(u,v)}  & \lesssim \overline{\mathbb{E}}^{n}_{data} [\alin, \ablin]    \, . 
\end{align}
\item Decay estimates
\begin{align}
\sum_{i=0}^{n-1} \sup_{u,v} \|\mathfrak{D}^i r \Omega^2 \alin \|^2_{(u,v)} +\sum_{i=0}^{n-1} \sup_{u,v} \| \mathfrak{D}^i \Omega^{-2} r^5 \ablin \|^2_{(u,v)} & \lesssim  \frac{1}{(\log v)^2} \overline{\mathbb{E}}^{n+1}_{data} [\alin, \ablin]    \, .
\end{align}
\end{itemize}
\end{corollary}

\begin{remark}
Truly pointwise estimates follow immediately from Sobolev embedding on spheres but are not stated explicitly. We also note that the above estimates are clearly not optimal, as we allow ourselves to lose one derivative for the embedding and another one for decay. 
\end{remark}

\begin{proof}
This follows from $1$-dimensional Sobolev embedding along the ingoing cones for which we control a non-degenerate energy by the previous proposition.
\end{proof}

\subsection{The statement of the main theorem} \label{sec:mtheo}

To state the main theorem, we recall the energies involving the gauge invariant quantities introduced in Section \ref{sec:normsa}. We require one additional (gauge dependent) initial data energy
involving $n$ derivatives of the Ricci coefficients. For $n\geq 3$ we define
\begin{align} \label{Dnorm}
\mathbb{D}^n_0 := 
& \sup_{u \in [0,\infty)} \sum_{\substack{ i+j+k=0 \\ i \leq 2, k \leq 1}}^{n-1}  \big\| \mathcal{A}^{[j]}[\slashed{\nabla}_T]^k  [r\slashed{div}] \left[ r^2\Omega^{-1} \slashed{\nabla}_3\right]^i( \Omega \xlin)) \big\|^2_{u,0} \nonumber \\
+& \sup_{u \in [0,\infty)} \sum_{\substack{i+j=0 \\ i \leq 1}}^{n-1} \big\| \mathcal{A}^{[j]} \left[ r^2 \Omega^{-1} \slashed{\nabla}_3\right]^i [r\slashed{\nabla}] \otx \Omega^{-2} r^2  \big\|^2_{u,0} . 
\end{align}

To state our main boundedness and decay theorem we define the initial master energy involving $n$ derivatives of curvature and Ricci-coefficients
\begin{align} \label{E0e}
\overset{\circ}{\mathbb{E}}{}^n := \mathbb{D}^n_0 + \overline{\mathbb{E}}^{n}_{data} [\alin, \ablin]  \, . 
\end{align}
\begin{theorem} \label{mtheo:boundedness}
Given a solution $\mathscr{S}$ of the system of gravitational perturbations satisfying the boundary conditions as arising from a smooth seed initial data set as in Theorem \ref{theo:wp}, let $\Si=\mathscr{S} + \mathscr{G}_f + \mathscr{G}_q - \mathscr{K}_{\mathfrak{m},\vec{\mathfrak{a}}}$ be the initial data normalised solution with vanishing $\ell=0,1$ modes obtained from Proposition \ref{prop:gotoinitialgauge}. Let the initial energy $\mathring{\mathbb{E}}^n$ in (\ref{E0e}) be defined with respect to the geometric quantities of $\Si$. Then the geometric quantities of the solution $\Si$ satisfy the following estimates. For any weighted Ricci or metric coefficient
\begin{align} \label{rico}
\xi \in \{ \Omega \xlin, \Omega^{-1} r^2 \xblin, \frac{r^2}{\Omega^2} \otx, \frac{r^2}{\Omega^2} \otxb , r \elin, r \eblin,  , \olin, \olinb, r\Olin , \frac{\glinto}{\sqrt{\slashed{g}}} , \glinh, r \Omega^{-2} \bmlin \} \, , 
\end{align}
and any curvature component
\begin{align} \label{cuco}
\Xi \in \{ r\Omega^2 \alin, r^2 \Omega \blin, r^3 \rlin, r^3 \slin, r^4 \Omega^{-1} \bblin, r^5 \Omega^{-2} \ablin \} \, ,
\end{align}
we have for $n\geq 3$ and any $u,v \geq 0$
\begin{align}
\sum_{j=0}^{n} \bigg\| [r\slashed{\nabla}]^j  \xi  \bigg\|^2_{u,v} + \sum_{j=0}^{n-1} \bigg\| [r\slashed{\nabla}]^j  \Xi  \bigg\|^2_{u,v} + \sum_{j=0}^{n-1} \bigg\| [r\slashed{\nabla}]^j  (r^2 \elin, r^2 \eblin, r^2 \Olin) \bigg\|^2_{u,v} &\lesssim \overset{\circ}{\mathbb{E}}{}^n  \, , \label{mainb} \\
\sum_{j=0}^{n} \bigg\| [r\slashed{\nabla}]^j \xi \bigg\|^2_{u,v} + \sum_{j=0}^{n-1} \bigg\| [r\slashed{\nabla}]^j  \Xi  \bigg\|^2_{u,v} + \sum_{j=0}^{n-1} \bigg\| [r\slashed{\nabla}]^j  (r^2 \elin, r^2 \eblin, r^2 \Olin) \bigg\|^2_{u,v}  &\lesssim \frac{\overset{\circ}{\mathbb{E}}{}^{n+1} }{(\log v)^2} \label{maind} \, .
\end{align}
Moreover, for the curvature components $\Xi$, we also obtain for any $v \geq v_0=0$ fixed, the following estimates for the top order fluxes:
\begin{align}
\int_{u=v}^\infty d\bar{u} \frac{\Omega^2}{r^2} \left(  \big\| [r^2\Omega^{-1} \slashed{\nabla}_3] [r\slashed{\nabla}]^{n}  \Xi  \big\|^2_{\bar{u},v} +\big\|  [r\slashed{\nabla}]^{n+1}  \Xi  \big\|^2_{\bar{u},v}  \right) &\lesssim \overset{\circ}{\mathbb{E}}{}^{n+1} \, , \label{setop}
\\
\int_{v}^{v_{f}} d\bar{v} \left(  \big\| [\Omega \slashed{\nabla}_4] [r\slashed{\nabla}]^{n}  \Xi  \big\|^2_{\bar{u},v} +\big\|  [r\slashed{\nabla}]^{n+1}  \Xi  \big\|^2_{{u},\bar{v}}  \right) &\lesssim \overset{\circ}{\mathbb{E}}{}^{n+1} \left(v_f -v \right) \, ,  \label{seto}
\end{align}
for any $u \geq v_f \geq v$. Finally, for $r_0>r_+$ fixed, one may drop the factor of $(v_f-v)$ in (\ref{seto}) if $r(u,v)\geq r_0$. In this case, the $\lesssim$ in (\ref{seto}) will depend on $r_0$. 
\end{theorem}

Theorem \ref{mtheo:boundedness} will be proven in Section \ref{sec:proof}. For now we provide some additional remarks. 

\begin{remark}
The fact that the (formally irregular at the horizon) quantity $\Omega^{-2} \otx$ appears in (\ref{rico}) is due to the fact that the solution is initial data normalised (see Definition \ref{def:datanormalised}) and hence $\otx=0$ on $S^2_{\infty,v_0}$ and, as a consequence of the linearised Raychaudhuri equation (\ref{uray}), on the entire event horizon.
\end{remark}

\begin{remark}
The last sum in (\ref{mainb}) and (\ref{maind}) expresses the fact that if we are willing to lose a derivative, we can show stronger $r$-weighted estimates for $\elin, \eblin$ and even stronger ones for $\elin + \eblin = 2[r\slashed{\nabla}] \Olin$ if we are willing to lose two derivatives. A similar improved estimate with loss holds for $\olin$, $\olinb$ (see Proposition \ref{prop:omega}) but has not been included explicitly in the main theorem. 
\end{remark}

\begin{remark}
The last statement after~\eqref{seto} can be paraphrased by saying that the top order outgoing flux is uniformly bounded provided the outgoing cone is uniformly away from the horizon. 
\end{remark}

\begin{remark}
In the proof of Theorem \ref{mtheo:boundedness} we will also show the estimate $\mathbb{D}^n \lesssim \mathring{\mathbb{E}}^n$ for $n \geq 4$, where $\mathbb{D}^n$ is defined as (\ref{Dnorm}) but replacing the $\sup_{u \in [u_0,\infty)}$ by $\sup_{\mathcal{M}}$. In other words, there is no loss of derivatives in the boundedness statement that we obtain.
\end{remark}

\begin{remark}
Note that contrary to the asymptotically flat case, one obtains here decay of all Ricci coefficients and curvature components even without adding a residual pure gauge solution.
\end{remark}

Finally, we note that Sobolev embedding on spheres gives the following corollary.
\begin{corollary}
We have the following pointwise bounds:
\begin{align}
 |r^{-1} \bmlin| + \Big| \frac{\glinh_{AB}}{\sqrt{\slashed{g}}} \Big| + \Big| \frac{\glinto}{\sqrt{\slashed{g}}}\Big|  \leq \frac{(\overset{\circ}{\mathbb{E}}{}^{4})^\frac{1}{2} }{\log v}   \, , 
\end{align}
\begin{align}
 |r^2 \Olin|  + |\Omega \xlin| + |r^2 \Omega^{-1} \xblin| + |\otx| + |r^2 \Omega^{-2} \otxb| + |r^2 \elin| + |r^2 \eblin|  + | \olin| + |r^2\Omega^{-2} \olinb|\leq \frac{(\overset{\circ}{\mathbb{E}}{}^{4})^\frac{1}{2} }{\log v} \, ,
\end{align}
\begin{align}
|r\Omega^2 \alin| + |r^2 \Omega \blin| + |r^3 \rlin| + |r^3 \slin| + |r^4 \Omega^{-1} \blin| + |r^5 \Omega^{-2} \ablin| \leq \frac{(\overset{\circ}{\mathbb{E}}{}^{4})^\frac{1}{2} }{\log v} \, .
\end{align}
\end{corollary}

\subsection{Future normalising the solution at the conformal boundary}
We can improve the radial decay in our estimates on the solution if we normalise the solution with respect to the conformal boundary by adding a pure gauge solution. The precise statement is the following:

\begin{theorem} \label{mtheo:decay}
With the assumptions of Theorem \ref{mtheo:boundedness}, there exists a further pure gauge solution $\mathscr{G}_f+\mathscr{G}_q$ such that the geometric quantities associated with the corresponding solution 
$\Sf = \Si + \mathscr{G}_f+\mathscr{G}_q$ satisfy the estimates of Theorem \ref{mtheo:boundedness} but now for 
\[
\xi \in \{ r^2 \Omega \xlin, \Omega^{-1} r^4 \xblin, r \otx, r^3 \Omega^{-2} \otxb , r^3 \elin, r^3 \eblin, , r\olin, r\olinb, r^2\Olin , r\frac{\glinto}{\sqrt{\slashed{g}}} , r\glinh, \bmlin \} \, .
\]
Moreover, the pure gauge solutions $\mathscr{G}_f$, $\mathscr{G}_q$ used in the above is uniformly bounded by initial data in the sense that the geometric quantities of the pure gauge solution $\mathscr{G}_f$, $\mathscr{G}_q$ satisfy the estimates of Theorem \ref{mtheo:boundedness}.
\end{theorem}

\begin{remark}
Note the improvement in $r$-weights which is a manifestation of the fact that the solution is now normalised at the conformal boundary. In particular, the metric perturbations now vanish identically on the conformal boundary.
\end{remark}

As the proof involves repeating many of the same estimates of the proof of Theorem \ref{mtheo:boundedness}, we only provide a sketch of Theorem \ref{mtheo:decay} in Section \ref{sec:otherproof}. 

\section{Proof of  the main theorem} \label{sec:proof}

\subsection{Brief overview}
As in the asymptotically flat case, the proof exploits the hierarchical structure of the system of gravitational perturbations in the double null gauge. In Section \ref{sec:transportlemma} we prove the basic transport lemmas that will be invoked throughout the proof when integrating along null cones. 
Since we will always consider the geometric quantities of $\Si$, which have vanishing $\ell=0,1$ modes, the elliptic operators $\mathcal{A}^{[i]}$ have trivial kernel when acting on such a quantity and hence allow estimating the entire $H^{i}$-Sobolev norm of angular derivatives (recall Section \ref{sec:elliptic}). In Section \ref{sec:hoz} we obtain control on certain horizon fluxes of non-gauge invariant quantities from the gauge invariant quantities. These are used in Section \ref{sec:shearestimates} to prove spacetimes boundedness and decay estimates for the ingoing linearised shear. The outgoing linearised shear is then estimated in Section \ref{sec:outgoingshear} using the boundary condition and the transport equation along the ingoing direction. The estimates on the shears allow to estimate various additional components in the system, discussed in Section \ref{sec:consequences}. However, these estimates are somewhat non-optimal in terms of regularity because estimating the ingoing shear required commutation with two transversal derivatives. The regularity is recovered in Section \ref{sec:expansion} applying again a hierarchy of propagation equations and the bounds already obtained.  We conclude the proof of the main theorem in Section \ref{sec:finalrp} after estimating the metric coefficients in Section \ref{sec:metric}.

\subsection{The transport lemmas} \label{sec:transportlemma}

\begin{lemma} \label{lem:bastra1}
Let $\xi$ be an $S^2_{u,v}$ tensor satisfying the propagation equation
\begin{align}
\Omega \slashed{\nabla}_3 \xi = \frac{\Omega^2}{r^2} B 
\end{align}
along the ingoing cone $\underline{C}_v$ (which intersects $\mathcal{I}$ in the sphere $S^2_{v,v}$). Assume $B$ satisfies 
\[
\int_{v}^{\infty} \frac{\Omega^2}{r^2} \|B\|^2_{\bar{u},v} d\bar{u} \lesssim \mathbb{D} \, 
\]
along the cone. Then, provided $ \|\xi \|_{v,v} < \infty$, we have 
\[
\sup_u \|\xi \|_{u,v} \lesssim  \|\xi \|_{v,v} + \sqrt{\frac{\mathbb{D}}{r}} \, .
\]
Moreover, the statement remains true replacing $\mathbb{D}$ by $\frac{\mathbb{D}}{(\log v)^2}$ everywhere.
\end{lemma}

\begin{proof}
Direct consequence of Cauchy-Schwarz and integrability in $u$ of $\frac{\Omega^2}{r^2}=-\frac{\partial_u r}{r^2} = \partial_u \left(\frac{1}{r}\right)$.
\end{proof}

\begin{lemma}\label{lem:bastra2}
Let $\xi$ be an $S^2_{u,v}$ tensor satisfying the propagation equation
\begin{align} \label{pro4}
\Omega \slashed{\nabla}_4 \xi = \frac{\Omega^2}{r^2} B 
\end{align}
along the outgoing cone ${C}_u$ (which intersects $\mathcal{I}$ in the sphere $S^2_{u,u}$). Assume $B$ satisfies
\begin{align} \label{leq1}
\int_{v_1}^{v_2} \|B\|^2_{{u},\bar{v}} d\bar{v} \lesssim \mathbb{D} \cdot \left(v_2-v_1\right) \, 
\end{align}
for any $v_2 \geq v_1 \geq 0$ along the cone and also for some fixed $r_0 > r_+$ the bound
\begin{align} \label{leq2}
\int_{v(r_0,u)}^{u} \|B\|^2_{{u},\bar{v}} d\bar{v} \lesssim_{r_0} \mathbb{D} . \, 
\end{align}
Then we have 
\begin{align} \label{leq3}
\sup_v \|\Omega^{-2} r^2 \xi \|_{u,v} \lesssim  \|\Omega^{-2} r^2 \xi \|_{u,v_0} +\sqrt{\mathbb{D}} \, .
\end{align}
\end{lemma}

\begin{proof}
We write (\ref{pro4}) as
\begin{align}
\Omega \slashed{\nabla}_4 (\Omega^{-2} \xi) + 2 \omega (\Omega^{-2} \xi)  = \frac{1}{r^2} B
\end{align}
and hence, contracting with $(\Omega^{-2} \xi)$ and applying Cauchy's inequality with an $\epsilon$ and an absorption argument we get
\begin{align}
\partial_v \left( \|  (\Omega^{-2} \xi) \|_{u,v}^2 \right) + \left(k^2r + \frac{M}{r^2}\right)  \|  (\Omega^{-2} \xi) \|_{u,v}^2  \lesssim \frac{4}{r^4} \|B\|_{u,v}^2 \, .
\end{align}
We have $\left(k^2r + \frac{M}{r^2}\right) \geq c_1$ and therefore integrating between $v_1$ and $v_2$ yields
\begin{align}
 \|  (\Omega^{-2} \xi) \|_{u,v_2}^2 + c_1 \int_{v_1}^{v_2}  \|  (\Omega^{-2} \xi) \|_{u,v}^2 dv \lesssim  \|  (\Omega^{-2} \xi) \|_{u,v_1}^2 + \mathbb{D} \cdot \left(v_2-v_1\right) \, .
\end{align}
An elementary calculus exercise yields the conclusion (\ref{leq3}) without the factor of $r^2$. To improve the weight near infinity we can integrate (\ref{pro4}) directly from $S^2_{u,v(r_0,u)}$ (where we have already proven the desired bound) using Cauchy-Schwarz and (\ref{leq2}) only as in the proof of Lemma \ref{lem:bastra1}. 
\end{proof}

\begin{lemma} \label{lem:bastra3}
Under the assumptions of Lemma \ref{lem:bastra2}, if $B$ satisfies in addition
 \begin{align} \label{leq1b}
\int_{v_1}^{v_2} \|B\|^2_{{u},\bar{v}} d\bar{v} \lesssim \mathbb{D} \cdot \frac{\left(v_2-v_1\right)}{(\log v_1)^2} \, 
\end{align}
for any $v_2 \geq v_1 \geq 0$ along the cone and also for some fixed $r_0 > r_+$ the bound
\begin{align} \label{leq2b}
\int_{v(r_0,u)}^{u} \|B\|^2_{{u},\bar{v}} d\bar{v} \lesssim_{r_0} \frac{\mathbb{D}}{(\log v(r_0,u)))^2} , \, 
\end{align}
then $\xi$ satisfies along $C_u$ for any $v$ the decay bound
\begin{align} 
 \|\Omega^{-2} r^2 \xi \|_{u,v} \lesssim \frac{ \|\Omega^{-2} r^2 \xi \|_{u,v_0} +\sqrt{\mathbb{D}}}{\log v} \, .
\end{align}
\end{lemma}

\begin{proof}
Simple variation of the previous proof. 
\end{proof}

\subsection{Estimates on the horizon} \label{sec:hoz}
We recall that $\otx=0$ and $\slashed{div} \elin + \rlin = 0$ identically on $\mathcal{H}^+$ by the initial data normalisation, and that all geometric quantities have vanishing $\ell=0,1$ modes.

\begin{proposition} \label{prop:hozfluxestimates}
We have for $n \geq 3$ the following flux estimates on the horizon for any $v_2 \geq v_1 \geq v_0$
\begin{align}
\sum_{i=1}^{n-1} \int_{v_1}^{v_2} d\bar{v} \|\mathcal{A}^{[i]}\slashed{\mathcal{D}}_2^\star (\Omega \blin), \mathcal{A}^{[i]} \slashed{\mathcal{D}}_2^\star \slashed{div} (\Omega \xlin), \mathcal{A}^{[i-1]}\slashed{\mathcal{D}}_2^\star \slashed{\mathcal{D}}_1^\star (r^3 \rlin,r^3 \slin)\|_{\infty,\bar{v}}^2 \lesssim \left(v_2-v_1\right) \overline{\mathbb{E}}^{n}_{data} [\alin, \ablin]  \, ,
\end{align}
\begin{align}
\sum_{i=1}^{n-1} \int_{v_1}^{v_2} d\bar{v} \|\mathcal{A}^{[i]}\slashed{\mathcal{D}}_2^\star (\Omega \blin), \mathcal{A}^{[i]} \slashed{\mathcal{D}}_2^\star \slashed{div} (\Omega \xlin), \mathcal{A}^{[i-1]}\slashed{\mathcal{D}}_2^\star \slashed{\mathcal{D}}_1^\star (r^3 \rlin,r^3 \slin)\|_{\infty,\bar{v}}^2  \lesssim \frac{\left(v_2-v_1\right)  \overline{\mathbb{E}}^{n+1}_{data} [\alin, \ablin] }{(\log v_1)^2} \, . 
\end{align}
We have on the horizon the following estimates on spheres: For $n\geq 3$
\begin{align}
\sum_{i=1}^{n-2} \|\mathcal{A}^{[i]}\slashed{\mathcal{D}}_2^\star (\Omega \blin), \mathcal{A}^{[i]} \slashed{\mathcal{D}}_2^\star \slashed{div} (\Omega \xlin), \mathcal{A}^{[i-1]}\slashed{\mathcal{D}}_2^\star \slashed{\mathcal{D}}_1^\star (r^3 \rlin,r^3 \slin)\|_{\infty,v}^2  \lesssim \overline{\mathbb{E}}^{n}_{data} [\alin, \ablin]  \, ,
\end{align}
\begin{align}
\sum_{i=1}^{n-2} \|\mathcal{A}^{[i]}\slashed{\mathcal{D}}_2^\star (\Omega \blin), \mathcal{A}^{[i]} \slashed{\mathcal{D}}_2^\star \slashed{div} (\Omega \xlin), \mathcal{A}^{[i-1]}\slashed{\mathcal{D}}_2^\star \slashed{\mathcal{D}}_1^\star (r^3 \rlin,r^3 \slin)\|_{\infty,v}^2  \lesssim  \frac{ \overline{\mathbb{E}}^{n+1}_{data} [\alin, \ablin] }{(\log v)^2} \, .
\end{align}
\end{proposition}

\begin{proof}
This follows as in \cite{Daf.Hol.Rod19}, so we merely sketch the argument. To obtain the bounds for $\blin$ we write
\begin{align}
\Big\| \frac{1}{\Omega} \slashed{\nabla}_3 (\Omega^2 r\alin)\Big\|_{\infty,v} = \int_{S^2_{\infty,v}}\sin \theta d\theta d\phi |2 \slashed{\mathcal{D}}_2^\star \Omega \blin + 3 \rho \Omega \xlin|^2 = \int_{S^2_{\infty,v}}\sin \theta d\theta d\phi \left( |2 \slashed{\mathcal{D}}_2^\star \Omega \blin|^2  - 6 \rho |\blin|^2 + 9\rho^2 |\Omega \xlin|^2 \right) \nonumber \, ,
\end{align} 
where we have used the Bianchi equation for the first identity. For the second identity we have integrated the cross term by parts and inserted the Codazzi equation (\ref{ellipchi}) on the horizon ($\slashed{div} \Omega \xlin = -\blin$). Note $\rho=-\frac{2M}{r^3}$, so the expression is indeed coercive. Angular commuted identities are obtained analogously. The result for $\blin$ (and by Codazzi for $\xlin$) now follows from the flux (and sphere) bounds available for the quantity on the left through Theorem \ref{theo:teukolsky}. The result for $(\rlin, \slin)$ follows from the the fact that, by~\eqref{Prel}, on the horizon
\[
\slashed{div} \Psilin = r^5\slashed{div} \slashed{\mathcal{D}}_2^\star \slashed{\mathcal{D}}_1^\star (- \rlin, \slin) +\frac{3}{4} r^5\rho \frac{tr \chi}{\Omega} \slashed{div} \Omega \xlin 
\]
and that we control the flux on the left from Theorem \ref{theo:teukolsky} and the flux of $\xlin$ from the first part of the proof. 
\end{proof}


\begin{corollary} \label{cor:hzf}
On the event horizon $\mathcal{H}^+$, we have for $n\geq 3$ the following flux estimates 
\begin{align}
\sum_{i=0}^{n} \int_{v_1}^{v_2} d\bar{v} \|[r \slashed{\nabla}]^{i} (\Omega \blin), [r \slashed{\nabla}]^{i} (\rlin, \slin) \|_{\infty,\bar{v}}^2  + \sum_{i=0}^{n+1} \int_{v_1}^{v_2} d\bar{v} \|[  [r \slashed{\nabla}]^{i} (\Omega \xlin), [r \slashed{\nabla}]^{i}\elin\|_{\infty,\bar{v}}^2 \lesssim \left(v_2-v_1\right) \overline{\mathbb{E}}^{n}_{data} [\alin, \ablin]  \, , \nonumber \\
\sum_{i=0}^{n} \int_{v_1}^{v_2} d\bar{v} \|[r \slashed{\nabla}]^{i} (\Omega \blin), [r \slashed{\nabla}]^{i} (\rlin, \slin) \|_{\infty,\bar{v}}^2  + \sum_{i=0}^{n+1} \int_{v_1}^{v_2} d\bar{v} \|[  [r \slashed{\nabla}]^{i} (\Omega \xlin), [r \slashed{\nabla}]^{i}\elin\|_{\infty,\bar{v}}^2 \lesssim \frac{\left(v_2-v_1\right) \overline{\mathbb{E}}^{n+1}_{data} [\alin, \ablin] }{(\log v_1)^2} \, . \nonumber
\end{align}
In addition, we have the estimates on spheres 
\begin{align}
\sum_{i=0}^{n-1} \|[r \slashed{\nabla}]^{i} (\Omega \blin), [r \slashed{\nabla}]^{i} (\rlin, \slin) \|_{\infty,{v}}^2  + \sum_{i=0}^{n}  \|[  [r \slashed{\nabla}]^{i} (\Omega \xlin), [r \slashed{\nabla}]^{i}\elin\|_{\infty,{v}}^2 \lesssim  \overline{\mathbb{E}}^{n}_{data} [\alin, \ablin]  \, , \nonumber \\
\sum_{i=0}^{n-1} \|[r \slashed{\nabla}]^{i} (\Omega \blin), [r \slashed{\nabla}]^{i} (\rlin, \slin) \|_{\infty,{v}}^2  + \sum_{i=0}^{n}  \|[  [r \slashed{\nabla}]^{i} (\Omega \xlin), [r \slashed{\nabla}]^{i}\elin\|_{\infty,{v}}^2 \lesssim \frac{ \overline{\mathbb{E}}^{n+1}_{data} [\alin, \ablin] }{(\log v)^2} \nonumber \, . 
\end{align}
In addition, we may add $[r\slashed{\nabla}]^{i-1}\frac{1}{\Omega} \slashed{\nabla}_3 (\Omega \xlin)$ to the list of Ricci-coefficients in the above estimates. 
\end{corollary}
\begin{proof}
In view of $\slashed{div} \elin + \rlin = 0$ and $\slashed{curl} \eta - \slin=0$ on the horizon, we can clearly add the expression $\mathcal{A}^{[i-1]}\slashed{\mathcal{D}}_2^\star \slashed{\mathcal{D}}_1^\star (-r^3 \slashed{div} \elin,r^3 \slashed{curl} \elin)$ to the list of quantities estimated in Proposition \ref{prop:hozfluxestimates}. Since the $\ell\geq 1$ quantities vanish by assumption, the estimates follows from standard elliptic estimates.  The last claim is immediate from restricting the linearised null structure equation (\ref{chih3}) for $ \Omega^{-1} \slashed{\nabla}_3 (\Omega \xlin)$ to the horizon and using the previous bounds. 
\end{proof}


\subsection{Preliminary estimates on the shears}
\subsubsection{The outgoing shear} \label{sec:shearestimates}
We give a brief overview. One starts with the quantity $\xlin$, which according to (\ref{tchi}) satisfies 
\begin{align} \label{shipro}
\Omega \slashed{\nabla}_4 (\xlin \Omega r^2) -2\omega  (\xlin \Omega r^2)  = - \alin \Omega^2 r^2 \, . 
\end{align}
Commuting twice with the operator $\Omega^{-1} \slashed{\nabla}_3$ turns the exponentially growth factor ($-2\omega$, a blueshift) into a decay factor ($+2\omega$, a redshift), after which the equation can be integrated forwards in $v$ using the flux bound for $\alin$ and derivatives thereof on the right hand side.\footnote{The lower order terms that arise in the commutation can be integrated by parts and produce terms of a good sign and boundary terms which are controlled on the horizon from Proposition \ref{prop:hozfluxestimates}.} Roughly speaking, since the structure of the horizon does not depend on the cosmological constant, the estimates near the horizon go through exactly as in \cite{Daf.Hol.Rod19}. Away from the horizon, where $\Omega^2$ is uniformly bounded away from zero, one can of course integrate directly
\begin{align} \label{shear4dir}
\Omega \slashed{\nabla}_4 (\xlin \Omega^{-1} r^2) = - \alin r^2
\end{align}
all the way to infinity. This gives in particular that $\frac{r^2 \xlin}{\Omega} \sim r \xlin$ is uniformly bounded on $\mathcal{I}$. Now let us turn to the details. We first derive the key estimate near the horizon.

\begin{proposition}
There exists an $2r_+>r_1>r_+$ such that the following estimate holds for any $v \geq v_0$ and any $j \in \mathbb{N}_0$, $k \in \{0,1\}$
\begin{align} 
\sum_{i=0}^2 \sup_{v \in (v_1,v_2)} \int_{u(r_1,v)}^\infty du \Omega^2 \big\| \mathcal{A}^{[j]} [\slashed{\nabla}_T]^k \left[ \Omega^{-1} \slashed{\nabla}_3\right]^i  (r^2 \Omega \xlin)) \big\|^2_{u,v} \nonumber \\
+ \sum_{i=0}^2 \int_{v_1}^{v_2} dv \int_{u(r_1,v)}^\infty du \Omega^2 \big\| \mathcal{A}^{[j]}[\slashed{\nabla}_T]^k \left[ \Omega^{-1} \slashed{\nabla}_3\right]^i (r^2 \Omega \xlin)) \big\|^2_{u,v} \nonumber \\
 \lesssim \int_{u(r_1,v_1)}^\infty du \Omega^2 \big\|  \mathcal{A}^{[j]}[\slashed{\nabla}_T]^k \left[\Omega^{-1} \slashed{\nabla}_3\right]^2 (r^2 \Omega \xlin)) \big\|^2_{u,v_1} + (v_2-v_1) \sup_{v \in (v_1,v_2)} \overline{\mathbb{E}}_v^{j+k+2} [ \Omega^2r \alin](u(v,r_1),\infty) \label{mids} \, . 
\end{align}
\end{proposition}

\begin{proof}
We provide a sketch of the proof as the argument in entirely analogous to the proof of Proposition 13.3.2 in \cite{Daf.Hol.Rod19}. From (\ref{shipro}) we derive upon commutation the identity 
\begin{align} \label{comid}
  \begin{aligned}
    \Omega \slashed{\nabla}_4 (\Omega[\Omega^{-1} \slashed{\nabla}_3]^2 \xlin \Omega r^2) + \omega  (\Omega[\Omega^{-1} \slashed{\nabla}_3]^2 \xlin \Omega r^2) & \\ + \left(-\frac{4M}{r^3} + \frac{1}{
        l^2}\right) \Omega  \frac{1}{\Omega} \slashed{\nabla}_3 (\xlin \Omega r^2) - \Omega \frac{12M}{r^4} (\xlin \Omega r^2) =&  - \Omega [\Omega^{-1} \slashed{\nabla}_3]^2 ( \alin \Omega^2 r^2) \, .
  \end{aligned}
\end{align}
Of course commutation with $\mathcal{A}^{[j]}$ and $\slashed{\nabla}_T^k$ is trivial and is omitted. 

One now proceeds as in\cite{Daf.Hol.Rod19} multiplying  (\ref{comid}) with $(\Omega[\Omega^{-1} \slashed{\nabla}_3]^2 \xlin \Omega r^2)$ and integrating over the spacetime region $\mathcal{M} \cap \{ r \leq r_0\} \cap \{v_1 \leq v \leq v_2\}$. The terms in the first line of (\ref{comid}) will produce the good desired terms in (\ref{mids}) (as well as the term first term on the right). The term on the right hand side of (\ref{comid}) can be dealt with by Cauchy-Schwarz borrowing a bit from the good spacetime term on the left. Finally, for the terms on the left in the second line of (\ref{comid}) we proceed as in \cite{Daf.Hol.Rod19}: For the first term we integrate by parts, controlling the (bad-signed) boundary term on the horizon by Corollary \ref{cor:hzf}, while the resulting spacetime term has a good sign. For the second term we use Cauchy-Schwarz and a Hardy inequality, which provides control on the $i=0$ terms on the left hand side of (\ref{mids}) in terms of the higher order quantities that have already been controlled using again the control of the horizon fluxes in Corollary \ref{cor:hzf}.
\end{proof}

Note that we can bound the first term on the right in (\ref{mids}) by (\ref{Dnorm}). A simple pigeonhole principle applied to (\ref{mids}) yields

\begin{proposition} \label{prop:shearhozf}
For any $v \geq 0$ and $n \geq 3$, $k \in \{0,1\}$:
\begin{align}
\sum_{j=0}^{n-2} \sum_{i=0}^2 \int_{u(r_1,v)}^\infty du \Omega^2 \big\| \mathcal{A}^{[j]} [\slashed{\nabla}_T]^k \left[ \Omega^{-1} \slashed{\nabla}_3\right]^i(r^2 \Omega \xlin)) \big\|^2_{u,v} &\lesssim  \overline{\mathbb{E}}^{n+k}_{data} [\alin, \ablin]  +\mathbb{D}^{n+k}_0 \, ,  \\
\sum_{j=0}^{n-2} \sum_{i=0}^2 \int_{u(r_1,v)}^\infty du \Omega^2 \big\| \mathcal{A}^{[j]} [\slashed{\nabla}_T]^k  \left[ \Omega^{-1} \slashed{\nabla}_3\right]^i  (r^2 \Omega \xlin)) \big\|^2_{u,v}  &\lesssim \frac{ \overline{\mathbb{E}}^{n+k+1}_{data} [\alin, \ablin] +\mathbb{D}^{n+k+1}_0}{(\log v)^2}\, .
\end{align}
\end{proposition}

We can now easily globalise the result to the entire exterior taking care of the correct $r$-weights and also improve to an estimate on spheres. 
\begin{proposition}  \label{prop:inshearfinal}
For any $v \geq 0$ and $n \geq 3$, $k \in \{0,1\}$
\begin{align}
\sum_{j=0}^{n-2} \sum_{i=0}^2 \big\| \mathcal{A}^{[j]}[\slashed{\nabla}_T]^k  \left[ r^2\Omega^{-1} \slashed{\nabla}_3\right]^i( \Omega \xlin)) \big\|^2_{u,v} &\lesssim \overline{\mathbb{E}}^{n+k}_{data} [\alin, \ablin]  +\mathbb{D}^{n+k}_0  \, , \nonumber \\
\sum_{j=0}^{n-2} \sum_{i=0}^2 \big\| \mathcal{A}^{[j]} [\slashed{\nabla}_T]^k  \left[ r^2\Omega^{-1} \slashed{\nabla}_3\right]^i  ( \Omega \xlin)) \big\|^2_{u,v}  &\lesssim  \frac{ \overline{\mathbb{E}}^{n+k+1}_{data} [\alin, \ablin] +\mathbb{D}^{n+k+1}_0}{(\log v)^2} \, .
\end{align}
Moreover, both estimates also hold replacing $[\slashed{\nabla}_T]^k  \left[ r^2\Omega^{-1} \slashed{\nabla}_3\right]^i$ by $[\slashed{\nabla}_T]^{k+2}$.
\end{proposition}

\begin{proof}
Note that the last claim follows immediately from the two estimates and the null structure equation (\ref{tchi}), hence we can focus on proving the two estimates. We will also suppress the (trivial) commutation with angular and $T$-derivatives in the algebra for the proof. 

We first obtain these estimates in the region $r \leq r_1$. If we restricted the sum over $i$ to run from $0$ to $1$ only, both estimates follow directly from Proposition \ref{prop:shearhozf} and the fundamental theorem of calculus (which loses one $\Omega^{-1} \slashed{\nabla}_3$ derivative, hence the restriction to $i \leq 1$). To show it for $i=2$ one revisits (\ref{comid}), now written as 
\begin{align} \label{comid2}
\Omega \slashed{\nabla}_4 ([\Omega^{-1} \slashed{\nabla}_3]^2 \xlin \Omega r^2) + 2\omega  (\Omega[\Omega^{-1} \slashed{\nabla}_3]^2 \xlin \Omega r^2) & \nonumber \\ + \left(-\frac{4M}{r^3} + \frac{1}{
l^2}\right)   \frac{1}{\Omega} \slashed{\nabla}_3 (\xlin \Omega r^2) -  \frac{12M}{r^4} (\xlin \Omega r^2) =&  - [\Omega^{-1} \slashed{\nabla}_3]^2 ( \alin \Omega^2 r^2) \, .
\end{align}
Using the estimates already shown we obtain (after trivially commuting the above with $\mathcal{A}^{[j]}[\slashed{\nabla}_T]^k$ )
\begin{align}
\|\mathcal{A}^{[j]}[\slashed{\nabla}_T]^k [\Omega^{-1} \slashed{\nabla}_3]^2 \xlin \Omega r^2\|_{u,v_2}^2  + \int_{v_1}^{v_2} dv \| \mathcal{A}^{[j]}[\slashed{\nabla}_T]^k [\Omega^{-1} \slashed{\nabla}_3]^2 \xlin \Omega r^2\|_{u,\bar{v}}^2 \nonumber \\
\lesssim \|\mathcal{A}^{[j]}[\slashed{\nabla}_T]^k [\Omega^{-1} \slashed{\nabla}_3]^2 \xlin \Omega r^2\|_{u,v_1}^2 +  \left(  \overline{\mathbb{E}}^{j+k+2}_{data} [\alin, \ablin]  +\mathbb{D}^{j+k}_0\right) (v_2 -v_1) \, , \nonumber
\end{align}
from which the estimates follow also for $i=2$ by simple ODE theory.

Having established the estimates of the proposition for $r \leq r_1$, we integrate the (appropriately commuted) linearised null structure equation $\Omega \slashed{\nabla}_4 ( \Omega^{-1} r^2 \xlin)= -r^2\alin$ from $r=r_1$ towards infinity to deduce the result also for $r \geq r_1$. (Observe that near infinity $\Omega^{-1}r^2 \sim \Omega \sim r$.) Note that $\Omega \slashed{\nabla}_3$ and also $\mathcal{A}^{[j]}$ commute trivially on the left and that $\Omega \slashed{\nabla}_3 \sim r^2 \Omega^{-1}\slashed{\nabla}_3$ near infinity. 
\end{proof}


\subsubsection{The ingoing shear} \label{sec:outgoingshear}
The boundary condition (\ref{bclxi}) now allows us to integrate (\ref{tchi}) written as   
\begin{align} \label{wrii}
\Omega \slashed{\nabla}_3 (\xblin \Omega^{-1} r^2) = - \Omega^{-2} r^5 \ablin \frac{\Omega^2}{r^3}
\end{align}
directly from the boundary to produce global uniform bounds on the (regular at $\mathcal{H}^+$) quantity $r^2 \Omega^{-1} \xblin$:
\begin{proposition}  \label{prop:outshearfinal}
For any $v \geq 0$ and $n \geq 3$, $k \in \{0,1\}$:
\begin{align}
\sum_{j=0}^{n-2} \sum_{i=0}^2 \big\| \mathcal{A}^{[j]}  [\slashed{\nabla}_T]^k\left[ \Omega \slashed{\nabla}_4\right]^i (r^2 \Omega^{-1} \xblin)) \big\|^2_{u,v} &\lesssim  \overline{\mathbb{E}}^{n+k}_{data} [\alin, \ablin]  +\mathbb{D}^{n+k}_0   \, , \nonumber \\
\sum_{j=0}^{n-2} \sum_{i=0}^2 \big\| \mathcal{A}^{[j]}  [\slashed{\nabla}_T]^k \left[\Omega \slashed{\nabla}_4\right]^i (r^2 \Omega^{-1} \xblin)) \big\|^2_{u,v}  &\lesssim \frac{\overline{\mathbb{E}}^{n+k+1}_{data} [\alin, \ablin]  +\mathbb{D}^{n+k+1}_0}{(\log v)^2} \, . \nonumber
\end{align}
\end{proposition}


\begin{proof}
We apply Lemma \ref{lem:bastra1} to (\ref{wrii}) and the $\Omega \slashed{\nabla}_4$-commuted (\ref{wrii}).
The only thing which is not immediate is the initial condition for the $\Omega \slashed{\nabla}_4$ commuted estimate. For this we note that the boundary condition $\slashed{\nabla}_T (r\xlin - r\xblin)=0$ on $\mathcal{I}$ translates into $\Omega \slashed{\nabla}_3 ( r \xlin) -  \Omega \slashed{\nabla}_4 ( r \xblin) = 0$  on $\mathcal{I}$ using that $\alin r^2, \ablin r^2$ vanish on the boundary. Indeed,
\begin{align}
0 = \slashed{\nabla}_T (r\xlin - r\xblin) &= \Omega \slashed{\nabla}_3 (r \xlin) - \Omega \slashed{\nabla}_4 (r\xblin) + \Omega \slashed{\nabla}_4 \left(\frac{r^2}{\Omega^2} \Omega \xlin \frac{\Omega}{r} \right) -  \Omega \slashed{\nabla}_3 \left(\frac{r^2}{\Omega^2} \Omega \xblin \frac{\Omega}{r} \right) \nonumber \\ 
&= \Omega \slashed{\nabla}_3 (r \xlin) - \Omega \slashed{\nabla}_4 (r\xblin) - \frac{\Omega}{r} \alin r^2 + \frac{\Omega}{r} \ablin r^2 + \left(r^2 \partial_r \frac{\Omega}{r} \right)  \left(\Omega \xlin + \Omega \xblin\right) ,
\end{align}
and the last $4$-terms vanish on the boundary by Definition~\ref{def:ssogp}. 
Similarly $\slashed{\nabla}_T^2 (r\xlin - r\xblin)=0$ on $\mathcal{I}$ translates into $[\Omega \slashed{\nabla}_3]^2 ( r \xlin) -  [\Omega \slashed{\nabla}_4]^2 ( r \xblin) = 2k^3\alin r^3$  on $\mathcal{I}$ since up to terms vanishing in the limit on $\mathcal{I}$ we have 
\begin{align}
0 = \slashed{\nabla}^2_T (r\xlin - r\xblin) &= [\Omega \slashed{\nabla}_3]^2 (r \xlin) - [\Omega \slashed{\nabla}_4]^2 \xblin + \Omega \slashed{\nabla}_3 \left(\Omega \slashed{\nabla}_4 \left(\frac{r^2}{\Omega^2} \Omega \xlin \frac{\Omega}{r} \right) \right) -  \Omega \slashed{\nabla}_4 \left(\Omega \slashed{\nabla}_3 \left(\frac{r^2}{\Omega^2} \Omega \xblin \frac{\Omega}{r} \right) \right) \nonumber 
\end{align}
and hence on $\mathcal{I}$
\begin{align}
0= [\Omega \slashed{\nabla}_3]^2 (r \xlin) - [\Omega \slashed{\nabla}_4]^2 \xblin + \Omega \slashed{\nabla}_3  \left( - \frac{\Omega}{r} \alin r^2 + \frac{r^2}{\Omega^2} \Omega \xlin \left( \Omega \slashed{\nabla}_4 \frac{\Omega}{r} \right) \right) -\Omega \slashed{\nabla}_4  \left( - \frac{\Omega}{r} \ablin r^2 + \frac{r^2}{\Omega^2} \Omega \xblin \left( \Omega \slashed{\nabla}_3 \frac{\Omega}{r} \right) \right) \nonumber \, ,
\end{align}
from which the claim on the boundary follows. 

This means that the initial condition in the commuted estimate is always controlled from Proposition \ref{prop:inshearfinal} and Corollary \ref{cor:teuonspheres}. Furthermore, the flux (on constant $v$) when integrating the transport equation from the boundary requires $m$-derivatives of $\ablin$ to estimate $m$ derivatives of $\xblin$ and $m+1$ derivatives of $\ablin$ if one would like to see $\log$-decay. 
\end{proof}

\subsubsection{Improving the weights near infinity}
We now establish a few improved estimates for certain combinations (and derivatives of) $\xlin$ and $\xblin$, which will be helpful in establishing estimates for the torsion later. %
Specifically, we claim the following:
\begin{proposition} \label{prop:improvedchidifference}
We have for $n \geq 3$ the following estimates for any $k \in \{0,1\}$:
\begin{align} \label{sum1}
\sum_{j=0}^{n-2} \sup_{r \geq 8M} \Big\| [\slashed{\nabla}_T]^k \mathcal{A}^{[j]}  r \left(\Omega \xlin - \Omega \xblin \right)\Big\|^2_{u,v} \lesssim   \overline{\mathbb{E}}^{n+k}_{data} [\alin, \ablin] +\mathbb{D}^{n+k-1}_0
\end{align}
and
\begin{align} \label{diff1}
\sum_{j=0}^{n-2} \sup_{r \geq 8M} \Big\| [\slashed{\nabla}_T]^k \mathcal{A}^{[j]}   r \left(\Omega \slashed{\nabla}_3 \left(\Omega \xlin \right)  - \Omega \slashed{\nabla}_4 \left(\Omega \xblin \right)  \right)\Big\|^2_{u,v} \lesssim  \overline{\mathbb{E}}^{n+k}_{data} [\alin, \ablin]  +\mathbb{D}^{n+k}_0
\end{align}
and
\begin{align} \label{longe}
\sum_{j=0}^{n-3} \sup_{r \geq 8M} \Big\| \mathcal{A}^{[j]} r^2 \left(\Omega \slashed{\nabla}_3 (\Omega \xlin )  + \Omega \slashed{\nabla}_4 (\Omega \xblin ) +  \left(\Omega tr \underline{\chi}\right) \Omega \xlin + \left( \Omega tr \chi\right) \Omega \xblin  \right)\Big\|^2_{u,v}
 \lesssim \overline{\mathbb{E}}^{n}_{data} [\alin, \ablin]  +\mathbb{D}^{n}_0  .
\end{align}
Moreover, we can add an additional factor of $(\log v)^{-2}$ on the right, provided we replace $n$ by $n+1$ in the $\overline{\mathbb{E}}_{data}$ energies on the right. 
\end{proposition}

\begin{proof}
To keep the notation in the proof tidy, we ignore the trivial angular commutation by $[\slashed{\nabla}_T]^k \mathcal{A}^{[j]}$ during the proof, which can be trivially inserted in all equations below. 
We define the shear along the boundary 
\[
X(t,\theta):= r \xlin (t,t,\theta) = r \xblin(t,t,\theta) \, , 
\]
with the last equality following from the boundary condition. Integrating (\ref{tchi}) from the boundary $\mathcal{I}$, we deduce after an integration by parts the identities (here $\hat{A},\hat{B}$ are the components in an \emph{orthonormal} frame!\footnote{Recall the formula for the \emph{coordinate} components $(\Omega \slashed{\nabla}_3 \xi)_{AB} = \partial_u \xi_{AB} - \Omega tr \underline{\chi} \xi_{AB}=r^2 \partial_u (r^{-2} \xi_{AB})=r^2 \partial_u (\xi_{\hat{A}\hat{B}})$ and hence $ (\Omega \slashed{\nabla}_3 \xi)_{\hat{A}\hat{B}}=\partial_u (\xi_{\hat{A}\hat{B}}$).})
\begin{align}
\frac{r^2}{\Omega} \xlin_{\hat{A}\hat{B}} (u,v,\theta) 
&= X_{\hat{A}\hat{B}}(u, \theta)+ \frac{1}{2}\frac{\alin_{\hat{A}\hat{B}} r^5}{\Omega^2}\frac{1}{r^2} (u,v,\theta)  + \int_v^u \left[ \frac{1}{2r^2} \partial_v \frac{\alin_{\hat{A}\hat{B}} r^5}{\Omega^2}  \right] (u, \bar{v}, \theta) d\bar{v} \, ,
\nonumber \\
\frac{r^2}{\Omega} \xblin_{\hat{A}\hat{B}} (u,v,\theta) 
&= X_{\hat{A}\hat{B}} (v, \theta)- \frac{1}{2}\frac{\ablin_{\hat{A}\hat{B}} r^5}{\Omega^2}\frac{1}{r^2} (u,v,\theta)  + \int_v^u \left[ \frac{1}{2r^2} \partial_u \frac{\ablin_{\hat{A}\hat{B}} r^5}{\Omega^2}  \right] (\bar{u}, v, \theta) d\bar{u} .
\end{align} 
It follows that
\begin{align}
\Omega \slashed{\nabla}_3 \left(\Omega \xlin \right)_{\hat{A}\hat{B}} 
&= \left(\frac{2}{r} - \frac{6M}{r^2}\right) \Omega \xlin_{\hat{A}\hat{B}} + \frac{\Omega^2}{r^2} \left( \dot{X}_{\hat{A}\hat{B}}(u,\theta) + \alin_{\hat{A}\hat{B}} r^2  +\frac{1}{2r^2} T \left(\frac{\alin_{\hat{A}\hat{B}} r^5}{\Omega^2}\right) + \int_v^u \partial_u \left[ \frac{1}{2r^2} \partial_v \frac{\alin_{\hat{A}\hat{B}} r^5}{\Omega^2}  \right] (u, \bar{v}, \theta) d\bar{v} \right) \nonumber
\end{align}
and similarly
\begin{align}
\Omega \slashed{\nabla}_4 \left(\Omega \xblin \right)_{\hat{A}\hat{B}}
&= -\left(\frac{2}{r} - \frac{6M}{r^2}\right) \Omega \xblin_{\hat{A}\hat{B}} + \frac{\Omega^2}{r^2} \left( \dot{X}_{\hat{A}\hat{B}}(v,\theta) + \ablin_{\hat{A}\hat{B}} r^2  -\frac{1}{2r^2} T \left(\frac{\ablin_{\hat{A}\hat{B}} r^5}{\Omega^2}\right) + \int_v^u \partial_v \left[ \frac{1}{2r^2} \partial_u \frac{\ablin_{\hat{A}\hat{B}} r^5}{\Omega^2}  \right] (\bar{u}, v, \theta) d\bar{u} \right) \, . \nonumber
\end{align}
We also have
\begin{align}
 \left(\Omega tr \underline{\chi}\right) \Omega \xlin_{\hat{A}\hat{B}} + \left( \Omega tr \chi\right) \Omega \xblin_{\hat{A}\hat{B}} = +\frac{2\Omega^4}{r^3}  \Big(  X_{\hat{A}\hat{B}}(v, \theta) &- \frac{1}{2}\frac{\ablin_{\hat{A}\hat{B}} r^5}{\Omega^2}\frac{1}{r^2} (u,v,\theta)  + \int_v^u \left[ \frac{1}{2r^2} \partial_u \frac{\ablin_{\hat{A}\hat{B}} r^5}{\Omega^2}  \right] (\bar{u}, v, \theta) d\bar{u} \nonumber \\
 -X_{\hat{A}\hat{B}}(u, \theta)&- \frac{1}{2}\frac{\alin_{\hat{A}\hat{B}} r^5}{\Omega^2}\frac{1}{r^2} (u,v,\theta)  - \int_v^u \left[ \frac{1}{2r^2} \partial_v \frac{\alin_{\hat{A}\hat{B}} r^5}{\Omega^2}  \right] (u, \bar{v}, \theta) d\bar{v} \, 
\Big) , \nonumber 
\end{align}
from which (\ref{sum1}) is already immediate after using Taylor’s theorem (as well as (\ref{rrstarrel}))
\[
\| r(u,v) \left[ X(v,\theta) - X(u,\theta)\right]\|_{u,v} \lesssim \big| r(u,v) (v-u) \big| \sup_u \|\dot{X}(u,\theta)\|_{u,u} \lesssim \sup_u \|T(r \xlin)\|_{u,u} \
\]
and using (the last claim of) Proposition \ref{prop:inshearfinal} for the term on the right. The bound (\ref{diff1}) follows similarly form Propositions \ref{prop:inshearfinal} and \ref{prop:outshearfinal} as well as another application of Taylor’s theorem, now for the $\dot{X}$-difference. 


To prove (\ref{longe}), the key is (besides applying (\ref{sum1}) and estimates from Theorem \ref{theo:teukolsky}) to establish
\begin{align}
\sup_{r \geq 8M} \Big\| r^2 \left( \dot{X}(u,\theta) + \ablin r^2 +  \dot{X}(v,\theta) + \alin r^2\right) + 2r^3 \left(   X(v, \theta) - \frac{1}{2}\frac{\ablin r^5}{\Omega^2}\frac{1}{r^2} - X(u, \theta) - \frac{1}{2}\frac{\alin r^5}{\Omega^2}\frac{1}{r^2} \right) \Big\| \nonumber \\
 \lesssim \sup_{\mathcal{M} \cap \{r \geq 8M\}} \|r^3 \alin\|_{u,v} +   \sup_{\mathcal{M} \cap \{r \geq 8M\}} \|r^3 \ablin\|_{u,v} + \sup_u \| T^3 (r \xlin) \|_{u,u} \, .
\end{align}
This follows from Taylor expanding 
\begin{align}
X(u,\theta) &= X(v,\theta) + \dot{X}(v,\theta)(v-u) +  \frac{1}{2} \ddot{X}(v,\theta)(v-u)^2 + R_X \, , 
\nonumber \\
\dot{X}(u,\theta) &= \dot{X}(v,\theta) + \ddot{X}(v,\theta)(v-u) + \tilde{R}_X \nonumber
\end{align}
and using (\ref{rrstarrel}) as well as Taylor’s theorem for the remainders. Note that Proposition \ref{prop:inshearfinal} controls at most three $T$-derivatives of $\xlin$ on the boundary, so we cannot commute further. 
\end{proof}

\subsection{Some immediate consequences} \label{sec:consequences}
In this section we obtain estimates for all curvature components and the torsions $\elin$, $\eblin$ from the preliminary estimates on the shears and the gauge invariant quantities. These estimates are not optimal in terms of regularity (caused by the loss in the estimate for the shears) and will be improved later. 

\subsubsection{Estimating curvature one-forms}
From the Bianchi identities rewritten as (\ref{rewriteBianchi}) we see that the estimates on $\xlin$ and $\xblin$ in Propositions \ref{prop:inshearfinal} and \ref{prop:outshearfinal} respectively, will provide estimates on $\blin$ and $\bblin$ using Theorem \ref{theo:teukolsky} and Corollary \ref{cor:teuonspheres}. For now we state these (immediate) estimates on spheres (recall that the $\ell\leq 1$ modes are trivial by assumption), deferring top order flux bounds to a later point (namely once the regularity in Propositions \ref{prop:inshearfinal} and \ref{prop:outshearfinal} has been improved further).

\begin{proposition} \label{prop:betaS2}
We have on any sphere $S^2_{u,v}$ for $n\geq 3$ the estimates
\begin{align}
\sum_{i=0}^{n-1}  \| \mathcal{A}^{[i]} (\Omega r^2 \blin) \|_{u,v}^2 +  \sum_{i=0}^{n-1} \| \mathcal{A}^{[i]} (\Omega^{-1} r^4 \bblin) \|_{u,v}^2 &\lesssim  \overline{\mathbb{E}}^{n}_{data} [\alin, \ablin]  +\mathbb{D}^{n}_0 \, ,  \\
\sum_{i=0}^{n-1}  \| \mathcal{A}^{[i]} (\Omega r^2 \blin) \|_{u,v}^2 + \sum_{i=0}^{n-1} \| \mathcal{A}^{[i]} (\Omega^{-1} r^4 \bblin) \|_{u,v}^2 &\lesssim  \frac{\overline{\mathbb{E}}^{n+1}_{data} [\alin, \ablin]  +\mathbb{D}^{n+1}_0}{(\log v)^2} \, .  
\end{align}
\end{proposition}

\subsubsection{Estimating curvature scalars}

We now recall the equation (\ref{Prel2}) noting that the left hand side of (\ref{Prel2}) can be written as a linear combination of the regular (at  both $\mathcal{H}^+$ and $\mathcal{I}$) $\Omega^{-2} r^5 \ablin$, $\Omega \slashed{\nabla}_4 (\Omega^{-2} r^5 \ablin)$ and $[\Omega \slashed{\nabla}_4]^2 (\Omega^{-2} r^5 \ablin)$ with smooth and uniformly bounded coefficients.  Combining this with fact the estimate (\ref{sum1}) we directly obtain:


\begin{proposition} \label{prop:rhosigma}
We have on any sphere $S^2_{u,v}$ for $n \geq 3$ the estimates:
\begin{align}
\sum_{i=0}^{n-3} \| \mathcal{A}^{[i]} r^5 \slashed{\mathcal{D}}_2^\star  \slashed{\mathcal{D}}_1^\star \left(-\rlin, \slin\right)  \|_{u,v}^2 &\lesssim \overline{\mathbb{E}}^{n}_{data} [\alin, \ablin]  +\mathbb{D}^{n}_0  \, ,  \\
\sum_{i=0}^{n-3} \| \mathcal{A}^{[i]} r^5 \slashed{\mathcal{D}}_2^\star  \slashed{\mathcal{D}}_1^\star \left(-\rlin, \slin\right)  \|_{u,v}^2 &\lesssim \frac{ \overline{\mathbb{E}}^{n+1}_{data} [\alin, \ablin] +\mathbb{D}^{n+1}_0}{(\log v)^2} \, .
\end{align}
We also have the top-order ingoing flux bound
\begin{align}
\sum_{i=0}^{n-2} \int_v^u d\bar{u} \frac{\Omega^2}{r^2} \| \mathcal{A}^{[i]} r^5 \slashed{\mathcal{D}}_2^\star  \slashed{\mathcal{D}}_1^\star \left(-\rlin, \slin\right)  \|_{u,\bar{v}}^2
\lesssim  \overline{\mathbb{E}}^{n}_{data} [\alin, \ablin]  +\mathbb{D}^{n}_0  \, 
\end{align}
and the top-order outgoing flux bound
\begin{align} \label{laste}
\sum_{i=0}^{n-2} \int_{v_1}^{v_2} d\bar{v} \| \mathcal{A}^{[i]} r^5 \slashed{\mathcal{D}}_2^\star  \slashed{\mathcal{D}}_1^\star \left(-\rlin, \slin\right)  \|_{u,\bar{v}}^2   \lesssim \left(v_2-v_1\right) \left[  \overline{\mathbb{E}}^{n}_{data} [\alin, \ablin]  +\mathbb{D}^{n}_0  \right] \, .
\end{align}
Finally, for any fixed $r_0>r_+$ and the sphere $S^2_{u,v_1}$ lying in the region $r\geq r_0$ we have the uniform estimate
\begin{align} \label{laste2}
\sum_{i=0}^{n-2} \int_{v_1}^{v_2} d\bar{v} \| \mathcal{A}^{[i]} r^5 \slashed{\mathcal{D}}_2^\star  \slashed{\mathcal{D}}_1^\star \left(-\rlin, \slin\right)  \|_{u,\bar{v}}^2   \lesssim_{r_0} \overline{\mathbb{E}}^{n}_{data} [\alin, \ablin]  +\mathbb{D}^{n}_0  \, .
\end{align}
\end{proposition}

\begin{proof}
The bounds on spheres follow directly from the identity (\ref{Prel}), estimate (\ref{sum1}) of Proposition \ref{prop:improvedchidifference} and Corollary \ref{cor:teuonspheres}. For the fluxes one uses Theorem \ref{theo:teukolsky} instead of Corollary \ref{cor:teuonspheres}.
\end{proof}


\subsubsection{Estimates on the torsion}

\begin{proposition} \label{prop:etaetab}
For any $v \geq 0$ and $n \geq 3$ and $k \in \{0,1\}$ and $i \in \{0,1\}$
\begin{align}
\sum_{j=0}^{n-1+i} \big\| [\slashed{\nabla}_T]^k \mathcal{A}^{[j]} (r^{2-i}\elin, r^{2-i}\eblin) \big\|^2_{u,v}  &\lesssim  \overline{\mathbb{E}}^{n+k}_{data} [\alin, \ablin]  +\mathbb{D}^{n+k}_0  \, ,  \label{ghj1} \\
\sum_{j=0}^{n-1-i} \big\| [\slashed{\nabla}_T]^k \mathcal{A}^{[j]} (r^{2-i}\elin, r^{2-i}\eblin) \big\|^2_{u,v}  &\lesssim \frac{ \overline{\mathbb{E}}^{n+k+1}_{data} [\alin, \ablin]  +\mathbb{D}^{n+k+1}_0 }{(\log v)^2} \label{ghj2}  \, .
\end{align}
We also have 
\begin{align}
\sum_{j=0}^{n-2} \big\|  \mathcal{A}^{[j]} (r^3\elin + r^3\eblin) \big\|^2_{u,v}  &\lesssim  \overline{\mathbb{E}}^{n}_{data} [\alin, \ablin]  +\mathbb{D}^{n}_0 \, , \\
\sum_{j=0}^{n-2} \big\|  \mathcal{A}^{[j]} (r^3\elin + r^3\eblin) \big\|^2_{u,v}  &\lesssim \frac{ \overline{\mathbb{E}}^{n+1}_{data} [\alin, \ablin]  +\mathbb{D}^{n+1}_0 }{(\log v)^2}  \, .
\end{align}
\end{proposition}

\begin{proof}
We show the estimate for $k=0$. For $i=1$, (\ref{ghj1}) and (\ref{ghj2}) for $(r \elin, r \eblin)$ follow immediately from the equation (\ref{chih3}) and (\ref{chih3b}) after inserting the estimates of Propositions \ref{prop:inshearfinal}, \ref{prop:outshearfinal} and Proposition \ref{prop:improvedchidifference}. In particular this already establishes immediately all estimates claimed in the region $r\leq 8M$, so we can focus on establishing the estimates for $i=0$ in the region $r \geq 8M$ for the remainder of the proof.

Replacing $(r^2 \elin, r^2 \eblin)$ by $(r^2 \elin - r^2 \eblin)$ (i.e.~only looking at the difference) both estimates follow after taking the difference of (\ref{chih3}) and (\ref{chih3b}) and using the estimate (\ref{diff1}) of Proposition \ref{prop:improvedchidifference}.

To show the actual (\ref{ghj1}) and (\ref{ghj2}) (i.e.~the estimate for $r^2 \elin$ and $r^2\eblin$ individually) we integrate (\ref{propeta}) backwards from the boundary, where $\elin$ and $\eblin$ are known to vanish by having established control on $(r \elin, r \eblin)$ at the beginning of the proof. Inserting the estimate for $(r^2 \elin - r^2 \eblin)$ already established and Proposition \ref{prop:betaS2} to control the right hand side, the estimates  (\ref{ghj1}) and (\ref{ghj2}) follow.

The last two estimates follows from adding (\ref{chih3}) and (\ref{chih3b}) and using the estimate (\ref{longe}).
\end{proof}

\subsubsection{Estimates on the lapses}

The following is an immediate corollary of Proposition \ref{prop:etaetab}: 
\begin{corollary} \label{cor:olin}
For any $v \geq 0$ and $n \geq 3$
\begin{align}
\sum_{j=0}^{n-1} \big\| r \mathcal{A}^{[j]} [r\slashed{\nabla}] (\Olin)   \big\|^2_{u,v} + \sum_{j=0}^{n-2} \big\| r^2 \mathcal{A}^{[j]} [r\slashed{\nabla}] (\Olin)   \big\|^2_{u,v}
&\lesssim  \overline{\mathbb{E}}^{n}_{data} [\alin, \ablin]  +\mathbb{D}^{n}_0  \, , \nonumber \\
\sum_{j=0}^{n-1} \big\| r \mathcal{A}^{[j]} [r\slashed{\nabla}]  (\Olin)   \big\|^2_{u,v} + \sum_{j=0}^{n-2} \big\| r^2 \mathcal{A}^{[j]} [r\slashed{\nabla}]  (\Olin)   \big\|^2_{u,v} &\lesssim \frac{ \overline{\mathbb{E}}^{n+1}_{data} [\alin, \ablin]  +\mathbb{D}^{n+1}_0 }{(\log v)^2} \, . \nonumber
\end{align}

\end{corollary}

\begin{proof}
All estimates follows straight from the definition $2\slashed{\nabla} (\Olin) = \elin + \eblin$ and using Proposition \ref{prop:etaetab}. 
\end{proof}

An estimate for $\olinb, \olin$ is easily obtained from the relations
\begin{align} \label{Teta}
2 \slashed{\nabla}_T \elin = -\Omega \blin - \Omega \bblin + \frac{\Omega^2}{r} (\elin + \eblin) + 2 \slashed{\nabla} \olinb \ \ \ , \ \ \ 2 \slashed{\nabla}_T \eblin = \Omega \blin + \Omega \bblin - \frac{\Omega^2}{r} (\elin + \eblin) + 2 \slashed{\nabla} \olin \, 
\end{align}
and previous bounds on the geometric quantities:

\begin{proposition} \label{prop:omega}
For any $v \geq 0$ and $n \geq 3$ and $i \in \{0,1\}$
\begin{align}
\sum_{j=0}^{n-2} \big\| r^ i \mathcal{A}^{[j]} [r\slashed{\nabla}] \olinb r^2 \Omega^{-2}   \big\|^2_{u,v} +\sum_{j=0}^{n-2} \big\| r^i \mathcal{A}^{[j]} [r\slashed{\nabla}] \olin   \big\|^2_{u,v} &\lesssim  \overline{\mathbb{E}}^{n+i}_{data} [\alin, \ablin]  +\mathbb{D}^{n+i}_0  \, , \nonumber \\
\sum_{j=0}^{n-3} \big\| r^i \mathcal{A}^{[j]} [r\slashed{\nabla}] \olinb r^2 \Omega^{-2}   \big\|^2_{u,v} +\sum_{j=0}^{n-3} \big\| r^i \mathcal{A}^{[j]} [r\slashed{\nabla}] \olin   \big\|^2_{u,v} &\lesssim \frac{ \overline{\mathbb{E}}^{n+i}_{data} [\alin, \ablin]  +\mathbb{D}^{n+i}_0 }{(\log v)^2} \, . \nonumber 
\end{align}
\end{proposition}

\begin{remark} \label{rem:soc}
In particular $\olin$ and $\olinb$ vanish on $\mathcal{I}$. Summing the relations (\ref{Teta}) one can also show that $r \olin + r \olinb$ vanishes on $\mathcal{I}$ as follows from $r^2 \slashed{\nabla}_T (\elin + \eblin)$ vanishing on $\mathcal{I}$. However, this requires controlling one more $T$-derivative on data and is hence omitted. See also Remark \ref{rem:betterbounds}.
\end{remark}


\subsection{Estimates on the expansion and improving the regularity} \label{sec:expansion}
We write the linearised Raychaudhuri equation (\ref{uray}) as 
\begin{align}\label{eq:Raybis}
\Omega \slashed{\nabla}_4 \left(\otx r^2 \Omega^{-2} - 4r \Olin \right) = -4 \Omega^2 (\Olin) \, . 
\end{align}
Commuting with $n$ angular derivatives and also with $\Omega^{-1} \slashed{\nabla}_3$ we deduce (note the  factor of $\Omega^{-2}$):

\begin{proposition} \label{prop:trx}
For $n \geq 3$ we have on any sphere $S^2_{u,v}$ for $i \in \{0,1\}$:
\begin{align}
\sum_{j=0}^{n-1-i} \big\| \mathcal{A}^{[j]} [r\slashed{\nabla}] \otx \Omega^{-2} r^2 \xcancel{ (\log r)^{-1+i}  } \big\|^2_{u,v} &\lesssim \overline{\mathbb{E}}^{n}_{data} [\alin, \ablin]  +\mathbb{D}^{n}_0  \label{est:trx1}\, ,  \\
\sum_{j=0}^{n-2-i} \big\| \mathcal{A}^{[j]} [r\slashed{\nabla}] \otx \Omega^{-2}  r^2 \xcancel{(\log r)^{-1+i}}    \big\|^2_{u,v} &\lesssim \frac{ \overline{\mathbb{E}}^{n}_{data} [\alin, \ablin]  +\mathbb{D}^{n}_0 }{(\log v)^2}  \, .\label{est:trx2}
\end{align}
We also have
\begin{align}
\sum_{j=0}^{n-2-i} \big\| \frac{r^{i+1}}{\Omega^2} [\Omega \slashed{\nabla}_3] \mathcal{A}^{[j]} [r\slashed{\nabla}] \otx   \big\|^2_{u,v} &\lesssim  \overline{\mathbb{E}}^{n}_{data} [\alin, \ablin]  +\mathbb{D}^{n}_0  \, ,  \label{est:trx3}\\
\sum_{j=0}^{n-3-i} \big\|  \frac{r^{i+1}}{\Omega^2} [\Omega \slashed{\nabla}_3] \mathcal{A}^{[j]} [r\slashed{\nabla}] \otx   \big\|^2_{u,v} &\lesssim \frac{ \overline{\mathbb{E}}^{n}_{data} [\alin, \ablin]  +\mathbb{D}^{n}_0 }{(\log v)^2} \, . \label{est:trx4} 
\end{align}
\end{proposition}
\begin{proof}
  With the $\log$-terms absent (for $i=1$), the estimates~\eqref{est:trx1},~\eqref{est:trx2} claimed are an immediate consequence of applying the transport Lemmas \ref{lem:bastra2} and \ref{lem:bastra3} to Equation~\eqref{eq:Raybis} using the estimates of Corollary \ref{cor:olin} with the strong $r^2$ weight. Estimates~\eqref{est:trx3},~\eqref{est:trx4} are obtained similarly, by commuting~\eqref{eq:Raybis} with $\Omega\slashed{\nabla}_3$ and using the estimates of Corollary \ref{cor:olin} and Proposition~\ref{prop:omega} with the strong $r$-weights. With the $\log$-terms present (\emph{i.e.} for $i=0$), estimates~\eqref{est:trx1},~\eqref{est:trx2},~\eqref{est:trx3},~\eqref{est:trx4} follow along the same lines but using the estimates with the weaker $r$-weights in Corollary~\ref{cor:olin} and Proposition~\ref{prop:omega}. By a slight variation of the transport lemmas which we leave to the reader, this generates $\log$-terms in~\eqref{est:trx1} and~\eqref{est:trx2} and the claimed weights in~\eqref{est:trx3},~\eqref{est:trx4}. The $\log$-terms will be removed immediately after the next proposition, which only uses the estimates of the current proposition with the $\log$-terms. 
\end{proof}

The above estimate (with the $\log$-terms) leads to an improvement Proposition \ref{prop:inshearfinal} via the Codazzi equations and previous bounds:

\begin{proposition}  \label{prop:outshearfinal2}
For any $v \geq 0$ and $n \geq 3$, $k\in\{0,1\}$ 
\begin{align}
\sum_{i+j \leq n, i \leq 2} \big\| \mathcal{A}^{[j]} [\slashed{\nabla}_T]^k  \left[ r^2\Omega^{-1} \slashed{\nabla}_3\right]^i( \Omega \xlin)) \big\|^2_{u,v} &\lesssim  \overline{\mathbb{E}}^{n+k}_{data} [\alin, \ablin]  +\mathbb{D}^{n+k}_0   \, ,  \\
\sum_{i+j \leq n, i \leq 2} \big\| \mathcal{A}^{[j]} [\slashed{\nabla}_T]^k \left[ r^2\Omega^{-1} \slashed{\nabla}_3\right]^i( \Omega \xlin)) \big\|^2_{u,v}  &\lesssim \frac{ \overline{\mathbb{E}}^{n+k+1}_{data} [\alin, \ablin]  +\mathbb{D}^{n+k+1}_0 }{(\log v)^2}  \, . 
\end{align}
\end{proposition}

\begin{proof}
For the term $i=2$ in the sum on the left, the estimate is a direct consequence of Proposition \ref{prop:inshearfinal} so we focus on $i \leq 1$. We write the Codazzi equation (\ref{ellipchi}) as
\begin{align} \label{codau1}
r\slashed{div} \Omega \xlin = -\Omega^2 \eblin  -\Omega r \blin + \frac{r}{2} \slashed{\nabla} \otx  \, .
\end{align}
Applying $\Omega^{-1}\slashed{\nabla}_3$ to both sides and inserting the relevant Bianchi and null structure equations, we obtain
\begin{align} \label{codau2}
r\slashed{div} \left(  \Omega^{-1} \slashed{\nabla}_3(\Omega \xlin) \right)= \left(2k^2r+\frac{2M}{r^2}\right) \eblin + \left(\frac{\Omega^2}{r}  \left( \elin - \eblin\right)  - \Omega \bblin\right) +  \Omega^{-1} \slashed{\nabla}_3 \left( \frac{r}{2} \slashed{\nabla} \otx \right) \nonumber \\
-\left( r \slashed{\mathcal{D}}_1^\star \left(-\rlin \, , \, \slin \, \right) + 3\rho r \, \elin \right) + \blin \Omega \, .
\end{align}
Therefore, with non-optimal $r$-weights (insert a factor of $\frac{1}{r^2}$ in the norms on the left), the desired estimates follow immediately from (\ref{codau1}) and (\ref{codau2}) after using the estimates of Propositions \ref{prop:betaS2}, \ref{prop:rhosigma}, \ref{prop:etaetab} and \ref{prop:trx}. To obtain the weights near infinity as claimed one integrates (\ref{shear4dir}) from some fixed $r_0$ as in the proof of Proposition \ref{prop:inshearfinal}.
\end{proof}

\begin{corollary} \label{cor:logrem}
The first two estimates of Proposition \ref{prop:trx} hold without the logarithmic term. 
\end{corollary}
\begin{proof}
Estimate $\otx$ from (\ref{codau1}) now using the improved estimate on $\Omega\xlin$ from Proposition \ref{prop:outshearfinal2}.
\end{proof}

We can now also improve the estimate on $\xblin$ of Proposition \ref{prop:outshearfinal} using that we now control more derivatives of $\xlin r$ and hence (by the boundary condition $\xlin r = \xblin r$ which holds with arbitrary many tangential derivatives by the smoothness of the solution) of $\xblin r$ on the boundary.

\begin{proposition}  \label{prop:inshearfinal2}
For any $v \geq 0$ and $n \geq 3$
\begin{align}
\sum_{i+j \leq n, i \leq 2} \big\| \mathcal{A}^{[j]} \left[ \Omega \slashed{\nabla}_4\right]^i (r^2 \Omega^{-1} \xblin)) \big\|^2_{u,v} &\lesssim  \overline{\mathbb{E}}^{n}_{data} [\alin, \ablin]  +\mathbb{D}^{n}_0  \, , \\
\sum_{i+j \leq n, i \leq 2} \big\| \mathcal{A}^{[j]} \left[\Omega \slashed{\nabla}_4\right]^i (r^2 \Omega^{-1} \xblin)) \big\|^2_{u,v}  &\lesssim \frac{ \overline{\mathbb{E}}^{n+1}_{data} [\alin, \ablin]  +\mathbb{D}^{n+1}_0  }{(\log v)^2} \, .
\end{align}
\begin{proof}
Revisit the proof of Proposition \ref{prop:outshearfinal} using that we control the higher order ``initial” term on the boundary $\mathcal{I}$ by the boundary condition and Proposition \ref{prop:outshearfinal2}. 
\end{proof}

\end{proposition}

A direct corollary, using (\ref{ellipchi}) pointwise, is
\begin{corollary}  \label{cor:inmeancurv}
For any $v \geq 0$ and $n \geq 3$
\begin{align}
\sum_{j=1}^{n-1} \big\| \mathcal{A}^{[j]}  (r\slashed{\nabla}  \Omega^{-2}r^2 \otxb) \big\|^2_{u,v} &\lesssim  \overline{\mathbb{E}}^{n}_{data} [\alin, \ablin]  +\mathbb{D}^{n}_0 \, , \\ 
\sum_{j=1}^{n-2} \big\| \mathcal{A}^{[j]}  (r\slashed{\nabla}  \Omega^{-2}r^2 \otxb) \big\|^2_{u,v}  &\lesssim \frac{ \overline{\mathbb{E}}^{n}_{data} [\alin, \ablin]  +\mathbb{D}^{n}_0  }{(\log v)^2} \, .
\end{align}
\end{corollary}

\begin{remark} \label{rem:toc}
One can prove control on $| r\otx  - r\otxb|$ by subtracting the two Codazzi equations and using previous bounds but we will not need this here. See again Remark \ref{rem:betterbounds}.
\end{remark}

With Propositions \ref{prop:inshearfinal2} and \ref{prop:outshearfinal2}, Proposition \ref{prop:etaetab} also improves by one order in regularity (at the cost of less $r$-weights) and, as a corollary of the relation  (\ref{Teta}), also our estimate on $\olin$:

\begin{proposition} \label{prop:etaetab2}
We have for $k\in\{0,1\}$ the estimates
\begin{align}
\sum_{j=0}^{n-k} \big\| \mathcal{A}^{[j]} [\slashed{\nabla}_T]^k (r\elin, r\eblin) \big\|^2_{u,v}  &\lesssim  \overline{\mathbb{E}}^{n}_{data} [\alin, \ablin]  +\mathbb{D}^{n}_0  \, ,  \label{ghj1b} \\
\sum_{j=0}^{n-k} \big\| \mathcal{A}^{[j]} [\slashed{\nabla}_T]^k (r\elin, r\eblin) \big\|^2_{u,v}  &\lesssim \frac{ \overline{\mathbb{E}}^{n+1}_{data} [\alin, \ablin]  +\mathbb{D}^{n+1}_0 }{(\log v)^2} \label{ghj2b}  \, .
\end{align}
Moreover, 
\begin{align}
\sum_{j=0}^{n} \big\| \mathcal{A}^{[j]} (\olin, \olinb) \big\|^2_{u,v}  &\lesssim  \overline{\mathbb{E}}^{n}_{data} [\alin, \ablin]  +\mathbb{D}^{n}_0  \, ,  \label{ghj1c} \\
\sum_{j=0}^{n} \big\| \mathcal{A}^{[j]}  (\olin, \olinb) \big\|^2_{u,v}  &\lesssim \frac{ \overline{\mathbb{E}}^{n+1}_{data} [\alin, \ablin]  +\mathbb{D}^{n+1}_0 }{(\log v)^2} \label{ghj2c}  \, .
\end{align}
\end{proposition}

We finally obtain the top order flux bounds for $\blin$ and $\bblin$ from the improved estimates on the shear of Propositions \ref{prop:inshearfinal2} and \ref{prop:outshearfinal2}:
\begin{proposition} \label{prop:betabfinalflux}
We have on any sphere $S^2_{u,v}$ for $n \geq 3$ the ingoing flux bounds:
\begin{align}
\sum_{i=0}^{n} \int_v^u d\bar{u} \frac{\Omega^2}{r^2} \| \mathcal{A}^{[i]} (\Omega r^2 \blin, \Omega^{-1} r^4 \bblin) \|_{\bar{u},{v}}^2   
&\lesssim  \overline{\mathbb{E}}^{n}_{data} [\alin, \ablin]  +\mathbb{D}^{n}_0  \, , 
\\
\sum_{i=0}^{n} \int_v^u d\bar{u} \frac{\Omega^2}{r^2} \| \mathcal{A}^{[i]} (\Omega r^2 \blin, \Omega^{-1} r^4 \bblin) \|_{\bar{u},{v}}^2  
&\lesssim  \frac{ \overline{\mathbb{E}}^{n+1}_{data} [\alin, \ablin]  +\mathbb{D}^{n+1}_0  }{(\log v)^2} \, .
\end{align}
We also have the top-order outgoing flux bound
\begin{align}
\sum_{i=0}^{n} \int_{v_1}^{v_2} d\bar{v} \| \mathcal{A}^{[i]} (\Omega r^2 \blin, \Omega^{-1} r^4 \bblin)  \|_{u,\bar{v}}^2   \lesssim \left(v_2-v_1\right) \left[ \overline{\mathbb{E}}^{n}_{data} [\alin, \ablin]  +\mathbb{D}^{n}_0  \right] \, .
\end{align}
Finally, for any fixed $r_0>r_+$ and the sphere $S^2_{u,v_1}$ lying in the region $r\geq r_0$ we have the uniform estimate
\begin{align} 
\sum_{i=0}^{n} \int_{v_1}^{v_2} d\bar{v} \| \mathcal{A}^{[i]} (\Omega r^2 \blin, \Omega^{-1} r^4 \bblin)  \|_{u,\bar{v}}^2   \lesssim_{r_0}  \overline{\mathbb{E}}^{n}_{data} [\alin, \ablin]  +\mathbb{D}^{n}_0  \, .
\end{align}
\end{proposition} 

\begin{proof}
This is a direct consequence of the relations (\ref{rewriteBianchi}), the estimates on $\alin$ and $\ablin$ of Theorem \ref{theo:teukolsky} and the estimates on the shear of Propositions \ref{prop:inshearfinal2} and \ref{prop:outshearfinal2}. For the last statement recall that $\Omega^2 \geq c_{r_0} > 0$ for $r_0 > r_+$. 
\end{proof}

\subsection{Estimates on the metric components} \label{sec:metric}

We can now integrate the propagation equation (\ref{bequat}) for $\bmlin$ using Lemmas \ref{lem:bastra2}, \ref{lem:bastra3} in conjunction with Proposition \ref{prop:etaetab2}  to obtain a bound on the shift:
\begin{proposition} \label{prop:bmlin}
For any sphere $S^2_{u,v}$ and $n \geq 3$
\begin{align} \label{bmlinbound}
\sum_{j=0}^{n} \big\| \mathcal{A}^{[j]} (r \Omega^{-2} \bmlin)) \big\|^2_{u,v}  &\lesssim  \overline{\mathbb{E}}^{n}_{data} [\alin, \ablin]  +\mathbb{D}^{n}_0  \, ,  \\
\sum_{j=0}^{n} \big\| \mathcal{A}^{[j]} (r \Omega^{-2} \bmlin)) \big\|^2_{u,v}  &\lesssim  \frac{ \overline{\mathbb{E}}^{n+1}_{data} [\alin, \ablin]  +\mathbb{D}^{n+1}_0  }{(\log v)^2}\, .
\end{align}
\end{proposition}

The propagation equation for the linearised metric in the outgoing direction (\ref{stos}) immediately yields after controlling the relevant flux from Proposition \ref{prop:trx}

\begin{proposition} \label{prop:metricd}
For any sphere $S^2_{u,v}$ and $n \geq 3$
\begin{align}
\sum_{j=0}^{n-1} \bigg\| \mathcal{A}^{[j]} [r\slashed{\nabla}]  \frac{\glinto}{\sqrt{\slashed{g}}}  \bigg\|^2_{u,v} &\lesssim  \overline{\mathbb{E}}^{n}_{data} [\alin, \ablin]  +\mathbb{D}^{n}_0  \, , \\
\sum_{j=0}^{n-1} \bigg\| \mathcal{A}^{[j]} [r\slashed{\nabla}]  \frac{\glinto}{\sqrt{\slashed{g}}}  \bigg\|^2_{u,v}  &\lesssim \frac{ \overline{\mathbb{E}}^{n+1}_{data} [\alin, \ablin]  +\mathbb{D}^{n+1}_0  }{(\log v)^2}  \, .
\end{align}
\end{proposition}
To estimate $\glinh$ we cannot integrate (\ref{stos2}) directly as $\Omega \xlin$ is not uniformly in $L^1_v$. We instead estimate it from the Gauss curvature. 
\begin{proposition} \label{prop:metricd2}
For any sphere $S^2_{u,v}$ and $n \geq 3$
\begin{align}
\sum_{j=0}^{n} \bigg\| \mathcal{A}^{[j]} \glinh \bigg\|^2_{u,v} &\lesssim  \overline{\mathbb{E}}^{n}_{data} [\alin, \ablin]  +\mathbb{D}^{n}_0  \, ,  \\
\sum_{j=0}^{n} \bigg\| \mathcal{A}^{[j]} \glinh  \bigg\|^2_{u,v}  &\lesssim \frac{ \overline{\mathbb{E}}^{n+1}_{data} [\alin, \ablin]  +\mathbb{D}^{n+1}_0  }{(\log v)^2}  \, .
\end{align}
\end{proposition}

\begin{proof}
By elliptic estimates, it suffices to prove these estimates replacing $\mathcal{A}^{[j]}$ by $\mathcal{A}^{[j-2]}r^2 \slashed{div} \slashed{div}$ and $\mathcal{A}^{[j-2]}r^2 \slashed{curl} \slashed{div}$ in each sum and letting the sum start at $j=2$. (We slightly abuse notation here and let $\mathcal{A}$ also act on scalars by taking $r\slashed{\mathcal{D}}_1^\star$, and on one forms by taking $r\slashed{\mathcal{D}}_2^\star$.) For the latter part we can integrate  (\ref{stos2}) commuted with $\mathcal{A}^{[j-2]}r^2 \slashed{curl} \slashed{div}$ because from the Codazzi equation (\ref{ellipchi}) we have
\[
\mathcal{A}^{[j-2]}r^2 \slashed{curl} \slashed{div} (\Omega \xlin) = \mathcal{A}^{[j-2]} \left(\Omega^2 r \slin - r^2 \slashed{curl} \Omega \blin \right)
\]
and hence
\begin{align}
\Omega \slashed{\nabla}_4 \left( \mathcal{A}^{[j-2]}r^2 \slashed{curl} \slashed{div} \, \glinh - 2\mathcal{A}^{[j-2]} r^3 \slin \right) = 2\Omega^2  \mathcal{A}^{[j-2]} r \slin \, .
\end{align}
We are in the situation of Lemmas \ref{lem:bastra2} and \ref{lem:bastra3} (their assumptions valid from Proposition \ref{prop:rhosigma}) and we hence obtain the desired estimate for the $\slashed{curl}\slashed{div}$-part. For the $\slashed{div} \slashed{div}$-part we use the linearised Gauss equation:
\begin{align}
r^3 \slashed{\nabla} \left( -\frac{1}{2}\slashed{\Delta} tr_{\slashed{g}} \glin +\slashed{div} \slashed{div} \glinh - \frac{1}{r^2}tr_{\slashed{g}} \glin \right) = r^3 \slashed{\nabla} \Klin  = -r^3 \slashed{\nabla} \rlin -r^2 \slashed{div} (\Omega \xblin) + r^2 \Omega  \bblin + r^2 \slashed{div} (\Omega \xlin) + r^2 \Omega \blin \, . \nonumber
\end{align}
The estimate now follows by solving this for $r^3 \slashed{\nabla} \slashed{div} \slashed{div} \glinh $ (which has vanishing spherical average) and inserting the estimates from Propositions \ref{prop:metricd}, \ref{prop:outshearfinal2} and \ref{prop:inshearfinal2} as well as Propositions \ref{prop:betaS2} and \ref{prop:rhosigma}.
\end{proof}
 
\subsection{Concluding the proof of Theorem \ref{mtheo:boundedness}} \label{sec:finalrp}
The estimates (\ref{mainb}) and (\ref{maind}) claimed in Theorem \ref{mtheo:boundedness} for the Ricci-coefficients $\xi$ are implied 
\begin{itemize}
\item for $\xlin$ and $\xblin$ by Propositions \ref{prop:outshearfinal2} and \ref{prop:inshearfinal2},
\item for $\otx$ and $\otxb$  by Proposition \ref{prop:trx} and Corollary \ref{cor:inmeancurv},
\item for $\elin$ and $\eblin$ by Propositions \ref{prop:etaetab} and \ref{prop:etaetab2}, and for $\olin$ and $\olinb$ by \ref{prop:omega} and \ref{prop:etaetab2},
\item for the metric quantities $\Olin$ by Corollary \ref{cor:olin}, for $\bmlin$ by Proposition \ref{prop:bmlin} and for $\glinh$, $\glinto$ by Propositions \ref{prop:metricd} and \ref{prop:metricd2} respectively.
\end{itemize} 
The estimates (\ref{mainb}) and (\ref{maind}) claimed in Theorem \ref{mtheo:boundedness} for the curvature components $\Xi$ are implied by Corollary \ref{cor:teuonspheres} for $\alin$ and $\ablin$, by Proposition \ref{prop:rhosigma} for $\rlin$ and $\slin$ and by Proposition \ref{prop:betabfinalflux} for $\blin$ and $\bblin$. 

Finally, the top order curvature bounds claimed in (\ref{setop}) and (\ref{seto}) are implied by Theorem \ref{theo:teukolsky} for $\alin$ and $\ablin$, by Proposition \ref{prop:rhosigma} for $\rlin$ and $\slin$ and by Proposition \ref{prop:betabfinalflux} for $\blin$ and $\bblin$. 

\section{Normalising the solution at infinity: Proof of Theorem \ref{mtheo:decay}} \label{sec:otherproof}
We now prove Theorem \ref{mtheo:decay}.
We consider the solution $\Si$ of Theorem \ref{mtheo:boundedness}.
We define a function $f : \mathbb{R}_0^+ \times S^2 \rightarrow \mathbb{R}$, supported on $\ell \geq 2$ as follows. Define for $ u \geq 0$ the limit $r \xlin_{\Si}^\infty (u, \theta,\phi)=\lim_{v \rightarrow u} r \xlin_{\Si}(u,v,\theta,\phi)$, which is the (smooth) restriction to $\mathcal{I}$ of the weighted tensor $r \xlin_{\Si}$. Similarly, define for $v \geq 0$, the limit $r \xblin_{\Si}^\infty (v, \theta,\phi)=\lim_{u \rightarrow v} r \xblin_{\Si}(u,v,\theta,\phi)$, which is the (smooth) restriction to $\mathcal{I}$ of the weighted tensor $r \xblin_{\Si}$. By the boundary condition we have $r \xlin_{\Si}^\infty (t, \theta,\phi) = r \xblin_{\Si}^\infty (t, \theta,\phi)$ for $t \geq 0$. We finally define a function $f$ by solving for each $t$ the elliptic (since $\ell \geq 2$) scalar equation\footnote{One computes $r^2 \slashed{div} \slashed{div}  r^2 \slashed{\mathcal{D}}_2^\star \slashed{\nabla} = r^4 \slashed{div} \left(-\frac{1}{2} \slashed{\Delta} - \frac{1}{2} K\right)\slashed{\nabla}=r^4 \slashed{div} \left(\frac{1}{2} \left(\slashed{\mathcal{D}}_1^\star \slashed{\mathcal{D}}_1 -K\right)- \frac{1}{2} K\right)\slashed{\nabla}= r^4 \left(\frac{1}{2} \slashed{\Delta}^2 -K \slashed{\Delta}\right)$.}
\begin{align}\label{eq:infinityxlin}
r^2 \slashed{div} \slashed{div} \,  r \xlin_{\Si}^\infty (t, \theta,\phi) = -2k r^2 \slashed{div} \slashed{div}  r^2 \slashed{\mathcal{D}}_2^\star \slashed{\nabla} f (t,\theta,\phi).
\end{align} 
The function $f$ generates a pure gauge solution $\mathscr{G}_f$ according to Lemma \ref{lem:exactsol} and using 
the notation of that lemma we have 
\begin{align}\label{eq:gaugefinal}
r^2 \slashed{div} \slashed{div}  \, r \xlin_{\mathscr{G}_f} (u, \theta,\phi)  = - 2 \Omega r \slashed{div} \slashed{div}  r^2 \slashed{\mathcal{D}}_2^\star \slashed{\nabla} f_u  \ \ \ \ , \ \ \ r^2 \slashed{div} \slashed{div}  \, r \xblin_{\mathscr{G}_f} (v, \theta,\phi)  = -2 \Omega r \slashed{div} \slashed{div}  r^2 \slashed{\mathcal{D}}_2^\star \slashed{\nabla} f_v \, .
\end{align}
Note also that from (\ref{sum1}) and Propositions \ref{prop:outshearfinal2} and \ref{prop:inshearfinal2} holding for the solution $\Si$ of Theorem \ref{mtheo:boundedness}, we have for $n \geq 3$ the quantitative estimates
\begin{align} \label{fest}
\sum_{i=0}^{n+2} | [r\slashed{\nabla}]^i f_u| + \sum_{i=0}^{n+2} | [r\slashed{\nabla}]^i f_v|  \lesssim \overset{\circ}{\mathbb{E}}{}^{n} \ \ \ , \ \ \ \sum_{i=0}^{n} |r [r\slashed{\nabla}]^i (f_u-f_v)| \lesssim \overset{\circ}{\mathbb{E}}{}^{n} \, ,
\end{align}
together with the corresponding estimates with the $\frac{1}{(\log v)^2}$-factor on the right-hand side. 

Using~\eqref{eq:infinityxlin} and~\eqref{eq:gaugefinal}, one can prove that $\Sff = \Si + \mathscr{G}_f$ satisfies $r\xlin_{\Sff} = r\xblin_{\Sff} = 0$ on $\mathcal{I}$. To see this we show separately that $r^2 \slashed{curl} \slashed{div} r\xlin_{\Sff} = 0$ and $r^2\slashed{div} \slashed{div} r\xlin_{\Sff} = 0$ on $\mathcal{I}$, which implies the claim for $\xlin_{\Sff}$ by standard elliptic estimates. Indeed, this follows immediately by our choice of pure gauge solution for the $\slashed{div} \slashed{div}$ part. On the other hand, it is not hard to see that $r^2\slashed{curl} \slashed{div} \xlin$ is actually gauge invariant and in fact equal to zero on $\mathcal{I}$ (use the linearised Codazzi equation (\ref{ellipchi}), the decay of $\blin$, $\slashed{curl} \elin = \slin$ and the boundary condition~\eqref{bcl4} for $\slin$). The argument for $\xblin_{\Sff}$ is entirely analogous.

It now immediately follows that in the new gauge we can estimate $r^2 \Omega \xlin$ and $r^2\Omega^{-1} r^2 \xblin$ instead of $ \Omega \xlin$ and $\Omega^{-1} r^2 \xblin$ in Theorem \ref{mtheo:boundedness}. Indeed, we can now integrate \emph{backwards} in the $4$- and $3$-direction from the boundary $\mathcal{I}$ using that $\xlin r_{\Sff} = 0$ and $\xblin r_{\Sff} = 0$ hold on the boundary in the new gauge and using the estimates on $\alin$ and $\ablin$ from Theorem~\ref{theo:teukolsky} and Corollary~\ref{cor:teuonspheres} just as in the proof of Propositions \ref{prop:inshearfinal} and  \ref{prop:outshearfinal}.

Next, since $\Omega \slashed{\nabla}_3 (\xlin r)_{\Sff}=0$ on the boundary (from $T( \xlin r)_{\Sff} = 0$ holding in the new gauge), integrating again backwards from the boundary one infers estimates for $r^2 \Omega \slashed{\nabla}_3 (\xlin r)_{\Sff}$ and $r^2 \Omega \slashed{\nabla}_4 (\xblin r)_{\Sff}$.

Estimates for $r^3 \elin_{\Sff}$ and $r^3 \eblin_{\Sff}$ are obtained directly by~\eqref{chih3},~\eqref{chih3b}, which also imply estimates for $r^2 \Olin_{\Sff}$ (modulo $\ell=0,1$ modes). Codazzi then gives control on $r\otx_{\Sff}$ and $r^3\Omega^{-2}\otxb_{\Sff}$ (modulo $\ell=0,1$ modes). 

Inserting the above bounds, one can infer from the linearised Gauss equation~\eqref{lingauss} that the linearised Gauss curvature behaves like $\Klin_{\Sff} \sim r^{-3}$. 

One finally adds a $\mathscr{G}_q$ pure gauge solution of Lemma~\ref{lem:puregaugemetric} so that for $\Sf = \Sff + \mathscr{G}_q$, one has $\glinh_{\Sf} = 0$ on the initial sphere of the boundary and $r^{-1} \bmlin_{\Sf} = 0$ along the boundary. More specifically we define
\begin{align}
(q_1,q_2) = - \int_0^u  \left(\Delta_{S^2}^{-1} \slashed{div} \bmlin_{\Sff}, - \Delta_{S^2}^{-1} \slashed{curl} \bmlin_{\Sff} \right) \left(\bar{u}, \bar{u} \right) d\bar{u} + (\bar{q}_1, \bar{q}_2) \ \ \ \textrm{with} \ \ 2r^2 \slashed{\mathcal{D}}_2^\star \slashed{\mathcal{D}}_1^\star (\bar{q}_1, \bar{q}_2) = -\glinh_{\Sff} (0,0) \, . \nonumber
\end{align}
It follows from the vanishing of $\bmlin_{\Sf}$ and $\otxb_{\Sf}$, $\otx_{\Sf}$, $\xlin_{\Sf}$, $\xblin_{\Sf}$ on $\mathcal{I}$ and the transport equations~\eqref{stos} that $\glinh_{\Sf}= 0$ on the whole boundary $\mathcal{I}$, which in turn also implies by~\eqref{gaussfootnote} that $r^{-2} \glinto_{\Sf} = 0$ (modulo $\ell=1$ modes). From this, we can infer bounds in $\mathcal{M}_{int}$ on $\bmlin_{\Sf},\glinto_{\Sf},\glinh_{\Sf}$ by integrating their respective transport equations backwards from $\mathcal{I}$. Now, using the estimates on $\bmlin_{\Sf}$ and $\bmlin_{\Sff}$, $\glinto_{\Sf}$ and $\glinto_{\Sff}$, and $\glinh_{\Sf}$ and $\glinto_{\Sff}$, one deduces that the non-vanishing pure gauge coefficients $\bmlin_{\mathscr{G}_f},\glinto_{\mathscr{G}_f},\glinh_{\mathscr{G}_f}$ also satisfy boundedness and logarithmic decay statements. This finishes the proof of Theorem~\ref{mtheo:decay}.



\appendix
\section{Boundary regularity and boundary conditions}\label{app:proofpropad}
This section is dedicated to the proof of Proposition~\ref{prop:ad} and Proposition~\ref{prop:NullWeylBC}.
\subsection{Proof of Proposition~\ref{prop:ad}}
We first have the following lemma. Its proof is based on ideas of~\cite{Fri95} which we adapt and strengthen in our geometric set up.
\begin{lemma}\label{lem:Fri}
  Assume that $\boldsymbol{g}$ is a solution to the Einstein equations~\eqref{EVEL} and that $\boldsymbol{\widetilde{g}}:=(u-v)^{2}\boldsymbol{g}$ extends smoothly to $\mathcal{I}=\{u-v=0\}$. Let $\boldsymbol{\widetilde{W}}:=\mathrm{W}(\boldsymbol{\widetilde{g}})$ denote the Weyl tensor of $\boldsymbol{\widetilde{g}}$. Let $\boldsymbol{\widetilde{N}} = \frac{1}{2}(u-v)^{-1}\left(e_4-e_3\right)$ denote the outgoing unit normal (for the metric $\boldsymbol{\widetilde{g}}$) to the $\{u-v=cst\}$-hypersurfaces, define $\boldsymbol{h}$ to be the induced metric by $\boldsymbol{\widetilde{g}}$ on the $\{u-v=cst\}$-hypersurfaces and define the second fundamental forms $\boldsymbol{\Theta}(X,Y) = \boldsymbol{\widetilde{g}}(\boldsymbol{\widetilde{\nabla}}_X \boldsymbol{\widetilde{N}},Y)$ for all tangent vectors $X,Y$ to the $\{u-v=cst\}$-hypersurfaces. Then, 
  \begin{enumerate}
  \item\label{item:optimals} $(u-v)^{-2}\left(k^2(u-v)^2\boldsymbol{\Omega}^2 -1\right)$ extends smoothly to $\mathcal{I}$,
  \item\label{item:secondform} $(u-v)^{-1}\boldsymbol{\Theta}$ extends smoothly to $\mathcal{I}$ (in a $\boldsymbol{h}$-normalised frame),
  \item\label{item:Weyl} $(u-v)^{-1}\boldsymbol{\widetilde{W}}$ extends smoothly to $\mathcal{I}$ (in a $\boldsymbol{\widetilde{g}}$-normalised frame).
  \end{enumerate}
\end{lemma}
\begin{proof}
  First note that by the double null form of $\boldsymbol{\widetilde{g}}$, we have the relations
  \begin{align}\label{eq:Nnabuv}
    \big|\boldsymbol{\widetilde{\nabla}}(u-v)\big|^2_{\boldsymbol{\widetilde{g}}} & = \frac{1}{(u-v)^2\boldsymbol{\Omega}^2}, & \boldsymbol{\widetilde{N}}(u-v) & = -\frac{1}{(u-v)\boldsymbol{\Omega}}, & \boldsymbol{\widetilde{N}} & = -(u-v)\boldsymbol{\Omega} \boldsymbol{\widetilde{\nabla}}(u-v).
  \end{align}
  The general conformal transformation formula\footnote{See \url{https://en.wikipedia.org/wiki/List_of_formulas_in_Riemannian_geometry}.} reads
  \begin{align*}
    (u-v)^2\mathrm{Ric}(\boldsymbol{\widetilde{g}}) & = (u-v)^2\mathrm{Ric}(\boldsymbol{g}) -2 (u-v)\boldsymbol{\widetilde{\nabla}}^2(u-v) - \left((u-v)\boldsymbol{\widetilde{\nabla}}^\mu\boldsymbol{\widetilde{\nabla}}_\mu(u-v) - 3|\boldsymbol{\widetilde{\nabla}}(u-v)|^2_{\boldsymbol{\widetilde{g}}}\right) \boldsymbol{\widetilde{g}}, 
  \end{align*}
  which, plugging in the Einstein equation~\eqref{EVEL} and using~\eqref{eq:Nnabuv}, rewrites as
  \begin{align}\label{eq:Ricconf}
    (u-v)^2\mathrm{Ric}(\boldsymbol{\widetilde{g}}) & = -2 (u-v)\boldsymbol{\widetilde{\nabla}}^2(u-v) - (u-v)\boldsymbol{\widetilde{\nabla}}^\mu\boldsymbol{\widetilde{\nabla}}_\mu(u-v) \boldsymbol{\widetilde{g}} - \mathfrak{s} \boldsymbol{\widetilde{g}}, 
  \end{align}
  with 
  \begin{align*}
    \mathfrak{s} & := 3\boldsymbol{\Omega}^{-2}(u-v)^{-2}\left(k^2(u-v)^2\boldsymbol{\Omega}^2 - 1\right).
  \end{align*}
  From~\eqref{eq:Ricconf}, using that $\boldsymbol{\widetilde{g}}$ extends smoothly to $\mathcal{I}$, we already deduce that $(u-v)^{-1}\mathfrak{s}$ extends smoothly to $\mathcal{I}$. We now want to obtain the better rate for $\mathfrak{s}$ claimed in Item~\ref{item:optimals}. We first note that, from Taylor's formula, relations~\eqref{eq:Nnabuv}, and the fact that $(u-v)^{-1}\mathfrak{s}$ extends smoothly to $\mathcal{I}$, the function $\widetilde{\mathfrak{s}}$ defined by 
  \begin{align}\label{eq:Taylors}
    \begin{aligned}
       \widetilde{\mathfrak{s}} & := (u-v)^{-2}\left(\mathfrak{s} + (u-v)\left((u-v)\boldsymbol{\Omega}\right)\boldsymbol{\widetilde{N}}\mathfrak{s}\right) \\
      & \,= (u-v)^{-2}\mathfrak{s} - 3\boldsymbol{\Omega}\boldsymbol{\widetilde{N}}\left((u-v)^{-2}\boldsymbol{\Omega}^{-2}\right) \\
      & \,= (u-v)^{-2}\mathfrak{s} + 6(u-v)^{-2}\boldsymbol{\Omega}^{-1}\boldsymbol{\widetilde{N}}\left(\log((u-v)\boldsymbol{\Omega})\right),
    \end{aligned}
  \end{align}
  extends smoothly to $\mathcal{I}$. Thus, if we can prove that $(u-v)^{-1}\boldsymbol{\widetilde{N}}\left(\log((u-v)\boldsymbol{\Omega})\right)$ extends smoothly to $\mathcal{I}$ then Item~\ref{item:optimals} follows from~\eqref{eq:Taylors}. To control $\boldsymbol{\widetilde{N}}\left(\log((u-v)\boldsymbol{\Omega})\right)$, we take the trace in~\eqref{eq:Ricconf} and use~\eqref{eq:Nnabuv}, and we have
  \begin{align}\label{eq:Rconf}
    \begin{aligned}
      (u-v)\mathrm{R}(\boldsymbol{\widetilde{g}}) & = -6 \boldsymbol{\widetilde{\nabla}}^\mu\boldsymbol{\widetilde{\nabla}}_\mu(u-v) - 4(u-v)^{-1}\mathfrak{s} \\
      & = -6 \left(\boldsymbol{\widetilde{N}}(\boldsymbol{\widetilde{N}}(u-v)) + \boldsymbol{\widetilde{N}}(u-v)\mathrm{tr}\boldsymbol{\Theta}\right) - 4(u-v)^{-1}\mathfrak{s} \\ 
      & = -6 \boldsymbol{\Omega}^{-1}(u-v)^{-1} \left(-\boldsymbol{\widetilde{N}}(\log((u-v)\boldsymbol{\Omega})) +\mathrm{tr}\boldsymbol{\Theta}\right) - 4(u-v)^{-1}\mathfrak{s},
    \end{aligned}
  \end{align}
  where $\mathrm{tr}\boldsymbol{\Theta} := \boldsymbol{h}^{ij}\boldsymbol{\Theta}_{ij}$ with $\boldsymbol{h}$ the induced metric by $\boldsymbol{\widetilde{g}}$ on the $\{u-v=cst\}$ hypersurfaces. Now, we want to express $\mathrm{tr}\boldsymbol{\Theta}$ -- at first order -- in terms of $\boldsymbol{\widetilde{N}}\left(\log((u-v)\boldsymbol{\Omega})\right)$. Letting $\boldsymbol{\widetilde{T}} := \frac{1}{2}(u-v)^{-1}(e_4+e_3)$, we have
  \begin{align*}
    \boldsymbol{\Theta}(\boldsymbol{\widetilde{T}},\boldsymbol{\widetilde{T}}) & = \boldsymbol{\widetilde{g}}\left(\left[\boldsymbol{\widetilde{T}},\boldsymbol{\widetilde{N}}\right],\boldsymbol{\widetilde{T}}\right) = -\boldsymbol{\widetilde{N}}\log((u-v)\boldsymbol{\Omega}). 
  \end{align*}
  where we used that $\boldsymbol{\widetilde{T}}= \frac{1}{2} \boldsymbol{\Omega}^{-1} (u-v)^{-1} \left(\partial_u+\partial_v + b^A\partial_A\right)$ and that $\boldsymbol{\widetilde{g}}(\boldsymbol{\widetilde{T}},\boldsymbol{\widetilde{N}})=\boldsymbol{\widetilde{g}}(\boldsymbol{\widetilde{T}},\partial_A)=0$. Hence
  \begin{align}\label{eq:trThTT}
    \begin{aligned}
      \mathrm{tr}\boldsymbol{\Theta} & = 3\boldsymbol{\Theta}(\boldsymbol{\widetilde{T}},\boldsymbol{\widetilde{T}}) + 3\boldsymbol{\widetilde{N}}\log((u-v)\boldsymbol{\Omega}) + \mathrm{tr}\boldsymbol{\Theta} = 3\boldsymbol{\widetilde{N}}\log((u-v)\boldsymbol{\Omega}) + 3\boldsymbol{\hat\Theta}(\boldsymbol{\widetilde{T}},\boldsymbol{\widetilde{T}}),
    \end{aligned}
  \end{align}
  with $\boldsymbol{\hat\Theta} := \boldsymbol{\Theta} - \frac{1}{3}\mathrm{tr}\boldsymbol{\Theta}\boldsymbol{h}$ denoting the traceless part of $\boldsymbol{\Theta}$. Plugging~\eqref{eq:Taylors} and~\eqref{eq:trThTT} into~\eqref{eq:Rconf}, we get
  \begin{align}\label{eq:NlogOmuvproof}
    \begin{aligned}
    (u-v)\mathrm{R}(\boldsymbol{\widetilde{g}}) & = -6 \boldsymbol{\Omega}^{-1}(u-v)^{-1} (-1+3+4)\boldsymbol{\widetilde{N}}(\log(u-v)\boldsymbol{\Omega}) \\
    & \quad - 18\boldsymbol{\Omega}^{-1}(u-v)^{-1}\boldsymbol{\hat\Theta}(\boldsymbol{\widetilde{T}},\boldsymbol{\widetilde{T}}) - 4 (u-v)\widetilde{\mathfrak{s}}.
    \end{aligned}
  \end{align}
  Let us now show that $\boldsymbol{\hat\Theta}$ vanishes at first order at $\mathcal{I}$. Projecting formula~\eqref{eq:Ricconf} on the $\{u-v=cst\}$-hypersurfaces, using~\eqref{eq:Nnabuv}, and taking the traceless part, we have 
  \begin{align*}
    \boldsymbol{\Omega}(u-v)\left(\mathrm{Ric}(\boldsymbol{\widetilde{g}})_{ij}-\boldsymbol{h}^{i'j'}\mathrm{Ric}(\boldsymbol{\widetilde{g}})_{i'j'}\boldsymbol{h}_{ij}\right) & = -2(u-v)^{-1}\boldsymbol{\hat\Theta}_{ij},
  \end{align*}
   Hence, using that $\boldsymbol{\widetilde{g}}$ extends smoothly to $\mathcal{I}$, we have that $(u-v)^{-1}\boldsymbol{\hat\Theta}$ extends smoothly to $\mathcal{I}$. Thus, from~\eqref{eq:NlogOmuvproof}, using that $\boldsymbol{\widetilde{g}}$, $(u-v)^{-1}\boldsymbol{\hat\Theta}$ and $\widetilde{\mathfrak{s}}$ extend smoothly at $\mathcal{I}$, we infer that $(u-v)^{-1}\boldsymbol{\widetilde{N}}(\log(u-v)\boldsymbol{\Omega})$ extends smoothly to $\mathcal{I}$. Hence, recalling formula~\eqref{eq:Taylors} and the definition of $\mathfrak{s}$, Item~\ref{item:optimals} is proved. Using~\eqref{eq:trThTT} and the regularity of $\boldsymbol{\hat\Theta}$ obtained above, we also directly infer Item~\ref{item:secondform}.\\
  From the conformal invariance of the Bianchi equations for the Weyl tensor, we have
  \begin{align}\label{eq:confBianchi}
  \boldsymbol{\widetilde{\nabla}}_\alpha\boldsymbol{\widetilde{d}}^\alpha_{\beta\gamma\delta} & = 0, & \text{with} \quad \boldsymbol{\widetilde{d}} & := (u-v)^{-1} \boldsymbol{\widetilde{W}}.
  \end{align}
  which, using~\eqref{eq:Nnabuv}, implies
  \begin{align}\label{eq:devconfBianchi}
    (u-v)\boldsymbol{\widetilde{\nabla}}^{\alpha}\boldsymbol{\widetilde{W}}_{\alpha\beta\gamma\delta} - (u-v)^{-1}\boldsymbol{\Omega}^{-1}\boldsymbol{\widetilde{N}}^\alpha\boldsymbol{\widetilde{W}}_{\alpha\beta\gamma\delta} = 0.
  \end{align}
  Using that $\boldsymbol{\widetilde{g}}$ extends smoothly to $\mathcal{I}$, $(u-v)^{-1}\boldsymbol{\widetilde{N}}^\alpha\boldsymbol{\widetilde{W}}_{\alpha\beta\gamma\delta}$ extends smoothly to $\mathcal{I}$. Using the symmetries of the Weyl tensor -- see \emph{e.g.} formulas (7.3.3) in \cite{Chr.Kla93} --, all the components of $\boldsymbol{\widetilde{W}}$ can be obtained by linear combinations of $\boldsymbol{\widetilde{N}}^\alpha\boldsymbol{\widetilde{W}}_{\alpha\beta\gamma\delta}$ and Item~\ref{item:Weyl} follows.
\end{proof}
We can now prove Proposition~\ref{prop:ad}.
\begin{proof}[Proof of Proposition~\ref{prop:ad}]
  The regularity of $\boldsymbol{\Omega}$ is a direct consequence of Lemma~\ref{lem:Fri}. The regularity of $\boldsymbol{b}$ follows from the coordinate components $\boldsymbol{b}^A$ (indices up!) extending regularly. By the conformal transformation formulas, we have
  \begin{subequations}\label{eq:conftrans}
    \begin{align}\label{eq:confchi}
      \begin{aligned}
        \boldsymbol{\chi}_{AB} = \boldsymbol{g}\left(\boldsymbol{\nabla}_{\partial_A}e_4,\partial_B\right) & = (u-v)^{-2}\boldsymbol{\widetilde{g}}\left(\boldsymbol{\widetilde{\nabla}}_{\partial_A}e_4,\partial_B\right) - e_4(\log(u-v)) \boldsymbol{g}_{AB} \\
        & = (u-v)^{-1}\boldsymbol{\widetilde{g}}\left(\boldsymbol{\widetilde{\nabla}}_{\partial_A}(\boldsymbol{\widetilde{T}}+\boldsymbol{\widetilde{N}}),\partial_B\right) + \frac{\boldsymbol{g}_{AB}}{(u-v)\boldsymbol{\Omega}},
      \end{aligned}
    \end{align}
    and similarly
    \begin{align}
      \boldsymbol{\underline{\chi}}_{AB} & = (u-v)^{-1}\boldsymbol{\widetilde{g}}\left(\boldsymbol{\widetilde{\nabla}}_{\partial_A}(\boldsymbol{\widetilde{T}}-\boldsymbol{\widetilde{N}}),\partial_B\right) - \frac{\boldsymbol{g}_{AB}}{(u-v)\boldsymbol{\Omega}}, \label{eq:confchib} \\
      \boldsymbol{\eta}_A & = \boldsymbol{\widetilde{g}}\left(\boldsymbol{\widetilde{\nabla}}_{\partial_A}\boldsymbol{\widetilde{N}}, \boldsymbol{\widetilde{T}}\right) + \boldsymbol{\slashed{\nabla}}_A\log\boldsymbol{\Omega}, \label{eq:confeta} \\
      \boldsymbol{\underline{\eta}}_A & = - \boldsymbol{\widetilde{g}}\left(\boldsymbol{\widetilde{\nabla}}_{\partial_A}\boldsymbol{\widetilde{N}}, \boldsymbol{\widetilde{T}}\right) + \boldsymbol{\slashed{\nabla}}_A\log\boldsymbol{\Omega}, \label{eq:confetab}
    \end{align}
  \end{subequations}
  From~\eqref{eq:confchi}, and the fact that $\boldsymbol{\widetilde{g}}$ extends regularly at $\mathcal{I}$, one infers that $\boldsymbol{\Omega}\boldsymbol{\chi} - \frac{\boldsymbol{g}}{(u-v)\boldsymbol{\Omega}}$ extends regularly to $\mathcal{I}$ in a $\boldsymbol{g}$ orthonormal frame. The corresponding regularity for $\boldsymbol{\underline{\chi}}$ follows similarly. Moreover, from formulas~\eqref{eq:confchi},~\eqref{eq:confchib}, one has
  \begin{align*}
    \boldsymbol{\chi}_{AB}-\boldsymbol{\underline{\chi}}_{AB} -2 \frac{\boldsymbol{g}_{AB}}{(u-v)\boldsymbol{\Omega}} & = 2(u-v)^{-1}\boldsymbol{\Theta}_{AB}, 
  \end{align*}
  and from the (better) regularity for $\boldsymbol{\Theta}$ of Item~\ref{item:secondform} of Lemma~\ref{lem:Fri}, we obtain the (better) regularity for the difference $\boldsymbol{\chi}-\boldsymbol{\underline{\chi}}$ in~\eqref{boundsextend}. From~\eqref{eq:confeta},~\eqref{eq:confetab}, we have that $r(\boldsymbol{\eta} - \underline{\boldsymbol{\eta}})_A = 2r\boldsymbol{\Theta}(\boldsymbol{\widetilde{T}},\partial_A)$ is regular by Item~\ref{item:secondform} of Lemma~\ref{lem:Fri}, hence $r^2(\boldsymbol{\eta} - \underline{\boldsymbol{\eta}})$ extends smoothly to $\mathcal{I}$ in a $\boldsymbol{g}$ orthonormal frame. Moreover, we have $r^2(\boldsymbol{\eta} + \underline{\boldsymbol{\eta}})_A = r^2\boldsymbol{\Omega}^{-1}\partial_A\boldsymbol{\Omega}$ is regular by (the good regularity of) Item~\ref{item:optimals} of Lemma~\ref{lem:Fri}, thus $r^3(\boldsymbol{\eta} + \underline{\boldsymbol{\eta}})$ extends smoothly to $\mathcal{I}$ in a $\boldsymbol{g}$ orthonormal frame, and combining the above, $r^2\boldsymbol{\eta},r^2\underline{\boldsymbol{\eta}}$ extend smoothly to $\mathcal{I}$ in a $\boldsymbol{g}$ orthonormal frame. We have
  \begin{align*}
    \boldsymbol{\omega} & = \frac{1}{2}\partial_v\log\left((u-v)^2\boldsymbol{\Omega}^2\right) + \frac{1}{u-v}, & \boldsymbol{\underline{\omega}} & = \frac{1}{2}\partial_u\log\left((u-v)^2\boldsymbol{\Omega}^2\right) - \frac{1}{u-v},
  \end{align*}
  from which, by Item~\ref{item:optimals} of Lemma~\ref{lem:Fri}, we infer that $r(\boldsymbol{\omega}- \boldsymbol{\underline{\omega}})$ and $r^2(\boldsymbol{\omega}+\boldsymbol{\underline{\omega}})$ extend regularly at $\mathcal{I}$. The regularity of the null curvature components is a direct consequence of the last item of Lemma~\ref{lem:Fri}, using the conformal invariance of the Weyl tensor. This finishes the proof of the corollary.
\end{proof}

\subsection{Proof of Proposition~\ref{prop:NullWeylBC}}
We first have the following lemma. See also~\cite{Fri95} and Section 6.2 in~\cite{Hol.Luk.Smu.War20}.
\begin{lemma}\label{lem:WeylBC}
  Assume that $\boldsymbol{\widetilde{g}}$ extends smoothly to $\mathcal{I}$, satisfies the Einstein equations~\eqref{EVEL}, and that the induced metric $\boldsymbol{h}$ by $\boldsymbol{\widetilde{g}}$ on $\mathcal{I}$ is conformal to the Anti-de Sitter metric at infinity $-\d t^2 + k^{-2}\boldsymbol{\gamma}$. Then, the $\{u-v\}$-tangent tensor $(u-v)^{-2}\boldsymbol{\widetilde{N}}^\alpha\boldsymbol{\widetilde{W}}_{\alpha ijk}$ extends smoothly to $\mathcal{I}$. Note that this is equivalent to the spacetime tensor $(u-v)^{-2}\boldsymbol{\widetilde{N}}^\alpha\boldsymbol{\widetilde{N}}^\gamma {^\star}\boldsymbol{\widetilde{W}}_{\alpha\mu\gamma\nu}$, with ${^\star}$ denoting the Hodge dual, extending smoothly to $\mathcal{I}$.
\end{lemma}
\begin{proof}
  From contractions of the second Bianchi identities and the definition of the Weyl tensor, we have the following general formula
  \begin{align}\label{eq:DivWeylCotton}
    \boldsymbol{\widetilde{\nabla}}^{\alpha}\mathrm{W}(\boldsymbol{\widetilde{g}})_{\alpha \beta \gamma \delta} & = \frac{1}{2}\boldsymbol{\widetilde{\nabla}}_\gamma \left(\mathrm{Ric}(\boldsymbol{\widetilde{g}})_{\beta\delta} - \frac{1}{6} \mathrm{R}(\boldsymbol{\widetilde{g}})\boldsymbol{\widetilde{g}}_{\beta\delta}\right) -  \frac{1}{2}\boldsymbol{\widetilde{\nabla}}_\delta \left(\mathrm{Ric}(\boldsymbol{\widetilde{g}})_{\beta\gamma} - \frac{1}{6} \mathrm{R}(\boldsymbol{\widetilde{g}})\boldsymbol{\widetilde{g}}_{\beta\gamma}\right) =: \frac{1}{2} \mathrm{C}(\boldsymbol{\widetilde{g}})_{\beta\delta\gamma},  
  \end{align}
  where $\mathrm{C}(\boldsymbol{\widetilde{g}})$ is called the \emph{Cotton tensor} of $\boldsymbol{\widetilde{g}}$. The Gauss-Codazzi equations on the boundary $\mathcal{I}$ read
  \begin{align}
    \mathrm{Ric}(\boldsymbol{\widetilde{g}})_{ij} - \mathrm{Rm}(\boldsymbol{\widetilde{g}})_{i\boldsymbol{\widetilde{N}} j\boldsymbol{\widetilde{N}}} & = \mathrm{Ric}(\boldsymbol{h})_{ij} - \boldsymbol{\Theta}_{ik}\boldsymbol{\Theta}^{k}_j + \mathrm{tr}\boldsymbol{\Theta} \boldsymbol{\Theta}_{ij}, \label{eq:GaussCodazziII}\\
    \mathrm{R}(\boldsymbol{\widetilde{g}}) - 2\mathrm{Ric}(\boldsymbol{\widetilde{g}})_{\boldsymbol{\widetilde{N}}\boldsymbol{\widetilde{N}}} & = \mathrm{R}(\boldsymbol{h}) - |\boldsymbol{\Theta}|^2 + (\mathrm{tr}\boldsymbol{\Theta})^2. \label{eq:GaussCodazziIItrace}
  \end{align}
  By the definition of the Weyl tensor, we have
  \begin{align*}
    \mathrm{W}(\boldsymbol{\widetilde{g}})_{i\boldsymbol{\widetilde{N}} j\boldsymbol{\widetilde{N}}} & = \mathrm{Rm}(\boldsymbol{\widetilde{g}})_{i\boldsymbol{\widetilde{N}} j\boldsymbol{\widetilde{N}}} - \frac{1}{2} \mathrm{Ric}(\boldsymbol{\widetilde{g}})_{ij} - \frac{1}{2} \mathrm{Ric}(\boldsymbol{\widetilde{g}})_{\boldsymbol{\widetilde{N}}\boldsymbol{\widetilde{N}}} \boldsymbol{\widetilde{g}}_{ij} + \frac{1}{6}\mathrm{R}(\boldsymbol{\widetilde{g}})\boldsymbol{\widetilde{g}}_{ij} =: \mathrm{w}(\boldsymbol{\widetilde{g}})_{ij}, 
  \end{align*}
  which, plugged in the Gauss-Codazzi equation~\eqref{eq:GaussCodazziII}, using~\eqref{eq:GaussCodazziIItrace} to replace $\mathrm{Ric}(\boldsymbol{\widetilde{g}})_{\boldsymbol{\widetilde{N}}\boldsymbol{\widetilde{N}}}$, gives
  \begin{align}\label{eq:GaussCodazziIIbis}
    \begin{aligned}
      \frac{1}{2}\left(\mathrm{Ric}(\boldsymbol{\widetilde{g}})_{ij} - \frac{1}{6} \mathrm{R}(\boldsymbol{\widetilde{g}})\boldsymbol{\widetilde{g}}_{ij}\right) - \mathrm{w}(\boldsymbol{\widetilde{g}})_{ij} & = \mathrm{Ric}(\boldsymbol{h})_{ij} - \frac{1}{4}\mathrm{R}(\boldsymbol{h})\boldsymbol{h}_{ij} + \theta_{ij},
    \end{aligned}
  \end{align}
  with $\theta_{ij} := - \boldsymbol{\Theta}_{ik}\boldsymbol{\Theta}^{k}_j + \mathrm{tr}\boldsymbol{\Theta} \boldsymbol{\Theta}_{ij} + \frac{1}{4} \left(|\boldsymbol{\Theta}|^2 - (\mathrm{tr}\boldsymbol{\Theta})^2\right)\boldsymbol{h}_{ij}$. Defining the Cotton tensor of $\boldsymbol{h}$ by
  \begin{align*}
    \mathrm{C}(\boldsymbol{h})_{ijk} := \boldsymbol{\overline{\nabla}}_k\left(\mathrm{Ric}(\boldsymbol{h})_{ij} - \frac{1}{4}\mathrm{R}(\boldsymbol{h})\boldsymbol{h}_{ij}\right) - \boldsymbol{\overline{\nabla}}_j\left(\mathrm{Ric}(\boldsymbol{h})_{ik} - \frac{1}{4}\mathrm{R}(\boldsymbol{h}) \boldsymbol{h}_{ik}\right),
  \end{align*}
  where $\boldsymbol{\overline{\nabla}}$ is the covariant derivative of $\boldsymbol{h}$, and applying  $\boldsymbol{\overline{\nabla}}$ to~\eqref{eq:GaussCodazziIIbis}, we get
  \begin{align}\label{eq:CottonCotton}
    \begin{aligned}
      \frac{1}{2} \mathrm{C}(\boldsymbol{\widetilde{g}})_{ijk} & = \mathrm{C}(\boldsymbol{h})_{ijk} + \boldsymbol{\overline{\nabla}}_k\theta_{ij} - \boldsymbol{\overline{\nabla}}_j\theta_{ik} + \boldsymbol{\overline{\nabla}}_k\mathrm{w}(\boldsymbol{\widetilde{g}})_{ij} - \boldsymbol{\overline{\nabla}}_k\mathrm{w}(\boldsymbol{\widetilde{g}})_{ij}.
    \end{aligned}
  \end{align}
  From~\eqref{eq:CottonCotton} and Lemma~\ref{lem:Fri}, we deduce that $(u-v)^{-1}\left(\frac{1}{2} \mathrm{C}(\boldsymbol{\widetilde{g}})_{ijk} - \mathrm{C}(\boldsymbol{h})_{ijk}\right)$ extends regularly to $\mathcal{I}$. Combining~\eqref{eq:devconfBianchi} and~\eqref{eq:DivWeylCotton}, we thus deduce that
  \begin{align}\label{eq:WeylCottonBoundary}
    (u-v)^{-1}\mathrm{W}(\boldsymbol{\widetilde{g}})_{\boldsymbol{\widetilde{N}} ijk} - \mathrm{C}(\boldsymbol{h})_{ikj} & = (u-v) E_{ijk},
  \end{align}
  with $E$ smoothly extending to $\mathcal{I}$. Now, the Cotton tensor of a 3-dimensional metric is invariant under a conformal transformation and it is easy to see from its definition that it vanishes for Lorentzian cylinders $-\mathrm{d} t^2 + k^{-2}\gamma$. Thus, if $\boldsymbol{h}$ is conformal to such a metric, one has by~\eqref{eq:WeylCottonBoundary} that $(u-v)^{-2}\mathrm{W}(\boldsymbol{\widetilde{g}})_{\boldsymbol{\widetilde{N}} ijk}$ extends smoothly at $\mathcal{I}$, and the conclusion of the lemma follows.
\end{proof}

We can now prove Proposition~\ref{prop:NullWeylBC}.
\begin{proof}[Proof of Proposition~\ref{prop:NullWeylBC}]
  Using the conformal invariance of the Weyl tensor, one has
  \begin{align*}
    (r^\star)^{-1} \left(\boldsymbol\alpha_{AB} -\underline{\boldsymbol\alpha}_{AB}\right) & = 2 (u-v)^{-2}\boldsymbol{\widetilde{W}}(\boldsymbol{\widetilde{N}},\partial_A,\boldsymbol{\widetilde{T}},\partial_B) + 2(u-v)^{-2}\boldsymbol{\widetilde{W}}(\boldsymbol{\widetilde{N}},\partial_B,\boldsymbol{\widetilde{T}},\partial_A), \\
    (r^\star)^{-2} \left( \boldsymbol\beta_A +\underline{\boldsymbol\beta}_A \right) & = 2(u-v)^{-2}\boldsymbol{\widetilde{W}}(\boldsymbol{\widetilde{N}},\boldsymbol{\widetilde{T}},\boldsymbol{\widetilde{T}},\partial_A), \\
   (r^\star)^{-3} \boldsymbol\sigma & = (u-v)^{-2}{^\star}\boldsymbol{\widetilde{W}}(\boldsymbol{\widetilde{N}},\boldsymbol{\widetilde{T}},\boldsymbol{\widetilde{N}},\boldsymbol{\widetilde{T}}).
  \end{align*}
  From the above formulas and the result of Lemma~\ref{lem:WeylBC} one directly deduces~\eqref{bc1},~\eqref{bc2},~\eqref{bc3}. From the Bianchi equations \eqref{Bianchi1} and \eqref{Bianchi10}, the boundary condition \eqref{bc1} and the fact that $\lim_{v \rightarrow u} r \hat{\boldsymbol{\chi}}= \lim_{v \rightarrow u} r \underline{\hat{\boldsymbol{\chi}}}$ holds by Proposition \ref{prop:ad}, one further infers~\eqref{bc4} and this finishes the proof of the proposition.
\end{proof}


\section{Computation of the $\ell=0$ mode} \label{sec:l0}
From the linear version of the Birkhoff theorem, we already know that the space of solutions supported on $\ell=0$ can only consist of the (linearised) Schwarzschild solution and pure gauge solutions. It turns out we can parametrise the space of solutions more or less explicitly. {\bf In this section all quantities are supported on $\ell=0$ so we simply write $\otx$ for $\otx_{\ell=0}$ etc.~to keep the notation clean. }

We first define two quantities (supported on $\ell=0$ by the above convention):
\[
\Pmcalin :=  r^3 \rlin - 3M \frac{\glinto}{\sqrt{\slashed{g}}}  \ \ \ , \ \ \ \Qlin:= \frac{r \otx}{\Omega^2}  -4\Olin +\frac{\glinto}{\sqrt{\slashed{g}}} \, .
\]
The importance of these quantities lies (partly) in their simple propagation equations (following from (\ref{Bianchi4}), (\ref{Bianchi5}), (\ref{stos}), (\ref{uray}) and (\ref{oml3}))
\begin{align} \label{prol0}
\Omega \slashed{\nabla}_3 \Pmcalin = 0  \ , \ \Omega \slashed{\nabla}_4 \Pmcalin = 0 \ \ \ \textrm{and} \ \ \ \Omega \slashed{\nabla}_4 \Qlin = 0  \, .
\end{align}
Using the formula (\ref{gaussfootnote}) we write the linearised Gauss equation (\ref{lingauss}) for $\ell=0$ as 
\[
-\frac{1}{r^2} \frac{\glinto}{\sqrt{\slashed{g}}} = - \frac{1}{r^3}\Pmcalin -\frac{3M}{r^3} \frac{\glinto}{\sqrt{\slashed{g}}}  - \frac{1}{2r} \Omega \slashed{\nabla}_3 \left( \frac{\glinto}{\sqrt{\slashed{g}}} \right)+ \frac{\Omega^2}{2r^2} \Qlin -\frac{\Omega^2}{2r^2} \frac{\glinto}{\sqrt{\slashed{g}}} \, ,
\]
or more concisely as
\begin{align} \label{ode3g}
\frac{\Omega^2}{r} \Omega \slashed{\nabla}_3  \left(\frac{r}{\Omega^2} \frac{\glinto}{\sqrt{\slashed{g}}} \right) =  -\frac{2}{r^2}\Pmcalin  + \frac{ \Omega^2}{r}\Qlin \, .
\end{align}
We first establish that if $\Pmcalin=0$ and $\Qlin=0$ hold on the initial data cone, the solution is necessarily trivial.
\begin{lemma} \label{lem:zeromodevanishing}
Let $\mathscr{S}$ be a smooth solution of the system of gravitational perturbations supported on $\ell=0$. If  
$\Pmcalin = 0$ and $\Qlin= 0$ hold on $\underline{C}_{v_0}$,
then the solution is necessarily equal to the zero solution. 
\end{lemma}
\begin{proof}
The quantities $\Pmcalin$ and $\Qlin$ are conserved in the $4$-direction by (\ref{prol0}), hence zero on $\mathcal{M}_{int}$. It follows that $\Pmcalin=\Qlin=0$ in (\ref{ode3g}), which since $\frac{r}{\Omega^2}\frac{\glinto}{\sqrt{\slashed{g}}}$ vanishes at $\mathcal{I}$ implies that $\frac{\glinto}{\sqrt{\slashed{g}}}=0$ on $\mathcal{M}_{int}$. It now immediately follows that $\Klin=0$, $\rlin = 0$, $\otx=0$, $\otxb=0$, $\Olin=0$, $\olin=0$, $\olinb=0$.
\end{proof}

We wish to study all radial solutions of (\ref{ode3g}) to exhaust the space of solutions for $\ell=0$. We first note that by adding a pure gauge solution, we can restrict to the case of both $\Pmcalin$ and $\Qlin$ being constant on $\mathcal{M}_{int}$.

\begin{lemma} \label{lem:l0m}
Let $\mathscr{S}$ be a solution of the system of gravitational perturbations supported on $\ell=0$. Then we can add a pure gauge solution $\mathscr{G}_f$ from Lemma \ref{lem:exactsol} such that the solution $\mathscr{S}^\prime = \mathscr{S} + \mathscr{G}_f$ satisfies 
\begin{align}
\otx \big|_{\mathscr{S}^\prime} (\infty,v_0) = 0 \ \ \textrm{and} \ \  \Qlin|_{\mathscr{S}^\prime} (u,v_0) \ \ \textrm{is constant in $u$.} \nonumber
\end{align}
As a consequence of (\ref{prol0}), $\Pmcalin$ and $\Qlin$ are both constant on $\mathcal{M}_{int}$ for $\mathscr{S}^\prime$.
\end{lemma}

\begin{proof}
Letting $\tilde{f}(u) = \frac{1}{\Omega^2(u,v_0)} \frac{1}{2} \frac{r_+^2}{\frac{2M}{r_+} + k^2 r_+^2}\otx_{\mathscr{S}}(\infty,v_0)$ generate $\tilde{f}_u$ and $\tilde{f}_v$ and a pure gauge solution $\mathscr{G}_{\tilde{f}}$ as in Lemma \ref{lem:exactsol} we achieve that $\mathscr{S} + \mathscr{G}_{\tilde{f}}$ satisfies $\otx \big|_{\mathscr{S} + \mathscr{G}_{\tilde{f}}} (\infty,v_0) = 0$. In particular, the quantity $\Qlin(\infty,v_0)$ is now regular at the horizon for the solution $\mathscr{S} + \mathscr{G}_{\tilde{f}}$. We next add a second pure gauge solution $\mathscr{G}_{\hat{f}}$ which does not affect $\otx(\infty,v_0)$ but achieves the second condition. For this we define $\hat{f}_u$ by the ODE
\[
\Qlin_{\mathscr{S} + \mathscr{G}_{\tilde{f}}}(\infty,v_0) = \Qlin_{\mathscr{S} + \mathscr{G}_{\tilde{f}}} (u,v_0)- 2\partial_u \hat{f}_u (u,v_0) \, \ \ , \ \ \hat{f}_u(u_0,v_0)=0 \, .
\]
One now checks that $\Qlin_{\mathscr{S}^\prime:=\mathscr{S} + \mathscr{G}_{\tilde{f}} + \mathscr{G}_{\hat{f}}}$
is indeed constant, $\Qlin_{\mathscr{S}^\prime} (u,v_0) = \Qlin_{\mathscr{S} + \mathscr{G}_{\tilde{f}}} (\infty,v_0)$, and that $\hat{f}_u$ is bounded. By Lemma \ref{lem:exactsol}, a bounded $\hat{f}_u$ will imply $\otx|_{\mathscr{S}^\prime}(\infty,v_0) = \otx \big|_{\mathscr{S} + \mathscr{G}_{\tilde{f}}} (\infty,v_0)= 0$.
\end{proof}

Let us denote the constants $d:=\Pmcalin$ and $c:=\Qlin$ and compute now the general regular radial solutions of (\ref{ode3g}). Setting  
$\frac{\glinto}{\sqrt{\slashed{g}}} = f(r)$, $f$ satisfies the ODE 
\begin{align} \label{rado}
\partial_r \left(f \frac{r}{\Omega^2}  \right) =+\frac{2d \cdot M}{r \Omega^4}   - \frac{c}{\Omega^2} \, , 
\end{align} 
which we can write as (setting $\frac{1}{l^2}=k^2$)
\[
\partial_r \left(f \frac{r}{\Omega^2}  + \frac{2dM}{1+3\frac{r^2}{l^2}}\frac{1}{\Omega^2} \right) = 2dM \frac{l^2(l^2-3r^2)}{(l^2+3r^2)^2}\frac{1}{\Omega^2 r}-\frac{c}{\Omega^2} \, .
\]
To make the solution regular at the horizon we require
\[
c=2dM \frac{l^2(l^2-3r_+^2)}{r_+(l^2+3r_+^2)^2} \, .
\]
Note that with this the right hand side is integrable near infinity and near the horizon. We finally obtain\footnote{Note that at this point we can no longer take the limit $l \rightarrow \infty$ to compare with the asymptotically flat case, since we have used that $\frac{r}{\Omega^2}$ goes to zero, which it does not in the asymptotically flat case. However, in (\ref{rado}) we can still take the limit $\ell \rightarrow \infty$ and check that in this case $c=d$ and $f=-d$ is indeed a solution, as was obtained in \cite{Daf.Hol.Rod19}.}
\begin{align}
f (r) = -\frac{2dM}{r \left(1+\frac{3r^2}{l^2}\right)} +\frac{\Omega^2}{r} \int_r^\infty \frac{2dM}{\Omega^2} \left(\frac{l^2(l^2-3r_+^2)}{(l^2+3r_+^2)^2} \frac{1}{r_+} -  \frac{l^2(l^2-3r^2)}{(l^2+3r^2)^2}\frac{1}{ r}\right) dr \, ,
\end{align}
which satisfies
\[
f(r_+) = -\frac{2d M}{r_+\left(1+\frac{3r_+^2}{l^2}\right)} \ \ \ \textrm{and} \ \ \ \ f(\infty) = 2dM \frac{l^2(l^2-3r_+^2)}{(l^2+3r_+^2)^2} \frac{1}{r_+} \, .
\]
In particular, $f$ is uniformly bounded and smooth on the exterior. All non-vanishing Ricci-coefficients and curvature components can easily be computed in terms of $f$. We find
\[
\otx = (\partial_r f) \Omega^2  \ \ \ , \ \ \ \otxb = -(\partial_r f) \Omega^2  \, ,
\]
and from the definition of $\Qlin$ the expression
\begin{align}
\Olin &= \frac{1}{4} \left( \partial_r (\Omega^2 \Omega^{-2} f \cdot r) - c \right)=
\frac{1}{4} \left( \frac{2dM}{r \Omega^2} - 4dM \frac{l^2(l^2-3r_+^2)}{r_+(l^2+3r_+^2)^2}\right) +\frac{1}{4} \left(\frac{2M}{r^2} + \frac{2r}{l^2} \right) \frac{f r}{\Omega^2} \, .
\end{align}
To check that $\Olin$ is finite at the horizon we compute
\[
\lim_{r \rightarrow r_+} \Olin(r)  = \frac{dM}{2\left(1+\frac{3r_+^2}{l^2}\right)} \, .
\]
Note also
\[
\lim_{r \rightarrow \infty} \Olin(r) =  - dM \frac{l^2(l^2-3r_+^2)}{r_+(l^2+3r_+^2)^2} + dM \frac{l^2(l^2-3r_+^2)}{r_+(l^2+3r_+^2)^2} = 0 \, , 
\]
so $\Olin$ vanishes at infinity. 
Finally, we obtain from the null structure equations 
\[
\olin = \Omega \slashed{\nabla}_4 \Olin \ \ , \ \ \olinb = \Omega \slashed{\nabla}_3 \Olin \ \ , \ \ \Klin = -\frac{1}{r^2} \frac{\glinto}{\sqrt{\slashed{g}}}  \ \ , \ \ \rlin = \frac{dM}{r^3} + \frac{3M}{r^3} \frac{\glinto}{\sqrt{\slashed{g}}} \, .
\]
This concludes our derivation of the solution appearing in Lemma \ref{lem:kerr} of the text. 

\bibliographystyle{graf_GR_alpha}
\bibliography{graf_GR}

\end{document}